%% file: paper.tex
\documentclass[11pt]{article}
\title{Undecidability of Tiling with a Tromino}%
\author{%
  MIT--ULB CompGeom Group\thanks{%
    Artificial first author to highlight that the other authors (in
    alphabetical order) worked as an equal group. Please include all
    authors (including this one) in your bibliography, and refer to the
    authors as “MIT--ULB CompGeom Group” (without “et al.”).}%
  \and
  Zachary Abel\thanks{%
    Department of Electrical Engineering and Computer Science,
    Massachusetts Institute of Technology, Cambridge, MA, USA,
    \protect\href{mailto:zabel@mit.edu}{\protect\nolinkurl{zabel@mit.edu}}}%
  \and
  Hugo Akitaya\thanks{%
    Miner School of Computer and Information Sciences,
    University of Massachusetts, Lowell, MA, USA,
    \protect\href{mailto:hugoakitaya@gmail.com}{\protect\nolinkurl{hugoakitaya@gmail.com}}}%
  \and
  Lily Chung\thanks{%
    Computer Science and Artificial Intelligence Laboratory,
    Massachusetts Institute of Technology, Cambridge, MA, USA,
    \protect\href{mailto:lkdc@mit.edu,edemaine@mit.edu,diomidova@mit.edu,della@mit.edu,jaysonl@mit.edu}{\protect\nolinkurl{{lkdc,edemaine,diomidova,della,jaysonl}@mit.edu}}}%
  \and
  Erik D. Demaine\footnotemark[4]
  \and
  Jenny Diomidova\footnotemark[4]
  \and
  Della Hendrickson\footnotemark[4]
  \and
  Stefan Langerman\thanks{%
    Computer Science Department, Universit\'e libre de Bruxelles, Belgium,
    \protect\href{mailto:stefan.langerman@ulb.ac.be}{\protect\nolinkurl{stefan.langerman@ulb.ac.be}}. S. Langerman is Directeur de Recherche du F.R.S.-FNRS.}%
  \and
  Jayson Lynch\footnotemark[4]}%

\date{}

  \usepackage{amsfonts}
  \usepackage{amsmath}
  \usepackage{amsthm}
  \usepackage{amssymb}
  \usepackage{booktabs}
  \usepackage[style=english]{csquotes}
  \usepackage{enumitem}
  \usepackage[margin=1in]{geometry}
  \usepackage{graphicx}
  \usepackage{needspace}
  \usepackage{subcaption}
  \usepackage{tabularx}
  \usepackage[table]{xcolor}
  \usepackage{tikz}
  \usepackage{url}
  \usepackage{xspace}
  \usepackage{hyperref}
  \usepackage{arydshln}
\MakeOuterQuote{"}%
\usetikzlibrary{calc,arrows.meta,positioning}%
\newcolumntype{L}{>{\raggedright\arraybackslash}p}

\definecolor
   {bgred}%
   {HTML}%
   {F6C1C1}
\definecolor
   {bgyellow}%
   {HTML}%
   {FFF5C4}
\definecolor
   {bggreen}%
   {HTML}%
   {CDEECF}
\definecolor
   {bgblue}%
   {HTML}%
   {C1C1F6}

\newtheorem
   {theorem}%
   {Theorem}%
   [section]%
\newtheorem
   {lemma}%
   [theorem]%
   {Lemma}%
\newtheorem
   {corollary}%
   [theorem]%
   {Corollary}%
\newtheorem
   {problem}%
   {Problem}%

\def\defn#1{\textbf{\textit{\boldmath #1}}}

{\makeatletter
 \gdef\xxxmark{%
   \expandafter\ifx\csname {@}mpargs\endcsname\relax 
     \expandafter\ifx\csname {@}captype\endcsname\relax 
       \marginpar{xxx}
     \else
       xxx 
     \fi
   \else
     xxx 
   \fi}%
 \gdef\xxx{\@ifnextchar[\xxx@lab\xxx@nolab}
 \long\gdef\xxx@lab[#1]#2{\textbf{[\xxxmark #2 ---{\sc #1}]}}
 \long\gdef\xxx@nolab#1{\textbf{[\xxxmark #1]}}
}

\let\realbfseries=\bfseries
\def\bfseries{\realbfseries\boldmath}%

\newcommand
   {\PER}%
   [1]%
   {\mathord{%
      \begin{tikzpicture}[baseline=(X.base),scale=0.2]
        \draw (-1.5,1) -- (1.5,1);
        \draw (-1.5,-1) -- (1.5,-1);
        \draw (-1,1.5) -- (-1,-1.5);
        \draw (1,1.5) -- (1,-1.5);
        \node (X) at (0,0) {$#1$};
      \end{tikzpicture}}}%

\def\u{_}
\def\L{\raisebox{-0.1ex}{\includegraphics[height=1.75ex]{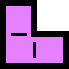}}\xspace}
\def\I{\raisebox{-0.1ex}{\includegraphics[height=1ex]{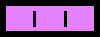}}\xspace}
\def\Ir{\raisebox{-0.1ex}{\includegraphics[height=2.5ex]{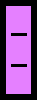}}\xspace}
\def\D{\raisebox{-0.1ex}{\includegraphics[height=1.75ex]{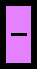}}\xspace}
\def\hide#1{}

\let\epsilon=\varepsilon

\begin{document}
\maketitle
\begin{abstract}
  Given a periodic placement of copies of a tromino (either \L or \I),
  we prove co-RE-completeness (and hence undecidability) of
  deciding whether it can be completed to a plane tiling.
  By contrast, the problem becomes decidable if the initial placement is
  finite, or if the tile is a domino \D instead of a tromino (in any dimension).
  As a consequence, tiling a given periodic subset of the plane
  with a given tromino (\L or \I) is co-RE-complete.

  We also prove co-RE-completeness of tiling the entire plane
  with two polyominoes
  (one of which is disconnected and the other of which has constant size),
  and of tiling 3D space with two connected polycubes
  (one of which has constant size).
  If we restrict to tiling by translation only (no rotation),
  then we obtain co-RE-completeness with one more tile:
  two trominoes for a periodic subset of 2D,
  three polyominoes for the 2D plane, and
  three connected polycubes for 3D space.

  Along the way, we prove several new complexity and algorithmic results about periodic (infinite) graphs.
  Notably, we prove that Periodic Planar (1-in-)3SAT-3, 3DM, and
  Graph Orientation are co-RE-complete in 2D and PSPACE-complete in 1D;
  we extend basic results in graph drawing to 2D periodic graphs; and
  we give a polynomial-time algorithm for perfect matching in
  bipartite periodic graphs.
\end{abstract}

\section{Introduction}

Given one or more \defn{prototiles} (shapes) and a target \defn{space}
(e.g., the plane), a \defn{tiling} \cite{Gruenbaum-Shephard}
is a covering of the space with
nonoverlapping copies of the prototiles, called \defn{tiles},
without gaps or overlaps.
By default, we allow the copies to translate, rotate, and reflect,
though reflections do not affect our (or most) results,
and we will also consider translation-only tiling.
In this paper, we study three fundamental computational problems about tilings:

\begin{problem}[$d$D Tiling] \label{prob:tiling}
  Given one or more prototiles, can they tile $d$-dimensional Euclidean space?
\end{problem}
\begin{problem}[$d$D Tiling Completion] \label{prob:completion}
  Given one or more prototiles, and given some already placed tiles,
  can this partial placement be extended to a tiling of
  $d$-dimensional Euclidean space?
\end{problem}
\begin{problem}[$d$D Subspace Tiling] \label{prob:subspace}
  Given one or more prototiles, and given a subset of $d$-dimensional
  Euclidean space, can the prototiles tile that space?
\end{problem}

Problem~\ref{prob:tiling} is a special case of Problem~\ref{prob:completion}
(with no preplaced tiles), and
Problem~\ref{prob:completion} is a special case of Problem~\ref{prob:subspace}
(where the preplaced tiles form the excluded subspace).
In Problems~\ref{prob:completion} and~\ref{prob:subspace}, there are multiple ways to
specify the preplaced tiles or subspace respectively:

\begin{enumerate}
\item \textbf{Finite}: There are finitely many preplaced tiles,
   or finitely many excluded regions from $d$-dimensional space,
   and we encode each explicitly.
\item \textbf{Periodic}: The preplaced tiles or excluded regions are periodic
   in $d' \leq d$ dimensions, and we encode the fundamental domain
   and the $d'$ translation vectors
   along which to repeat the fundamental domain.
   (Our results use $d'=d$.)
\item \textbf{Eventually periodic}:
   The preplaced tiles or excluded regions are periodic outside
   a finite region, so we use a hybrid: a periodic encoding,
   plus an explicit finite list of exceptions
   (excluded/included preplaced tiles or excluded regions).
   (Our results do not use this form of the problems,
   but we mention it for completeness.)
\end{enumerate}

All three problems have been shown \defn{undecidable}
(solved by no finite algorithm) in a variety of settings.
Such undecidability proofs generally simulate a Turing machine,
where finding an (infinite) tiling corresponds to the machine running forever,
which shows \defn{co-RE-hardness}.
Recently, Demaine and Langerman \cite{Demaine-Langerman-2025}
proved that these problems are in co-RE in very general settings,
and thus co-RE-hardness in fact establishes \defn{co-RE-completeness}.

Table~\ref{tab:history} summarizes the history of many such results,
focusing on Problem~\ref{prob:tiling},
but also capturing Problem~\ref{prob:subspace} in the form of a periodic "piece".
In general, we aim for undecidability under the following objectives:

\begin{enumerate}
\item Minimize the target dimension~$d$.
   In addition to integer $d$, we define $d=i+\frac12$ to consist of
   $i$ infinite real dimensions plus one bounded dimension
   given by a real interval.
   For example, 2.5D means $\mathbb R^2 \times [a,b]$
   for some $a, b \in \mathbb R$.
\item Minimize the number of distinct prototiles required.
\item Simplify the prototile shapes:
  \begin{enumerate}
  \item Prefer smaller families of shapes (e.g., polyominoes or polycubes)
     over general shapes (polygons or polyhedra).
  \item Prefer connected over disconnected shapes.
  \item Prefer shapes of smaller size/complexity
     (e.g., trominoes over pentominoes).
  \end{enumerate}
\item Minimize the preplaced tiles or excluded subspace:
   zero is better than finite, which is better than periodic,
   which is better than eventually periodic.
\end{enumerate}

\begin{table}
  \small
  \def\NEW{\textsc{new}}%
  \def\halfdown#1{\raisebox{-1ex}{#1}}
  \centering
  \begin{tabular}{l L{1.6in} L{1.6in} L{1.5in} l}%
    \toprule
    \multicolumn{1}{l}{\halfdown{Dim.}} & \multicolumn{2}{c}{Number/types of pieces\qquad} & \multicolumn{1}{c}{\halfdown{Result\qquad}} & \multicolumn{1}{l}{\halfdown{Date}} \\
    \cline{2-3}
         & Tiling by translation & by rotation${}+{}$translation &  & \\
    \midrule
    $d$D & \textbf{1 disconnected polycube + periodic} & n/a & undecidable \cite{greenfeld2025undecidability} & 2023-09 \\
    4D & 4 connected polycube & n/a & undecidable \cite{yang2024undecidabilitytranslationaltiling4dimensional} & 2024-09 \\
    4D & 3 connected polycube & n/a & undecidable \cite{yang2024undecidabilitytranslationaltilingtiles} & 2024-12 \\
    3D & 6 connected polycube & n/a & undecidable \cite{Yang_2025_sixpolycubes} & 2024-08 \\
    3D & 3 connected polycube & n/a & undecidable \cite{yang2025undecidabilitytiling3dimensionalspace} & 2025-07* \\
    3D & \textbf{2 connected polycube} & n/a & undecidable \cite{kim2025undecidabilitytranslationaltiling2} & 2025-08* \\
    \rowcolor{bgblue}\cellcolor{white}%
    2.5D & \textbf{3 connected polycube} & \textbf{2 connected polycube} & {undecidable (Cor.~{\ref{cor:threepolycube-translation}}~\&~\ref{cor:twopolycube})} & \NEW \\
    \midrule
    2D & $n$ connected polyomino & $n$ connected polyomino & undecidable \cite{Golomb1970} & 1970 \\
    2D & 11 connected polyomino & 5 connected polyomino & undecidable \cite{Ollinger09} & 2009-04 \\
    2D & 10 connected polyomino & n/a & undecidable \cite{Yang2023-10polyominoes} & 2023-02 \\
    2D & 9 connected polyomino & n/a & undecidable \cite{Yang2025-9polyominoes} & 2024 \\
    2D & 8 connected polyomino & n/a & undecidable \cite{Yang-Zhang-2024-8polyominoes} & 2024-03 \\
    2D & 7 polyomino & n/a & undecidable \cite{Yang-Zhang-2024-7polyominoes} & 2024-12 \\
    2D & 7 orthoconvex polyomino & n/a & undecidable \cite{Yang-Zhang-2025-7orthoconvexpolyominoes} & 2025-06* \\
    2D & 4 disconnected polyomino & n/a & undecidable \cite{Yang-Zhang-2025-4polyominoes} & 2025-06* \\
    2D & \textbf{5 connected polyomino} & \textbf{3 connected polyomino} & undecidable \cite{kim2025undecidabilitytilingplaneset} & 2025-08* \\
    2D & n/a & 3 polygons, or 2 polygons + periodic & undecidable, co-RE complete \cite{Demaine-Langerman-2025} & 2024-09 \\
    2D & n/a & \textbf{2 polyhex} & undecidable \cite{Stade2025-2tiling} & 2025-06* \\
    \rowcolor{bgblue}\cellcolor{white}%
    2D & \textbf{3 polyomino: 2~connected + 1~disconnected} & \textbf{2 polyomino: 1~connected + 1~disconnected} & undecidable (Cor.~{\ref{cor:threepolyomino-translation}}~\&~Thm.~\ref{thm:twopolyomino}) & \NEW \\
    \noalign{\vskip 1pt}%
    \rowcolor{bgblue}\cellcolor{white}%
    2D & \textbf{2 tromino + periodic} & \textbf{1 tromino + periodic} & undecidable ({Cor~\ref{cor:periodicsubsettranslation}}~\&~Thm.~\ref{thm:periodicsubset}) & \NEW \\
    \hdashline
    \rowcolor{bgyellow}%
    2D & 2 polyomino & 1 polyomino & OPEN & --- \\
    \hdashline
    \rowcolor{bgblue}\cellcolor{white}%
    $d$D & domino + periodic & domino + periodic & polynomial (Cor.~\ref{cor:domino-periodic-polytime}) & \NEW\\
    2D & 1 connected polyomino & n/a & decidable, $O(n)$ \cite{winslow2015translated_polyomino} & 2015 \\
    2D & 1 disconnected polyomino & n/a & decidable, periodic \cite{B20-R2-periodicity} & 2016-02 \\
    \rowcolor{bgblue}\cellcolor{white}%
    1.5D & \textbf{2 tromino + periodic} & \textbf{1 tromino + periodic} & PSPACE-complete ({Cor.~\ref{cor:periodicsubsettranslation}}~\&~Thm.~\ref{thm:periodicsubset}) & \NEW \\
    \noalign{\vskip 1pt}%
    \rowcolor{bgblue}\cellcolor{white}%
    1.5D & \textbf{3 polyomino: 2~connected + 1~disconnected} & \textbf{2 polyomino: 1~connected + 1~disconnected} & PSPACE-complete ({Cor.~\ref{cor:threepolyomino-translation}}~\&~Thm.~\ref{thm:twopolyomino}) & \NEW \\
    \midrule
    1D & $n$ polyomino & $n$ polyomino & decidable, $O(n)$ \cite{Greenfeld-Tao-2023} & 2023 \\
    \bottomrule
  \end{tabular}
  \caption{Past and new (un)decidability results for tiling Euclidean space,
    in various dimensions,
    and tiling by translation (left) or by rotation/translation (right).
    Bold indicates current hardness record holders.
    Dates are year-month, and ``*'' indicates recent independent work.
    Dashed lines separate undecidable from decidable. Our results are highlighted in blue.}%
  \label{tab:history}
\end{table}

In this paper, we improve the state-of-the-art for all three problems.

\subsection{Our Results: Tiling Completion}

We obtain particularly tight results for tiling completion
with a single prototile (see Section~\ref{sec:completion}):

\begin{enumerate}
\item Given a periodic preplacement of copies of a single tromino (\L or \I) in 2D,
   tiling completion is co-RE-complete (Theorem~\ref{thm:periodic-completion}),
   and hence undecidable.
   As a consequence, there are periodic preplacements that can be completed but only aperiodically (Corollary~\ref{cor:aperiodic-completion}).
   This undecidability result is tight by the following contrasting results:
\item Given a periodic preplacement of copies of a single tromino (\L or \I) in \emph{1.5D},
   tiling completion is PSPACE-complete (Theorem~\ref{thm:periodic-completion}),
   and hence decidable.
\item Given a \emph{finite} preplacement of copies of a single tromino (\L or \I) in 2D,
   tiling completion is NP-complete (Theorem~\ref{thm:finite-completion}),
   and hence decidable.
\item Given a periodic preplacement of \emph{dominoes} \D in $d$D for any $d \geq 1$,
   tiling completion can be solved in polynomial time
   (Corollary~\ref{cor:domino-periodic-polytime}).
   This is a consequence of the mathematical property that,
   if a periodic preplacement of dominoes \D in $d$D can be completed
   to a tiling, then it can be completed periodically with the same period
   (Corollary~\ref{cor:domino-completion-periodic}).
\end{enumerate}

Our results are the first to prove undecidability of tiling completion
with a \emph{single prototile}.
The earliest undecidability result for tiling completion was by
Robinson \cite{Robinson1971}, who used 36 Wang tiles
(which can be implemented by 36 polyominoes \cite{Golomb1970})
and required only finite preplacement.
Yang \cite{Yang-2013-thesis} showed how rule 110 can be simulated using 6 Wang tiles, and thus 6 polyominoes; together with Cook's undecidability of rule 110 for an eventually periodic initial configuration \cite{cook2004universality}, this implies undecidability of tiling completion with 6 polyominoes and eventually periodic preplacement \cite{Demaine-Langerman-2025}, using only translations.
More recently, this bound was lowered for fixed general polygons
to 3 using finite preplacement, or 2 using eventually periodic preplacement \cite{Demaine-Langerman-2025}.
In addition to reducing to 1 shape and simplifying the shape to a polyomino,
we characterize the exact size of polyomino (3) and dimension (2)
required for undecidability,
while (necessarily) requiring periodic preplacement.

\subsection{Our Results: Subspace Tiling}

Our hardness results for tiling completion immediately apply to
the more general Problem~\ref{prob:subspace}.
In fact, all four results generalize (see Section~\ref{sec:subspace}):

\begin{enumerate}
\item Tiling a given periodic polyomino subset of 2D
   with copies of a single tromino (\L or \I) is co-RE-complete (Theorem~\ref{thm:periodicsubset}).
   This result is tight by the following contrasting results:
\item Tiling a given periodic polyomino subset of \emph{1.5D}
   with copies of a single tromino (\L or \I) is PSPACE-complete (Theorem~\ref{thm:periodicsubset}).
\item Tiling a given \emph{finite} polyomino subset of 2D
   with copies of a single tromino (\L or \I) is NP-complete.
   This result was already known for both \I
   \cite{Beauquier-Nivat-Remila-Robson-1995,Horiyama-Ito-Nakatsuka-Suzuki-Uehara-2017} and
   \L \cite{Moore-Robson-2001,Horiyama-Ito-Nakatsuka-Suzuki-Uehara-2017},
   but we give a simpler proof based on Planar 3DM.
\item Tiling a given periodic polyomino subset of $d$D
   with a \emph{domino} \D is polynomial (Corollary~\ref{cor:domino-periodic-polytime}),
   and when possible, possible with the same period
   (Theorem~\ref{thm:domino-periodic}), for any $d \geq 1$.
\end{enumerate}

The main point of comparison is a recent result of
Greenfeld and Tao \cite{greenfeld2025undecidability}:
tiling a given periodic subset of $d$D space, where $d$ is part of the input,
with a single disconnected polyhypercube by translation is co-RE-complete.
Our result improves the number of dimensions from unbounded
(depending on the Turing machine being simulated) to the optimal~$2$,
improves the shape to be connected, and optimizes the complexity of the shape.
However, while Greenfeld and Tao's results is translation-only,
our result requires rotations.
For translation-only, our results need two trominoes:

\begin{enumerate}[start=5]
  \item Tiling a given periodic polyomino subset of 2D by translation
    with the two \I trominoes (\I and \Ir) is co-RE-complete (Corollary~\ref{cor:periodicsubsettranslation}).
  \item Tiling a given periodic polyomino subset of \emph{1.5D} by translation
    with the two \I trominoes (\I and \Ir) is PSPACE-complete (Corollary~\ref{cor:periodicsubsettranslation}).
\end{enumerate}

The only previous result exploiting rotations is that tiling a given
periodic polygonal subset of 2D with two polygons is co-RE-complete
\cite{Demaine-Langerman-2025}.
Our results improve the number of shapes by~$1$,
improve the shape to be polyomino,
and optimize the complexity of the polyomino to the optimal~$3$.

\subsection{Our Results: Tiling}

With a little work (see Section~\ref{sec:two-polyomino}),
we can convert the complement of the periodic subspace
into a second prototile which is a disconnected polyomino,
while expanding the tromino prototiles to constant-size connected polyominoes.
As a result, we obtain several new results for the most specific
Problem~\ref{prob:tiling}:

\begin{enumerate}
\item Tiling 2D with two polyomino prototiles, one of which is disconnected
   and the other of which is connected and has constant size,
   is co-RE-complete (Theorem~\ref{thm:twopolyomino}).
\item Tiling 2.5D or 3D with two connected polycube prototiles,
   one of which has constant size, is co-RE-complete (Corollary~\ref{cor:twopolycube}).
   The 2.5D result needs just two layers.
\item Tiling 1.5D with two polyomino prototiles, one of which is disconnected
   and the other of which is connected and has constant size,
   is PSPACE-complete (Theorem~\ref{thm:twopolyomino}).
\end{enumerate}

Until recently, the only previous polyomino tiling results were the
original by Golomb \cite{Golomb1970}, which required arbitrarily many polyominoes
(depending on the Turing machine being simulated),
and Ollinger's improvement to just five polyominoes \cite{Ollinger09}.
Very recently (and independent of our work),
Kim \cite{kim2025undecidabilitytilingplaneset}
improved this result to just three polyominoes
(also improving on a prior result for three \emph{polygons}
\cite{Demaine-Langerman-2025}).
Even so, by allowing one polyomino to be disconnected or lifting to 3D,
we improve the bound to just two polyominoes/polycubes.
Also very recently (and independent of our work),
Stade \cite{Stade2025-2tiling} showed undecidability of tiling
with two \emph{polyhexes},
while mentioning that it seems difficult to modify his construction
to two polyominoes.
Our result achieves this goal, by very different techniques,
and either allowing one polyomino to be disconnected or lifting to 3D.

Table~\ref{tab:tiling} summarizes the current state-of-the-art
for tiling space, allowing rotation.

\begin{table}
  \centering
  \begin{tabular}{c c c c c c}%
    \toprule
    pieces & 1D & 2D & 3D & $d$D \\
    \midrule
    \rowcolor{bgyellow}
    \cellcolor{white}1 & \cellcolor{bggreen}R ($\downarrow$) & OPEN & OPEN & OPEN \\
    \rowcolor{bgred}
    \cellcolor{white}2 & \cellcolor{bggreen}R ($\downarrow$) & co-RE-c (Thm.~\ref{thm:twopolyomino}) & co-RE-c ($\leftarrow$, Cor.~\ref{cor:twopolycube}) & co-RE-c ($\leftarrow$) \\
    \rowcolor{bgred}
    \cellcolor{white}3 & \cellcolor{bggreen}R ($\downarrow$) & co-RE-c \cite{kim2025undecidabilitytilingplaneset} & co-RE-c ($\leftarrow$) & co-RE-c ($\leftarrow$) \\
    \rowcolor{bgred}
    \cellcolor{white}$n$ & \cellcolor{bggreen}R \cite{greenfeld2025undecidability} & co-RE-c ($\uparrow$) & co-RE-c ($\uparrow$) & co-RE-c ($\uparrow$) \\
    \bottomrule
  \end{tabular}
  \caption{Current records for complexity of tiling space (allowing rotation) with possibly disconnected polycubes.
    "R" denotes decidability, and "co-RE-c" denotes co-RE-completeness.}%
  \label{tab:tiling}
\end{table}

For translation-only tiling, our results need three polyominoes
(to represent the two rotations of the \I tromino):

\begin{enumerate}[start=4]
  \item Tiling 2D by translation with three polyomino tiles,
    one of which is disconnected
    and the other two of which are connected and have constant size,
    is co-RE-complete (Corollary~\ref{cor:threepolyomino-translation}).
  \item Tiling 2.5D or 3D by translation with three connected polycube tiles,
    two of which have constant size, is co-RE-complete (Corollary~\ref{cor:threepolycube-translation}).
    The 2.5D result needs just two layers.
  \item Tiling 1.5D by translation with three polyomino tiles,
    one of which is disconnected
    and the other two of which are connected and have constant size,
    is PSPACE-complete (Corollary~\ref{cor:threepolyomino-translation}).
\end{enumerate}

For comparison, when tiling by translation,
Ollinger's construction \cite{Ollinger09} needs 11 connected polyominoes.
This bound was reduced to 10 \cite{Yang2023-10polyominoes}, then 9 \cite{Yang2025-9polyominoes},
then 8 \cite{Yang-Zhang-2024-8polyominoes}, then 7 \cite{Yang-Zhang-2024-7polyominoes},
and then 4 disconnected polyominoes \cite{Yang-Zhang-2025-4polyominoes}.
Very recently (and independent of our work),
Kim \cite{kim2025undecidabilitytilingplaneset}
proved undecidability of translation-only tiling with
just 5 connected polyominoes.
Our results improve these results to 3 polyominoes,
either allowing one polyomino to be disconnected or lifting to 3D.
But very recently (and independently),
Kim \cite{kim2025undecidabilitytranslationaltiling2}
improved the 3D result to just 2 connected polycubes.
(In 2.5D, our 3 connected polycubes still hold the record.)

Table~\ref{tab:tiling-translation} summarizes the current state-of-the-art
for tiling space by translation only.

\begin{table}
  \centering
  \begin{tabular}{c c c c c c}%
    \toprule
    pieces & 1D & 2D & 3D & $d$D \\
    \midrule
    \rowcolor{bgyellow}
    \cellcolor{white}1 & \cellcolor{bggreen}R ($\downarrow$) & \cellcolor{bggreen}R \cite{B20-R2-periodicity,Beauquier-1991} & OPEN & OPEN \\
    \rowcolor{bgred}
    \cellcolor{white}2 & \cellcolor{bggreen}R ($\downarrow$)& \cellcolor{bgyellow}OPEN & co-RE-c \cite{kim2025undecidabilitytranslationaltiling2} & co-RE-c ($\leftarrow$) \\
    \rowcolor{bgred}
    \cellcolor{white}3 & \cellcolor{bggreen}R ($\downarrow$)& co-RE-c (Cor.~{\ref{cor:threepolyomino-translation}}) & co-RE-c ($\leftarrow$, Cor.~\ref{cor:threepolycube-translation}) & co-RE-c ($\leftarrow$) \\
    \rowcolor{bgred}
    \cellcolor{white}$n$ & \cellcolor{bggreen}R \cite{greenfeld2025undecidability} & co-RE-c ($\uparrow$) & co-RE-c ($\uparrow$) & co-RE-c ($\uparrow$) \\
    \bottomrule
  \end{tabular}
  \caption{Current records for complexity of tiling space by translation only with possibly disconnected polycubes.
    "R" denotes decidability, and "co-RE-c" denotes co-RE-completeness.}%
  \label{tab:tiling-translation}
\end{table}

\subsection{Our Results: Periodic Graphs}

\begin{figure}
  \centering
  \includegraphics[width=0.5\textwidth]{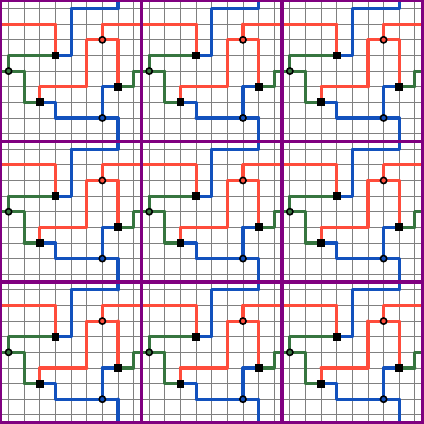}
  \caption{A planar periodic orthogonal drawing of a local 2D periodic graph,
    with 6 protovertices and 9 protoedges.}%
  \label{fig:periodicgraph-example}
\end{figure}

All of our tiling results are based on new results about periodic
(infinite) graphs, developed in Sections~\ref{sec:periodicgraphs},
\ref{sec:periodic-reductions}, and \ref{sec:periodic-algs}.
Conceptually, a $d$D periodic graph is defined by a fundamental domain,
which gets repeated by translation in $d$ directions;
see Figure~\ref{fig:periodicgraph-example}.
More precisely (and abstractly), a \defn{$d$D periodic graph} is defined
by a set of \defn{protovertices}, where each protovertex gets repeated
for every point in the $d$D integer lattice,
and a set of \defn{protoedges}, which specify not only which two protovertices
to connect but also the distance vector in the integer lattice
of the actual vertices to connect (again periodically).
In most cases, we assume the graph is \defn{local},
meaning that the distance vector of each protoedge has $L_1$ norm at most~$1$,
i.e., it connects protovertices in the same or neighboring lattice points.

Such graphs were originally studied by Orlin in the 1D case
\cite{Orlin1981_STOC, orlin1983matchingI, orlin1983matchingII, Orlin1984}
and then by others in 2D and higher.
(In some cases, periodic graphs are called "dynamic graphs"
\cite{Iwano1987,IwanoSteiglitz1987,IwanoSteiglitz1988planarity}.)
Table~\ref{tab:graphs} summarizes the history of these results.
Notably useful for our work is that Periodic SAT is
PSPACE-complete in 1D periodic graphs \cite{Orlin1981_STOC} and
undecidable in 2D periodic graphs \cite{Freedman1998ksat,MaratheHuntRosenkrantzStearns1998}.
We give careful proofs of these prior results
so that this paper can stand alone as
an introduction to periodic graphs.

\begin{table}
  \def\complete{c\xspace}
  \def\PseudoPolynomial{pseudo-P\xspace}
  \def\Polynomial{P\xspace}
  \def\Linear{linear $\Rightarrow$ P\xspace}
  \centering
  \begin{tabular}{>{\hangindent=1.5em} L{2.2in} L{1.25in} L{1.2in} L{1.2in}}%
    \toprule
    Problem & 1D & 2D & $d$D\\
    \midrule
    2SAT & \Polynomial \cite{MaratheHuntStearnsRosenkrantz1995} & decidable \cite{Freedman1998ksat} & \cellcolor{bgblue}\Polynomial (Thm.~\ref{thm:periodic-2sat}) \\
    (Dual) Horn SAT & \Polynomial ($\rightarrow$) & \Polynomial \cite{MaratheHuntRosenkrantzStearns1998} & \cellcolor{bgblue}\Linear (Thm.~\ref{thm:periodic-horn}) \\
    3SAT & PSPACE-\complete \cite{Orlin1981_STOC} & undecidable \cite{Freedman1998ksat,MaratheHuntRosenkrantzStearns1998} & \cellcolor{bgblue}co-RE-\complete (Thm.~\ref{thm:periodic-3sat}) \\
    3SAT, wide & EXPSPACE-\complete \cite{MaratheHuntRosenkrantzStearns1998} & undecidable \cite{Freedman1998ksat} & \cellcolor{bgblue}co-RE-\complete (Thm.~\ref{thm:periodic-3sat}) \\
    Planar 3SAT & PSPACE-\complete \cite{MaratheHuntStearnsRosenkrantz1995} & \cellcolor{bgblue}co-RE-\complete (Thm.~\ref{thm:planar-3sat}) & \cellcolor{bgblue}co-RE-\complete (Thm.~\ref{thm:planar-3sat}) \\
    Planar 3SAT-3 & \cellcolor{bgblue}PSPACE-\complete (Thm.~\ref{thm:planar-3sat}) & \cellcolor{bgblue}co-RE-\complete (Thm.~\ref{thm:planar-3sat}) & n/a \\
    Planar 1-in-3SAT-3 & \cellcolor{bgblue}PSPACE-\complete (Thm.~\ref{thm:planar-1-in-3sat}) & \cellcolor{bgblue}co-RE-\complete (Thm.~\ref{thm:planar-1-in-3sat}) & n/a \\
    3DM (3D Matching) & PSPACE-\complete \cite{Orlin1981_STOC} & \cellcolor{bgblue}PSPACE-\complete (Thm.~\ref{thm:planar-3dm}) & \cellcolor{bgblue}PSPACE-\complete (Thm.~\ref{thm:planar-3dm}) \\
    Planar 3DM & \cellcolor{bgblue}PSPACE-\complete (Thm.~\ref{thm:planar-3dm}) & \cellcolor{bgblue}co-RE-\complete (Thm.~\ref{thm:planar-3dm}) & n/a \\
    \midrule
    Planar trichromatic graph orientation & \cellcolor{bgblue}PSPACE-\complete (Thm.~\ref{thm:planar-graph-orientation}) & \cellcolor{bgblue}co-RE-\complete (Thm.~\ref{thm:planar-graph-orientation}) & n/a \\
    Independent set & PSPACE-\complete \cite{Orlin1981_STOC} & PSPACE-\complete ($\leftarrow$) & PSPACE-\complete ($\leftarrow$) \\
    Hamiltonian cycle & PSPACE-\complete \cite{Orlin1981_STOC} & PSPACE-\complete ($\leftarrow$) & PSPACE-\complete ($\leftarrow$) \\
    Bipartiteness, wide & \Polynomial \cite{Orlin1984} & \Polynomial \cite{Iwano1987} & \Polynomial \cite{CohenMegiddo1991} \\
    3-coloring & PSPACE-\complete \cite{Orlin1981_STOC} & undecidable \cite{Freedman1998ksat} & undecidable ($\leftarrow$) \\
    Cycle detection, wide & NC $\Rightarrow$ \Polynomial \cite{CohenMeggido1993} & NC $\Rightarrow$ \Polynomial \cite{CohenMeggido1993} & NC $\Rightarrow$ \Polynomial \cite{CohenMeggido1993} \\
    Connected components, wide & \Polynomial \cite{Orlin1984} & \Polynomial \cite{Iwano1987} & \Polynomial \cite{CohenMegiddo1991} \\
    Strongly connected components, wide & \Polynomial \cite{Orlin1984} & \Polynomial \cite{CohenMeggido1993,KodialamOrlin1991} & \Polynomial \cite{CohenMeggido1993,KodialamOrlin1991} \\
    Reachability & \Polynomial \cite{HoftingWanke1993} & \Polynomial \cite{HoftingWanke1993} & \Polynomial \cite{HoftingWanke1993} \\
    Reachability, wide & NP-\complete \cite{Orlin1984} & NP-\complete \cite{HoftingWanke1993} & NP-\complete \cite{HoftingWanke1993} \\
    Shortest paths & \Polynomial \cite{HoftingWanke1993} & \Polynomial \cite{HoftingWanke1993} & \Polynomial \cite{HoftingWanke1993} \\
    Shortest paths, wide & NP-\complete \cite{Orlin1984} & NP-\complete \cite{HoftingWanke1993} & NP-\complete \cite{HoftingWanke1993} \\
    Eulerian path & \Polynomial \cite{Orlin1984} & \Polynomial \cite{Iwano1987} \\
    Minimum spanning tree, wide & \Polynomial \cite{Orlin1984} & \Polynomial \cite{CohenMegiddo1991} & \Polynomial \cite{CohenMegiddo1991} \\
    Planarity testing \small (avoiding accumulation points) & \Polynomial \cite{IwanoSteiglitz1988planarity} & \Polynomial \cite{IwanoSteiglitz1988planarity} \\
    Weighted matching & \Polynomial \cite{orlin1983matchingI,orlin1983matchingII} \\
    Bipartite perfect matching, wide & \Polynomial ($\uparrow$, Thm.~\ref{thm:matchingalgorithm}) & \cellcolor{bgblue}\Polynomial (Thm.~\ref{thm:matchingalgorithm}) & \cellcolor{bgblue}\Polynomial (Thm.~\ref{thm:matchingalgorithm}) \\
    \bottomrule
  \end{tabular}
  \caption{Past and new results about $d$D periodic SAT and graphs,
    for fixed $d$.
    These results generally assume locality (or polynomial locality),
    except for problems marked ``wide''. Our results are highlighted in blue.
    }
  \label{tab:graphs}
\end{table}

More interestingly, we prove the following new results about periodic graphs
and related periodic satisfiability problems:

\paragraph{Complexity Results}
\Needspace{5\baselineskip}%
\begin{enumerate}
\item We characterize the complexity of 2D (and higher dimensional) periodic SAT
   as \defn{co-RE-complete}, not just undecidable (Theorem~\ref{thm:periodic-cnf-sat}).
   Our proof of membership in co-RE uses a weak form of compactness that
   holds under standard ZF axioms (without the Axiom of Choice).
\item Many \emph{planar} problems are hard on periodic graphs
   (PSPACE-complete in 1D and co-RE-complete in 2D):
   Planar 3SAT-3, Planar 1-in-3SAT, Planar 3DM, and a new problem we introduce
   called Planar Trichromatic Graph Orientation;
   see Section~\ref{sec:periodic-reductions}.
\end{enumerate}

\paragraph{Algorithmic Results}
\Needspace{5\baselineskip}%
\begin{enumerate}[start=3]
 \item These planar hardness results are built on generalizations of classic results
   in graph drawing to periodic graphs:
   \begin{enumerate}
   \item Every periodic graph has a periodic "nice" drawing in the plane,
      where edges do not pass through vertices except their endpoints,
      and edges cross orthogonally and only in pairs (Theorem~\ref{thm:orthocrossing-drawing}).
      This result enables planarizing any problem
      given an (orthogonal) crossover gadget.
   \item Every periodic crossing-free drawing of a maximum-degree-4 periodic graph
      has an orthogonal crossing-free drawing (Lemma~\ref{lem:orthodraw}).
      This result allows us to draw intermediate reductions as general graphs,
      while still guaranteeing an orthogonal result
      (as needed by polyomino tiling).
   \end{enumerate}
 \item Contrasting with the 3SAT hardness results,
   we show that both 2SAT and Horn-SAT can be solved in polynomial time on $d$D periodic graphs (Sections~\ref{sec:2sat} and~{\ref{sec:horn}}). 
 \item Deciding whether a given $d$D periodic bipartite graph has a \emph{perfect matching} can be
   solved in polynomial time (Theorem~\ref{thm:matchingalgorithm}), and in positive instances,
   there is always a periodic perfect matching with the same period (Theorem~\ref{thm:periodicmatching}).
   This problem generalizes tiling with dominoes \D
   in a periodic polycube subspace of $d$D.
   The core approach is the following:
   \begin{enumerate}
   \item We show the existence of \emph{finite} augmenting paths,
      at least when the graph has a perfect matching.
      This follows from a packing argument.
   \item We show further that the \emph{shortest} augmenting path
      does not repeat a vertex, which allows us to augment the path
      while preserving periodicity of the matching.
      Effectively, we show how to simplify crossings between
      translates of the augmenting path.
   \end{enumerate}
\end{enumerate}

\section{Periodic Graphs \label{sec:periodicgraphs}}

In this section, we formally define periodic graphs
(similar to past work such as \cite{Orlin1984,CohenMegiddo1991,HoftingWanke1993,MaratheHuntRosenkrantzStearns1998}),
while introducing our own notation which we view as slightly cleaner.
We also define periodic drawings, and prove basic results about graph drawing.

A \defn{$d$-dimensional periodic graph} is an infinite graph
$\PER{G} = (\PER{V}, \PER{E})$
induced by a finite set $V$ of \defn{protovertices}
and a finite set $E$ of \defn{protoedges}.%
\footnote{We pronounce $\PER{X}$ as ``periodic $X$''. Note that we assume $V$ is finite, which we might call ``finite period''. We do not allow, for example, $V$ to be the unit square $[0,1]^2$ for a periodic graph over the Euclidean plane.}
\ 
Each protovertex $v \in V$ induces a countably infinite grid of vertices,
$v^{\vec x} \in \PER{V}$ for each lattice point $\vec x \in \mathbb Z^d$:
$$\PER{V}= \{v^{\vec x} \mid v\in V, \vec x \in \mathbb Z^d\}.$$
Each protoedge $(u^{\vec x},v^{\vec y}) \in E$ is between two vertices
(not protovertices), and induces a countably infinite set of edges
$(u^{\vec x + \vec \Delta}, v^{\vec y + \vec \Delta})$
for all offset vectors $\vec \Delta \in \mathbb Z^d$:
$$
  \PER{E} = \left\{
    (u^{\vec x + \vec \Delta}, v^{\vec y + \vec \Delta})
    \,\middle|\,
    (u^{\vec x},v^{\vec y}) \in E,
    \vec \Delta \in \mathbb Z^d
  \right\}.
$$


A \defn{labeling} of the vertices $\PER{V}$ of the periodic graph $\PER{G}$ is a mapping $\lambda : \PER{V} \rightarrow L$ from the vertices to some label set $L$, such as integers, colors, or strings.
Such a labeling is \defn{$k$-periodic} if $\lambda(v^{\vec x}) = \lambda(v^{\vec x + k\vec\Delta})$ for all $v^{\vec x}\in\PER{V}$ and $\vec\Delta\in\mathbb Z^d$.
Likewise, a \defn{labeling} $\lambda : \PER{E} \rightarrow L$ of the edges $\PER{E}$ is \defn{$k$-periodic} if $\lambda((u^{\vec x},v^{\vec y})) = \lambda((u^{\vec x + k\vec\Delta},v^{\vec y + k\vec\Delta}))$ for all $(u^{\vec x},v^{\vec y})\in\PER{E}$ and $\vec\Delta\in\mathbb Z^d$.
A 1-periodic labeling of vertices or edges of $\PER{G}$ is equivalent to a labeling of the protovertices or protoedges (with identical labels on all copies).
Labelings provide the formalism to talk about structures on top of graphs,
such as colorings, weights, and subsets (where the labeling indicates membership
in the subset) of vertices or edges (e.g., matchings).

The periodic graph is \defn{local} if every protoedge $(u^{\vec x},v^{\vec y}) \in E$ spans an $L_1$ distance $\|\vec x - \vec y\|_1$ of at most $1$.
More generally, a graph is \defn{$k$-local} if every protoedge $(u^{\vec x},v^{\vec y}) \in E$ has $\|\vec x - \vec y\|_1 \leq k$.

If a periodic graph is $k$-local, then there is an equivalent ($1$-)local graph by expanding the prototile to $k \times k$ original prototiles.
So algorithmically, the only difference between local and $k$-local is when $k$ is larger than a polynomial.
Then, for example, reachability becomes weakly NP-hard \cite{HoftingWanke1993} and SAT becomes EXPSPACE-hard \cite{MaratheHuntRosenkrantzStearns1998}.
Prior works (e.g., \cite{Orlin1981_STOC,MaratheHuntRosenkrantzStearns1998}) use the term \defn{narrow} for local or polynomially local periodic graphs, and \defn{wide} when the locality is not bounded by a polynomial.


\subsection{Periodic Drawings and Planarity}

A \defn{periodic drawing} of a 1D or 2D periodic graph is defined by a mapping from each protovertex $v \in V$ to a distinct point $\vec p_v$ in the open unit square $(0,1)^2$ (excluding the boundary), and a mapping from each protoedge $(u^{\vec x},v^{\vec y}) \in E$ to a simple curve that connects point $\vec p_u + \vec x$ to point $\vec p_v + \vec y$ without visiting any other points $\vec p_w + \vec z$ corresponding to a vertex $w^{\vec z} \in \PER{V}$.
Then, for $\vec \Delta \in\mathbb Z^2$ (padding with a zero $y$ coordinate for a 1D graph), the vertex $v^{\vec \Delta} \in \PER{V}$ has coordinates $\vec p_v + \vec \Delta$, and the edge $(u^{\vec x + \vec\Delta},v^{\vec y + \vec\Delta}) \in \PER{E}$ is mapped to the curve representing protoedge $(u^{\vec x},v^{\vec y})$ translated by $\vec\Delta$.
A periodic drawing is \defn{local} if the curve representing each edge $(u^{\vec x},v^{\vec y})$ is contained in the union of the unit grid squares containing $\vec p_u$ and $\vec p_v$ (and their shared boundary if the squares are distinct);
this property implies that the periodic graph is local.

In this paper, all edges will be drawn as polygonal chains (or straight line segments).
A periodic drawing is \defn{orthogonal} if every edge is drawn as a polygonal chain composed of horizontal and vertical segments.
A periodic drawing is \defn{orthocrossing} if,
whenever two drawn edges meet at a common point other than a shared endpoint,
it is a (proper) crossing between a horizontal and vertical segment
of the respective edges.
In particular, this notion forbids edges from touching without crossing,
and forbids three edges meeting at a point other than a shared endpoint.

When all points defining the drawing (graph vertices and corners of the polygonal chains representing edges, and for orthocrossing drawings, all orthogonal crossing points) have rational coordinates, the \defn{grid size} of the periodic drawing is the size of the grid on which those points lie. That is, assuming all coordinates are of the form $z/M$ where $z$ and $M$ are integers, and using the same denominator $M$ for all points, the grid size is $M$.

\begin{theorem}\label{thm:orthocrossing-drawing}
  Every 1D or 2D local periodic graph has an orthocrossing local periodic drawing of grid size $O(|V|+|E|)$.
  Furthermore, if the graph has maximum degree $4$, then the drawing can be orthogonal.
\end{theorem}


\begin{proof}
  Let $d(v)$ be the degree of protovertex $v$ in the periodic graph $\PER{G}$, which is the total number of occurrences of protovertex $v$ as either endpoint of an edge in $E$ (counting twice when the edge involves $v$ in both endpoints).
  We will draw each vertex $v$ as well as $d(v)$ \defn{ports} for connecting
  to the incident edges; see Figure~\ref{fig:ports}.

  \begin{figure}[bt]
    \centering
    \includegraphics[width=0.4\textwidth]{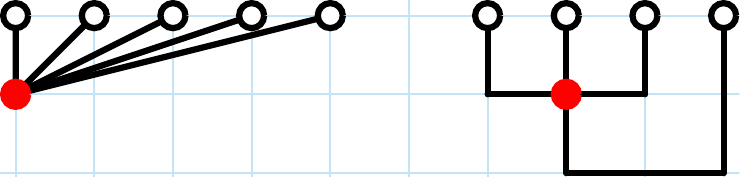}
    \caption{A vertex with degree $5$ connected to its $5$ ports, and a vertex of degree $4$ connected to its $4$ ports orthogonally.}%
    \label{fig:ports} 
  \end{figure}

  \begin{figure}
    \centering
    \includegraphics[width=0.5\textwidth]{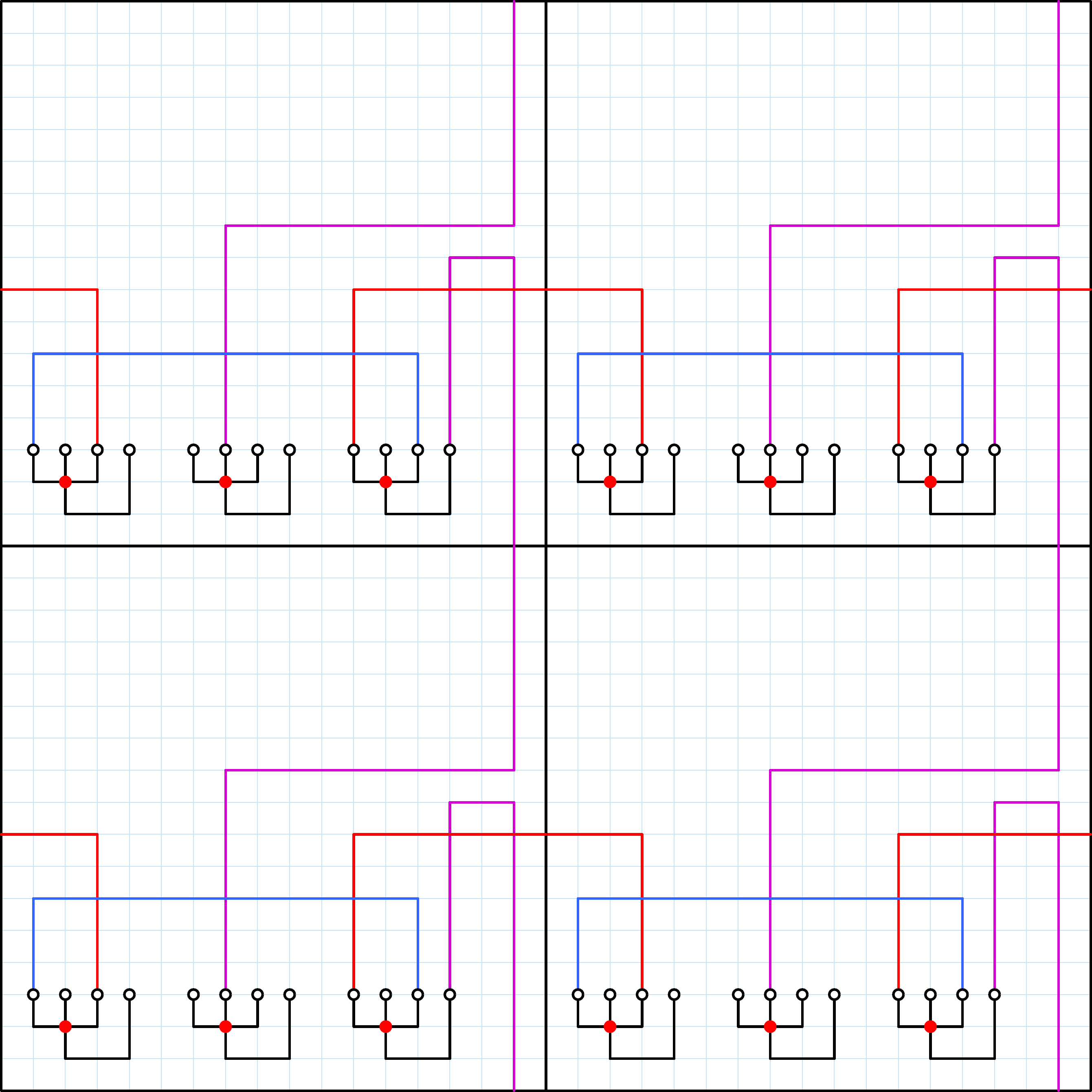}
    \caption{An orthocrossing local periodic drawing of a local periodic graph.}%
    \label{fig:orthocrossing}
  \end{figure}

  \begin{itemize}
  \item Let $M=5\,|V|+2\,|E|$.
  \item Initialize $a$ to $1$.
  \item For each protovertex $v \in V$:
    \begin{itemize}
    \item Increment $a$ by $1$.
    \item Place protovertex $v$ at coordinates $((a+1)/M, 2/M)$.
    \item For each occurrence of $v$ in an edge $e$:
      \begin{itemize}
      \item Place port $p$ for that occurrence of $v$ in $e$ at coordinates $(a/M, 3/M)$.
      \item Draw an edge from $v$ to $p$. If the graph is of maximum degree 4, route orthogonally, as in Figure~\ref{fig:ports}; otherwise, use a straight line segment.
      \item Increment $a$ by $1$.
      \end{itemize}
    \end{itemize}
  \item Initialize $b$ to $4$.
  \item For each protoedge $e = (u^{\vec x},v^{\vec y}) \in E$:
    \begin{itemize}
    \item Increment $b$ by $2$.
    \item Let $p$ and $q$ be the ports for these occurrences of $u$ and $v$ in $e$, respectively.
    \item If $\vec x = \vec y$, then draw $e$ by drawing vertical segments upward from $p$ and from $q$ to height $b/M$, and connecting these points horizontally.
    \item If $\vec y = \vec x \pm (1,0)$, then draw vertical segments upward from $p$ and from $q$ to height $b/M$, and connect these points horizontally, to draw $e$. But there are two possible horizontal connections, and we choose to go from the rightmost of $p$ and $q$ rightward (into the next unit square) to the leftmost of $p$ and $q$. This choice prevents this polygonal path from crossing translated copies of itself.
    \item If $\vec y = \vec x \pm (0,1)$, then
      \begin{itemize}
      \item Assume $\vec y = \vec x + (0,1)$ (or swap the vertices).
      \item Increment $a$ by $1$.
      \item Place a \defn{gate} $g$ at coordinates $(a/M, 1)$ (at the top of the unit square, which has a translated copy $(a/M,0)$ at the bottom of the unit square).
      \item Draw $e$ by drawing vertical segments upward from $p$ and from $q$ to heights $b$ and $b+1$ respectively, connecting these points horizontally to $x$ coordinate $a/M$, and then connecting these points vertically through gate $g$. But there are two possible vertical connections, and we choose to go from height $b+1$ up to height $b$. This choice prevents this polygonal path from crossing translated copies of itself.
      \end{itemize}
    \end{itemize}
  \end{itemize}
  
  Figure~\ref{fig:orthocrossing} shows an example including each of the three cases. By construction, the result is an orthocrossing local periodic drawing of grid size $M$.
  In particular, assigning each protovertex occurrence to its own $x$ coordinate and each protoedge to its own pair of $y$ coordinates guarantees that edges only cross orthogonally and hence no three edges cross at a common point.
\end{proof}

A periodic drawing is \defn{planar} if edges do not intersect except at shared endpoints.
For local periodic drawings with edges embedded as polygonal chains or segments, this noncrossing property can be checked
in polynomial time by checking for intersections
within the origin square $[0,1]^2$,
which must come from edges whose endpoints $v^{\vec x}$
have $\|\vec x\|_1 \leq 1$.
(More generally, for $k$-local graphs,
it suffices to consider vertices $v^{\vec x}$ where $\|\vec x\|_1 \leq k$.)
If there are no crossings within $[0,1]^2$, then by periodicity,
there are no crossings anywhere.

A periodic graph is \defn{planar} if it has a periodic planar drawing.
This property can be checked in polynomial time \cite{IwanoSteiglitz1988planarity}.

\subsection{Orthogonal Planar Drawings of Periodic Graphs}

Although Theorem~\ref{thm:orthocrossing-drawing} gives an orthogonal drawing for
max-degree-4 graphs, it may have crossings.
Next we show that any \emph{planar} drawing of a max-degree-4 graph
can be made \emph{orthogonal and planar}, with quadratic blowup on the grid size.
This lets us use planar drawings (e.g., for gadgets in reductions),
while guaranteeing that they can be made orthogonal.

\begin{lemma}\label{lem:orthodraw}
  For a periodic graph of maximum degree $4$, every planar periodic graph drawing of grid size $M$ can be converted into a planar orthogonal periodic drawing, preserving the position of the graph vertices, of grid size $O(M^2)$.
\end{lemma}

\begin{proof}
  First, assume all edges in the drawing are straight line segments. For drawings with polygonal chains, we can temporarily add a vertex at each bend of every chain, apply the transformation to the resulting graph, and merge back the resulting orthogonal chains, removing the extra degree-2 vertices.

  We explicitly transform the periodic drawing within the $[0,1]^2$ square into an orthogonal drawing. 
  Let $\epsilon = 1/30M^2$. 
  Associate each protovertex $v$ with the point within the $[0,1]^2$ square at which it is currently drawn, which will not change in our transformation.
  Surround every vertex $v$ with a $5\epsilon \times 10\epsilon$ rectangle centered at $v$.
  These rectangles do not intersect each other and do not intersect any edges as currently drawn.
  We route the edges incident to $v$ through one of 10 anchor points on the boundary of the rectangle: 4 on the left, 4 on the right, 1 on the top, and 1 on the bottom.
  We assign anchor points to edges incident to $v$ as follows: edges connecting $v$ to other vertices with the same $x$ coordinate and smaller or larger $y$ coordinate connect through the bottom and top anchor points, respectively; and edges connecting $v$ to other vertices with smaller or larger $x$ coordinates connect through the left and right anchor points, respectively.
  For edges connecting to vertices with larger $x$ coordinates, the edges connect through the anchor points along the right side of the rectangle, by decreasing order of slope, starting at the top anchor point and continuing downward.
  For edges connecting to vertices with smaller $x$ coordinates, the edges connect through the anchor points along the left side of the rectangle, by decreasing order of slope, starting at the bottom anchor point and continuing upward.

  For any edge intersecting the boundary of the $[0,1]^2$ square, add an extra vertex at the intersection point. This added vertex connects to only one edge within the square, so it serves as its own anchor point.

  Let $v_1,\ldots,v_n$ be the vertices of the periodic drawing (including extra boundary vertices) represented as points in the $[0,1]^2$ square, sorted by increasing $x$ coordinate.
  Consider the vertical line $\ell_i$ with equation $x=x_i$ through each point $v_i$.
  For each vertical slab between $\ell_i$ and $\ell_{i+1}$, let $E_i = \{e_1,\ldots,e_k\}$ be the currently drawn edges intersecting the interior of the slab, sorted by increasing $y$ coordinate $e_1(x_i),\ldots,e_k(x_i)$ of their intersections with $\ell_i$.
  When a vertex $v$ of the graph lies on line $\ell_i$, order the edges of $E_i$ incident to $v$ by increasing slope (matching the bottom-to-top order of their anchor points).
  Likewise, for any vertex $v$ on $\ell_{i+1}$, order the edges of $E_{i+1}$ incident to $v$ by decreasing slope.

  Let $E^+_i\subseteq E_i$ be the set of edges properly crossing $\ell_i$, i.e., edges that are not incident to a vertex lying on $\ell_i$.
  In order to keep the grid size of the drawing small, we move the intersection of those edges with $\ell_i$ to $y$ coordinate $e'_j(x_i)$ in the new drawing. More precisely, let $y^-, y^+ \in[0,1]$ be the $y$ coordinates of vertices lying on $\ell_i$ directly below and above $e_j$, respectively, within the $[0,1]^2$ square if they exist; otherwise, set $y^-=0$ and/or $y^+=1$.
  Set $e'_j(x_i):= y^- + j(y^+-y^-)/M$.
  Because $y^-$ and $y^+$ are coordinates of the original drawing, they are of grid size $M$, that is, they are rationals of the form $z/M$; and because $j<M$, the order of the edge crossings on $\ell_i$ is preserved.
  Therefore the new coordinates will be valid for grid size a multiple of $M^2$.
  For edges not in $E^+_i$, we keep their position $e'_j(x_i):= e_j(x_i)$.

  Because the original edges are noncrossing, the $y$ coordinates $e_1(x_{i+1}),\ldots,e_k(x_{i+1})$ of the intersections of these edges with $\ell_{i+1}$ have the same order as along $\ell_{i}$.

  We now reroute the edges in $E_i$ between the anchor points on the left and right sides of the slab.
  For each edge $e_j \in E_i$, if $e'_j(x_i) \geq e'_j(x_{i+1})$, then we draw a horizontal edge from the anchor point to the point $(x_i + 10 j \epsilon, e'_j(x_i))$, then a vertical edge to $(x_i + 10 j \epsilon, e'_j(x_{i+1}))$, and finally a horizontal edge to the anchor point on the right side of the slab.
  If $e'_j(x_i) > e'_j(x_{i+1})$, then we draw a horizontal edge from the anchor point on the left side of the slab to the point $(x_{i+1} - 10 j \epsilon, e'_j(x_i))$, then a vertical edge to $(x_{i+1} - 10 j \epsilon, e'_j(x_{i+1}))$, and finally a horizontal edge to the anchor point on the right side of the slab.
  Note that the maximum number of edges intersecting any vertical slab is less than $M$, so none of those paths intersect.

  Finally, we connect the vertices to their used anchor points; see Figure \ref{fig:connectors}.
  We also remove the vertices we introduced along the boundary of the unit square and concatenate their two incident edges.
\end{proof}

\begin{figure}
  \centering
  \input{connectors}
  \caption{Connecting the vertices to their anchor points.}%
  \label{fig:connectors}
\end{figure}
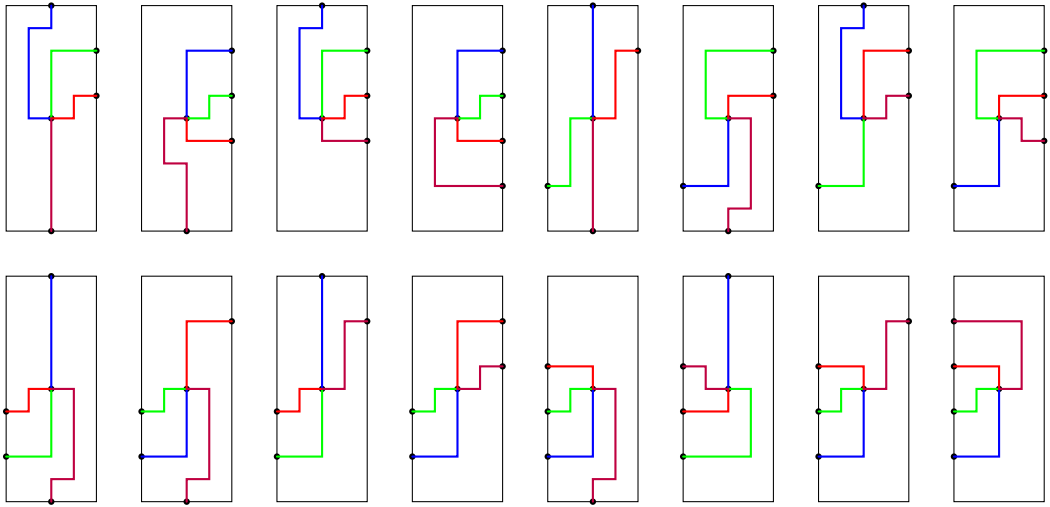

To simplify our reductions, we may assume that degree-3 vertices do not have an incident edge pointing down.

\begin{lemma}\label{lem:degree3}
  A planar periodic orthogonal graph drawing of grid size $M$ and maximum degree 3 can be transformed so that all vertices of degree 3 have the first segment of their incident edges going left, up, and right.
  Furthermore, each vertex can specify which incident edge should start by going left.
  The new grid size is $O(M)$.
\end{lemma}
\begin{proof}
  To achieve the first property, we refine the grid size by a factor 6, and replace the $6\times 6$ square around each vertex according to the diagram in Figure~\ref{fig:degree3}.
  \begin{figure}
    \centering
    \begin{tikzpicture}[scale=0.4, line cap=round]
      \draw[thick] (0,0) rectangle (6,6);
      \fill[gray] (3,6) circle (2pt);   
      \fill (6,3) circle (3pt);         
      \fill (3,0) circle (3pt);         
      \fill (0,3) circle (3pt);         
      \fill (3,3) circle (3pt);         
      \draw[blue , line width=1pt] (3,3) -- (4,3) -- (5,3) -- (6,3);
      \draw[green, line width=1pt] (3,3) -- (2,3) -- (2,2) -- (3,2) -- (3,1) -- (3,0);
      \draw[red  , line width=1pt] (3,3) -- (3,4) -- (2,4) -- (1,4) -- (1,3) -- (0,3);
    \end{tikzpicture}
    \begin{tikzpicture}[scale=0.4, line cap=round]
      \draw[thick] (0,0) rectangle (6,6);
      \fill (3,6) circle (3pt);         
      \fill[gray] (6,3) circle (2pt);   
      \fill (0,3) circle (3pt);         
      \fill (3,0) circle (3pt);         
      \fill (3,3) circle (3pt);         
      \draw[blue , line width=1pt] (3,3) -- (3,4) -- (3,5) -- (3,6);
      \draw[green, line width=1pt] (3,3) -- (2,3) -- (1,3) -- (0,3);
      \draw[red  , line width=1pt] (3,3) -- (4,3) -- (4,2) -- (3,2) -- (3,1) -- (3,0);
    \end{tikzpicture}
    \begin{tikzpicture}[scale=0.4, line cap=round]
      \draw[thick] (0,0) rectangle (6,6);
      \fill (3,6) circle (3pt);         
      \fill (6,3) circle (3pt);         
      \fill (0,3) circle (3pt);         
      \fill[gray] (3,0) circle (2pt);   
      \fill (3,3) circle (3pt);         
      \draw[blue , line width=1pt] (3,3) -- (3,4) -- (3,5) -- (3,6);
      \draw[green, line width=1pt] (3,3) -- (4,3) -- (5,3) -- (6,3);
      \draw[red  , line width=1pt] (3,3) -- (2,3) -- (1,3) -- (0,3);
    \end{tikzpicture}
    \begin{tikzpicture}[scale=0.4, line cap=round]
      \draw[thick] (0,0) rectangle (6,6);
      \fill (3,6) circle (3pt);         
      \fill (6,3) circle (3pt);         
      \fill[gray] (0,3) circle (2pt);   
      \fill (3,0) circle (3pt);         
      \fill (3,3) circle (3pt);         
      \draw[blue , line width=1pt] (3,3) -- (3,4) -- (3,5) -- (3,6);
      \draw[green, line width=1pt] (3,3) -- (4,3) -- (5,3) -- (6,3);
      \draw[red  , line width=1pt] (3,3) -- (2,3) -- (2,2) -- (3,2) -- (3,1) -- (3,0);
    \end{tikzpicture}
    \caption{Replacing the $6\times 6$ square around a degree-3 vertex.}%
    \label{fig:degree3}
  \end{figure}
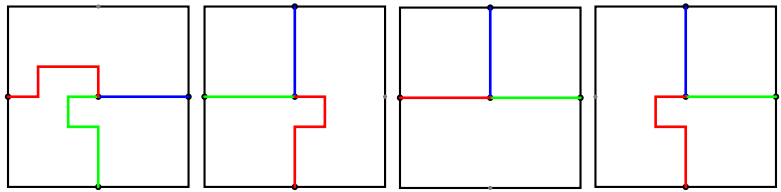
  To achieve the second property, the edges incident to a vertex can be cyclically rotated as shown in Figure~\ref{fig:rotation}.
  \begin{figure}
    \centering
    \begin{tikzpicture}[scale=0.4, line cap=round]
      \draw[thick] (0,0) rectangle (6,6);
    
      \draw[blue, line width=1pt] (6,3) -- (5,3) -- (5,2) -- (4,2) -- (3,2) -- (2,2) -- (2,3) -- (3,3);
      \draw[green, line width=1pt] (0,3) -- (1,3) -- (1,4) -- (2,4) -- (3,4) -- (3,3);
      \draw[red, line width=1pt] (3,6) -- (3,5) -- (4,5) -- (4,4) -- (4,3) -- (3,3);
      \fill (3,3) circle (3pt);  
      \fill (6,3) circle (3pt);  
      \fill (0,3) circle (3pt);  
      \fill (3,6) circle (3pt);  
    \end{tikzpicture}
    \caption{Rotating edges clockwise.}%
    \label{fig:rotation}  
  \end{figure}
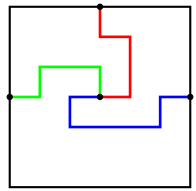
\end{proof}

\section{Periodic Problems and Reductions \label{sec:periodic-reductions}}

In this section, we give a chain of reductions to periodic problems
that establish PSPACE-completeness in 1D and
co-RE-completeness in 2D and higher.

\subsection{Wang Tiling}

For 1D periodic problems, our reduction chain starts from
the canonical PSPACE problem, polynomial-space Turing machine acceptance.
For 2D and higher periodic problems,
our reduction chain starts from a known co-RE-complete problem,
Wang's "domino problem".
An (unsigned) \defn{Wang tile} is a square with a \defn{glue} label on each edge. Given a finite set of Wang tiles, the \defn{domino problem} is to decide whether there exists a tiling of the plane with these tiles such that adjacent tiles have matching glues on their touching edges.
In 1966, Berger proved Wang tiling undecidable:

\begin{theorem}[{\cite{Berger1966}}] \label{thm:wang}
  Given a finite set of Wang tiles, it is co-RE-complete to determine
  whether they tile the plane, by translation only,
  with matching edge glues.
\end{theorem}

\subsection{Periodic CNF SAT}

The \defn{$d$D Periodic CNF SAT} problem is a $d$-dimensional infinite periodic version of Satisfiability.
The input consists of
\begin{enumerate}
\item a finite set of \defn{protovariables} $v_1,\ldots,v_n$, which get copied into a countable infinity of \defn{variables} $v_i^{\vec x}$ for $i \in \{1,\dots,n\}$ and $\vec x \in \mathbb Z^d$; and
\item a finite collection of \defn{protoclauses}, each a finite disjunction of \defn{literals} (a variable or its negation), of the form
   $$ \sigma_1 v_{i_1}^{\vec x_1} \lor \sigma_2 v_{i_2}^{\vec x_2} \lor \cdots \lor \sigma_c v_{i_c}^{\vec x_c},$$
   where $\sigma_i$ is the sign of the $i$th literal, which gets copied into a countable infinity of \defn{clauses}
   $$ \sigma_1 v_{i_1}^{\vec x_1+\vec \Delta} \lor \sigma_2 v_{i_2}^{\vec x_2+\vec \Delta} \lor \cdots \lor \sigma_c v_{i_c}^{\vec x_c+\vec \Delta}, \forall \vec \Delta\in \mathbb Z^d.$$
\end{enumerate}
As usual, the decision question is whether (infinitely many) variables can be assigned to satisfy all of the (infinitely many) clauses.
We call the instance \defn{local} if all the $\vec x_i$ in a protoclause pairwise differ by at most $1$ in $L_1$ distance,
and define \defn{local $d$D Periodic CNF SAT} to be the restriction of $d$D Periodic CNF SAT to local instances.

Equivalently,
a (local) periodic SAT instance can be represented by a (local) bipartite periodic graph of the same dimension: represent each protovariable $v_i$ by a protovertex, and for each clause
$$ \sigma_1 v_{i_1}^{\vec x_1} \lor \sigma_2 v_{i_2}^{\vec x_2} \lor \cdots \lor \sigma_c v_{i_c}^{\vec x_c},$$
create a protovertex $c$ to represent it, and add edges 
$(v_{i_j}^{\vec x_j}, c^{\vec 0})$ labeled $\sigma_j$ for all $j=1,\ldots,c$.

\begin{theorem}\label{thm:periodic-cnf-sat}
  For local Periodic CNF SAT,
  the 1D problem is PSPACE-complete and
  the 2D problem (and all higher dimensions, local or nonlocal) is co-RE-complete.
\end{theorem}



\begin{proof}
  \textbf{2D co-RE hardness} \cite{Freedman1998ksat} follows by reduction from the Wang tile problem of Theorem~\ref{thm:wang}, as in \cite{Freedman1998ksat}.
  For each color $i=1,\ldots,C$, create protovariables $s_i$ and $w_i$ representing whether the west and south edges of a tile are colored $i$.
  For each tile $j=1, \ldots, T$, create protovariables $t_i$ representing whether tile $i$ is used at a given position. Note that all these variables will be copied at each point $\vec x \in \mathbb Z^d$. As the colors of adjacent Wang tiles must match, the colors along the north and east edges of the tile at $\vec x$ are encoded in $s_{i}^{\vec x+(0,1)}$ and $w_{i}^{\vec x+(1,0)}$, respectively.
  The following protoclauses ensure each tile uses only one color along each border:
  \begin{align*}
    \neg s_{i}^{(0,0)} \lor \neg s_{j}^{(0,0)} \quad \forall i,j \in \{1,\ldots,C\},
    \\
    \neg w_{i}^{(0,0)} \lor \neg w_{j}^{(0,0)} \quad \forall i,j \in \{1,\ldots,C\}.
  \end{align*}
  The following clauses ensure each position chooses at least one tile:
  $$ \bigvee_{i=1}^T t_{i}^{(0,0)} \quad \forall i \in \{1,\ldots,C\}.$$
  Finally, each tile $i$ with colors $a_i,b_i,c_i,d_i$ along the south, west, north, and east edges, respectively, defines implications that ensure the colors match:
  \begin{align*}
    & \neg t_{i}^{(0,0)} \lor s_{a_i}^{(0,0)}; &
    & \neg t_{i}^{(0,0)} \lor w_{b_i}^{(0,0)}; &
    & \neg t_{i}^{(0,0)} \lor s_{c_i}^{(0,1)}; &
    & \neg t_{i}^{(0,0)} \lor w_{d_i}^{(1,0)}.
  \end{align*}
  Thus, the Periodic CNF SAT problem is satisfiable if and only if there exists a tiling of the plane with the given Wang tiles, and so the local (and thus nonlocal) 2D Periodic CNF SAT is co-RE-hard.

  \textbf{$d$D co-RE membership}: 
  We use a standard compactness argument.
  Define an order on the SAT variables by increasing distance from the origin:
  order vertex $v^{\vec x}$ by increasing $L_1$ norm $\|\vec x\|_1$,
  then lexicographically by $\vec x$, then by~$v$
  (according to a fixed order of the finite protovertex set~$V$).
  Define the potentially infinite binary tree of valid partial assignments,
  where a left branch at depth $i$ corresponds to assigning the $i$th variable to false, and a right branch corresponds to assigning it to true,
  and we omit nodes whenever the partial assignment violates a clause.
  We can construct this tree level by level, and detect whether it is infinite in co-RE.
  If the algorithm discovers that the tree is finite, then the SAT instance is unsatisfiable, and we return false.
  If the algorithm runs forever, then the tree is infinite.
  By the Weak K\"onig's Lemma (equivalent to the compactness of Cantor space $2^\omega$), an infinite binary tree contains an infinite path, which corresponds to an infinite satisfying assignment to all variables satisfying the entire formula.
  (Unlike K\"onig's Lemma for arbitrary trees, the Weak K\"onig's Lemma is provable in ZF: for $i=1,2,\dots$, assign the $i$th variable so that the remaining rooted subtree is infinite. See, e.g., \cite{Bauer-2006}.)

  \textbf{1D PSPACE membership} \cite{Orlin1981_STOC}:
  We give an NPSPACE algorithm for local 1D Periodic CNF SAT, which implies PSPACE membership by Savitch's Theorem.
  Suppose the 1D periodic CNF SAT problem is satisfiable, with variable values $v_i^z = \alpha_i^z$.
  By the Pigeonhole Principle, there are two values $x,y \in\mathbb Z$ where $|y-x| \leq 2^{n}$ for which $\alpha_i^x=\alpha_i^y$ for all $i \in \{1,\ldots,n\}$. 
  By repeating the sequence of variable assignments between $x$ and $y$, we can obtain another satisfying assignment $v_i^z=\beta_i^z$ which is periodic and of period $|y-x| \leq 2^n$.

  To build an NPSPACE algorithm for satisfiability, guess the assignment of variables at position $0$: $v_i^0 = \beta_i^0$ for all $i$. Then, for each step $z=1,2,\ldots,2^{|V|}$, guess the next assignment $v_i^z = \beta_i^z$ for all $i$, check that all clauses involving variables at positions $z$ and $z-1$ are satisfied, and then forget the assignment at position $z-1$. Because the 3SAT instance is local, all clauses will be checked by this process. If the first position's assignment $\beta_i^0$ ever gets repeated at a later $\beta_i^z$, then the instance is satisfiable.

  \textbf{1D PSPACE-hardness} \cite{Orlin1981_STOC}:
  Our reduction is based on a proof by Lance Fortnow described by Scott Aaronson about time-travel computing \cite[Section~8]{Aaronson-2005}.
  Consider any language $L$ in PSPACE.
  Let $M$ be a Turing machine that accepts $L$ and runs in space $s(n)$ where $n$ is the input size and $s(n)$ is a polynomial in $n$, and assume $s(n)\geq n$ by rounding up.
  Given an input $x$ of size $n$, we define a 1D periodic CNF SAT instance where the variables at position $i\in\mathbb Z$ represent:
  \begin{itemize}
  \item A step counter $c$ which represents the step of the Turing machine $M$ being simulated at position $i$. Because $\text{PSPACE} \subseteq \text{EXPTIME}$, $c$ can be represented in $\gamma \, s(n)$ bits for a constant $\gamma = \gamma(M)$, or $\gamma \, s(n)$ boolean protovariables $c_0,\ldots,c_{\gamma s(n)-1}$.
  \item The tape of $M$ at step $c$ using protovariables $t_{j,k}$ expressing that the symbol at position $j$ on the tape is $k$, for $j\in[1,s(n)$].
  \item The head position of the Turing machine $M$ at step $c$ using protovariables $h_j$ for the position $j\in[1,s(n)]$ on the tape.
  \item The state of the Turing machine $M$ at step $c$ using protovariables $q_k$ for the state $k\in[1,|Q|]$ where $Q$ is the set of states of $M$.
  \end{itemize}
  As in Garey and Johnson's proof of Cook's Theorem \cite[Theorem~2.1, p.~39]{Garey-Johnson-1979},
  we can construct a set of clauses using variables $t_j^i$, $h_j^i$, $q_k^i$, $t_j^{i+1}$, $h_j^{i+1}$, $q_k^{i+1}$ to ensure that the Turing machine $M$ simulates correctly.
  To this we add clauses expressing that the counter $c$ increments at each non-accepting step, without overflow (i.e., failing to increment $c$ when it is all $1$s).
  For any accepting state at position $i$,
  \begin{itemize}
  \item $c_j^{i+1}$ is false (i.e., $c$ goes back to 0);
  \item $t_j^{i+1}$ is set to the input instance $x$;
  \item the head position $h_1^{i+1}$ is true, and $h_j^{i+1}$ is false for all $j> 1$; and
  \item $q_0^{i+1}$ is true and $q_j^{i+1}$ is false for all $j> 1$.
  \end{itemize}
  That is, the Turing machine at position $i+1$ is set to its initial state.
  This ensures that, if $M$ accepts input $x$, then there exists a periodic CNF SAT assignment that satisfies the clauses at all positions $i\in\mathbb Z$.
  Other the other hand, any satisfying assignment of the periodic CNF SAT instance must contain a position $i$ where the Turing machine $M$ is in an accepting state because the counter $c$ cannot overflow, and the clauses ensure that the Turing machine at position $i+1$ is set to its initial state, and then the clauses simulate the Turing machine and reach another accepting state (because again the counter $c$ cannot overflow), which in turn implies that $x$ is in the language~$L$.
\end{proof}

\subsection{Periodic 3SAT}

The \defn{Periodic 3SAT} problem is a special case of Periodic CNF SAT, where each clause has at most three literals.

\begin{theorem}\label{thm:periodic-3sat}
  Local Periodic 3SAT is co-RE-complete in 2D and PSPACE-complete in 1D.
\end{theorem}
\begin{proof}
  We use the usual reduction from CNF SAT to 3SAT \cite{Cook-1971},
  but applied to local Periodic CNF SAT from Theorem~\ref{thm:periodic-cnf-sat}.
  For each clause of the SAT instance with more than three literals,
  $$ \sigma_1 v_{i_1}^{\vec x_1} \lor \sigma_2 v_{i_2}^{\vec x_2} \lor \sigma_3 v_{i_3}^{\vec x_3} \lor \cdots \lor \sigma_c v_{i_c}^{\vec x_c},$$
  create a new protovariable $z$ and replace the clause with
  $$ z^{\vec x_1} \lor \sigma_1 v_{i_1}^{\vec x_1} \lor \sigma_2 v_{i_2}^{\vec x_2},$$
  and
  $$\neg z^{\vec x_1} \lor \sigma_3 v_{i_3}^{\vec x_3} \lor \cdots \lor \sigma_c v_{i_c}^{\vec x_c}.$$
  Continue iteratively until no clause with more than three literals remains.
  The resulting instance is satisfiable if and only if the original instance is.
  If the original instance is local, then so is the modified instance.
\end{proof}

\subsection{Periodic 3SAT-3}
The \defn{Periodic 3SAT-3} problem is a special case of Periodic 3SAT, where each variable occurs in at most three clauses.

\begin{theorem}\label{thm:periodic-3sat-3}
  Local Periodic 3SAT-3 is co-RE-complete in 2D and PSPACE-complete in 1D.
\end{theorem}
\begin{proof}
  We apply the standard cycle reduction from 3SAT to 3SAT-3 \cite{tovey1984simplified},
  but reducing from local Periodic 3SAT of Theorem~\ref{thm:periodic-3sat}.
  For each variable $x$ that occurs $d$ times,
  we replace $x$ by $d$ new variables $x_1, x_2, \dots, x_d$,
  and add $d$ clauses $\neg x_i \lor x_{i+1}$
  where indices are modulo~$d$.
  These clauses are equivalent to the implication chain
  $x_1 \to x_2 \to \cdots \to x_d \to x_1$,
  which forces all $x_i$ to have the same truth value.
  Figure~\ref{fig:3sat-3} shows the case $d=8$.
  \begin{figure}
    \centering
    $
    \vcenter{\hbox{%
      \includegraphics
         [scale=0.35,page=1]%
         {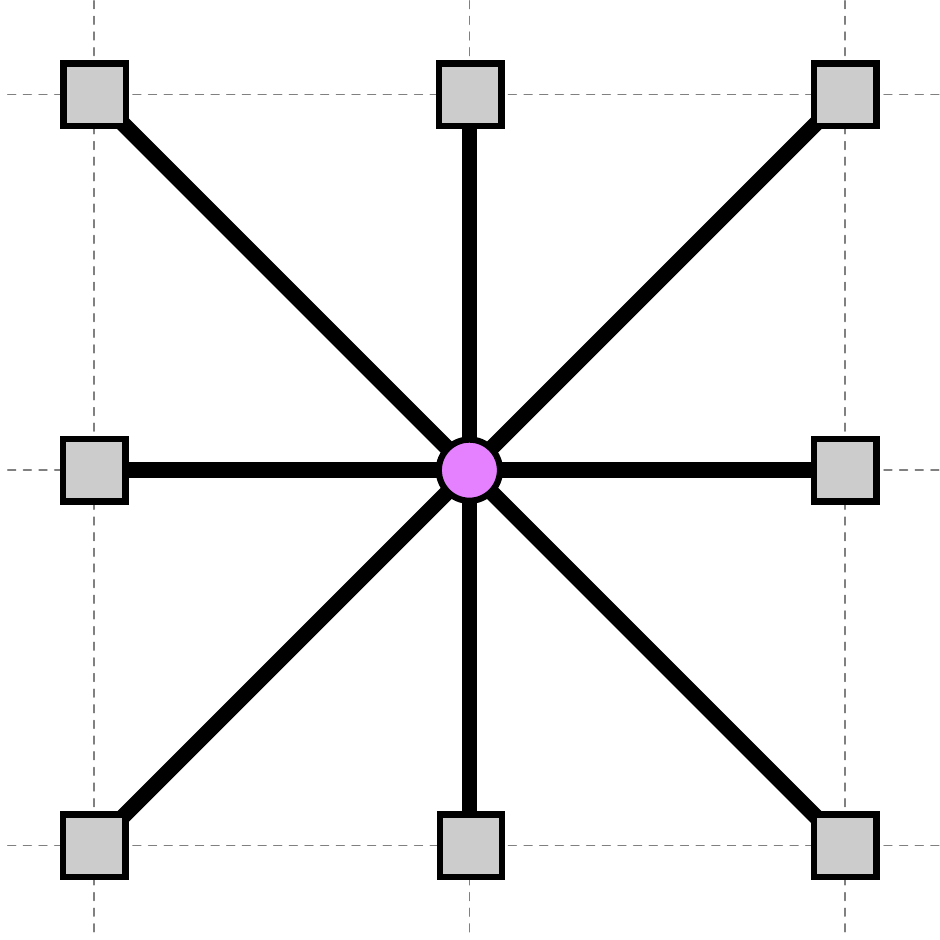}}}%
    \quad \to \quad
    \vcenter{\hbox{%
      \includegraphics
         [scale=0.35,page=2]%
         {figs/3sat-3}}}%
    $
    \caption{Reduction from 3SAT to 3SAT-3: splitting a $d$-occurrence variable into $d$ 3-occurrence variables. (Here, $d=8$.) Circles represent variables and squares represent clauses; edge color denotes sign.}%
    \label{fig:3sat-3}
  \end{figure}
\end{proof}

\subsection{Periodic Planar 3SAT}

The \defn{Periodic Planar 3SAT} problem is a special case of Periodic 3SAT,
where the periodic graph is guaranteed to be planar.

\begin{theorem}\label{thm:planar-3sat}
  Local Periodic Planar 3SAT is co-RE-complete in 2D and PSPACE-complete in 1D, even when restricting the grid size of the drawing to be polynomial.
  The same is true of local Periodic Planar 3SAT-3.
\end{theorem}
\begin{proof}
  We reduce from local Periodic 3SAT-3 from Theorem~\ref{thm:periodic-3sat-3}.
  Applying Theorem~\ref{thm:orthocrossing-drawing} to the graph representing the SAT instance, which has maximum degree 3, we obtain an orthogonal orthocrossing local periodic drawing of grid size $O(|V|+|E|)$.
  We view this drawing within a single $[0,1]^2$ square; because the drawing is local, this involves at most two translations of each edge.

  Decompose every edge into a sequence of horizontal and vertical unit segments, and add a variable vertex at the midpoint of each segment. 
  Then each original gridpoint is either
  \begin{enumerate}
  \item a clause vertex with up to three incident edges;
  \item a variable vertex with up to three incident edges;
  \item a straight or bent path connecting two variable vertices; or
  \item an orthocrossing between two edges.
  \end{enumerate}
  In each case, we replace the gridpoint and its connections to neighboring half-grid variables with a gadget.
  In the first case, we leave the connections from clause to variables as they are.
  In the second case, we use the gadget in Figure~\ref{fig:3sat-duplicator} (or a subset for lower degree), which adds up to three duplicators ($a \to b$ and $b \to a$ to force $a=b$) to connect the variables together while forcing their value to be equal.
  In the third case, we use a subset of Figure~\ref{fig:3sat-duplicator} with only two of the incident edges (and no original variable in the center).
  In the fourth case, we use the crossover gadget 
  in Figure~\ref{fig:3sat-crossover}, which is exactly the gadget of Lichtenstein \cite{Lichtenstein-1982} but drawn on a grid.
  As argued by Lichtenstein \cite{Lichtenstein-1982}, this gadget preserves satisfiability while removing crossings.

  \begin{figure}
    \centering
    \begin{subfigure}{\textwidth}%
      \centering
      \def\u{_} 
      $
      \vcenter{\hbox{%
        \includegraphics
           [scale=0.25,page=1]%
           {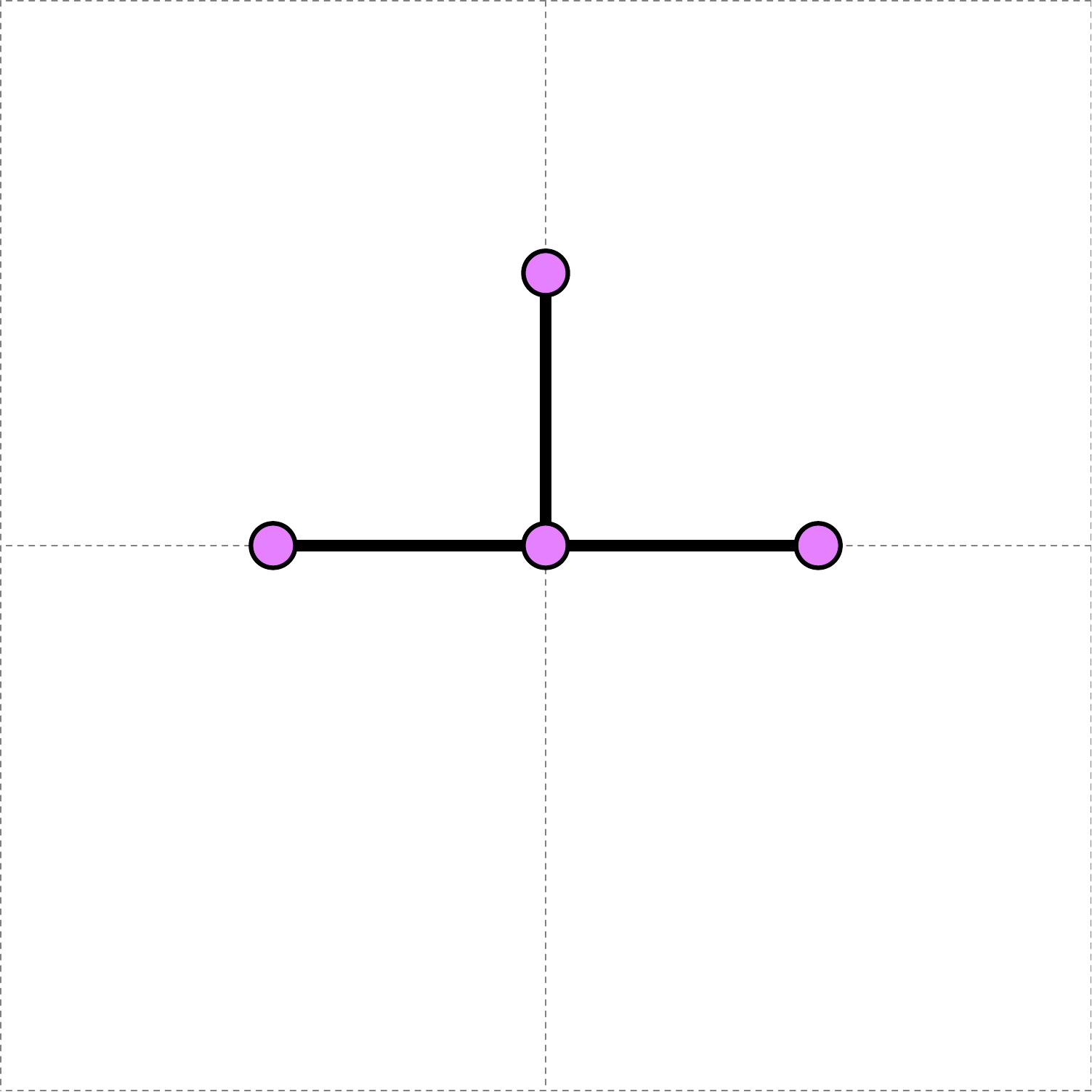}}}%
      \quad \to \quad
      \vcenter{\hbox{%
        \includegraphics
           [scale=0.25,page=2]%
           {figs/3sat_duplicator}}}%
      $
      \caption{Variable gadget. Subsets of this gadget also implement straight and turn gridpoints.}%
      \label{fig:3sat-duplicator}
    \end{subfigure}
    
    \medskip
    
    \begin{subfigure}{\textwidth}%
      \centering
      \def\u{_} 
      $
      \vcenter{\hbox{%
        \includegraphics
           [scale=0.2,page=1]%
           {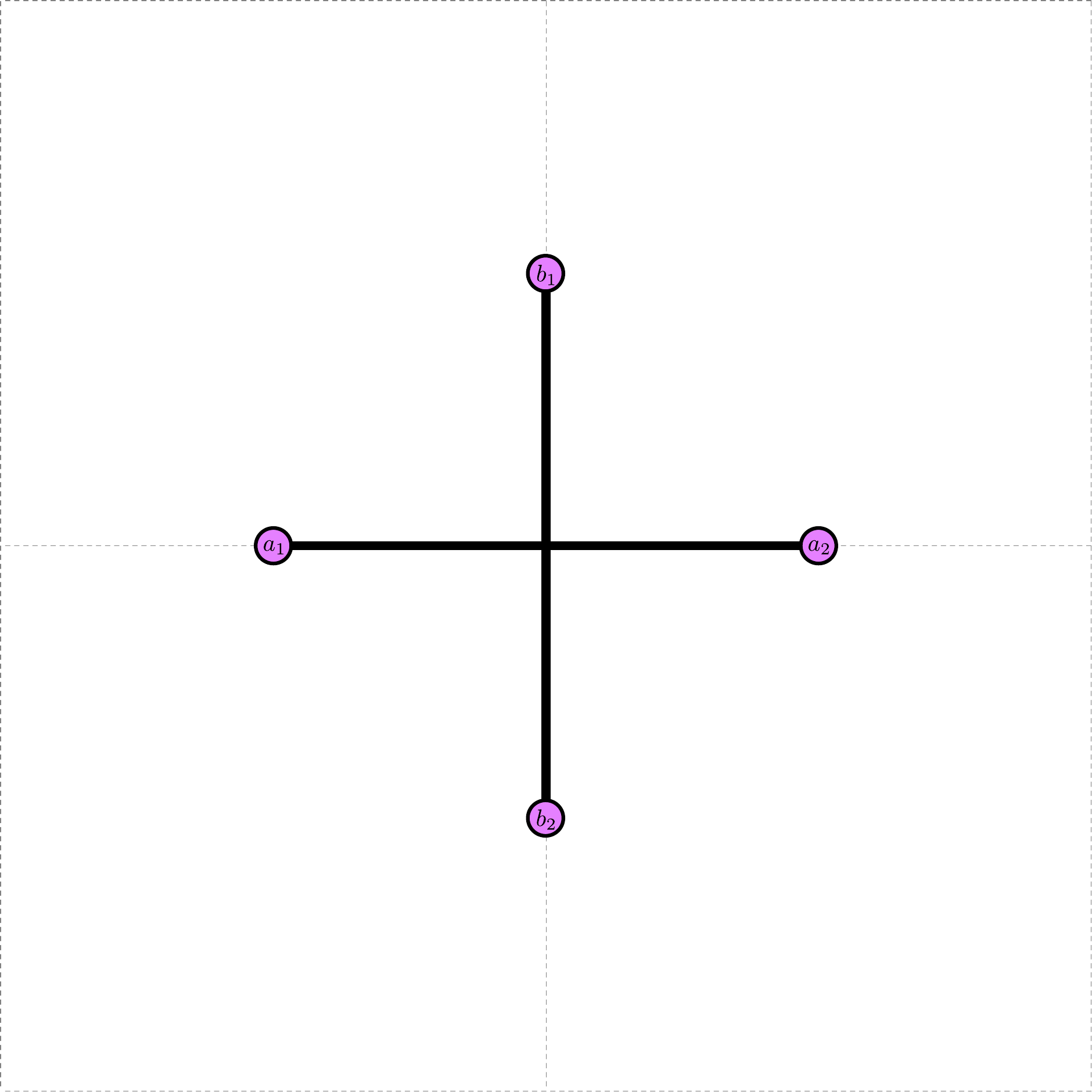}}}%
      \quad \to \quad
      \vcenter{\hbox{%
        \includegraphics
           [scale=0.35,page=2]%
           {figs/3sat_crossover}}}%
      $
      \caption{Crossover gadget, based on \cite[Figs.~4 and 5]{Lichtenstein-1982}.}%
      \label{fig:3sat-crossover}
    \end{subfigure}
    \caption{Reduction from 3SAT-3 to Planar 3SAT.}%
  \end{figure}

  As shown by the figures, the gadgets can be drawn in a grid of constant size, so the grid size of the resulting drawing is $O(|V|+|E|)$.
  (Specifically, Figure~\ref{fig:3sat-duplicator} refines by a factor of $4$, and Figure~\ref{fig:3sat-crossover} refines by a factor of $10$, so we use the least common multiple of~$20$.)

  The gadgets of Figures~\ref{fig:3sat-duplicator} and \ref{fig:3sat-crossover}
  introduce variable vertices of degree larger than~$3$.
  We can reduce to Periodic Planar 3SAT-3 by using
  the cycle construction of Theorem~\ref{thm:periodic-3sat-3},
  while taking care to use a constant factor of additional grid size.
  Specifically, all variable vertices have incident edges at angles that are integer multiples of $45^\circ$, and Figure~\ref{fig:3sat-3} shows how to solve this case while refining the grid by a factor of~$4$.
\end{proof}

\subsection{Periodic Planar 1-in-3SAT}

The \defn{Periodic Planar 1-in-3SAT} problem is a modification of Periodic planar 3SAT, where for each clause, exactly one literal must be true.

\begin{theorem}\label{thm:planar-1-in-3sat}
  Local Periodic Planar 1-in-3SAT is co-RE-complete in 2D and PSPACE-complete in 1D, even when restricting the grid size of the drawing to be polynomial.
\end{theorem}
\begin{proof}
  We follow the proof of Dyer and Freeze \cite{Dyer-Freeze-1986} that 1-in-3SAT is NP-hard, and reduce from Periodic Planar 3SAT. Given a planar periodic 3SAT instance, we replace each clause with the gadget of Figure~\ref{fig:3sat-to-1in3sat}, which ensures that exactly one of the three literals is true. The gadget is drawn on a grid of constant size, so the grid size of the resulting drawing is only a constant factor larger.
\end{proof}

\begin{figure}
  \centering
  \def\u{_} 
  $
  \vcenter{\hbox{%
    \includegraphics
       [scale=0.4,page=1]%
       {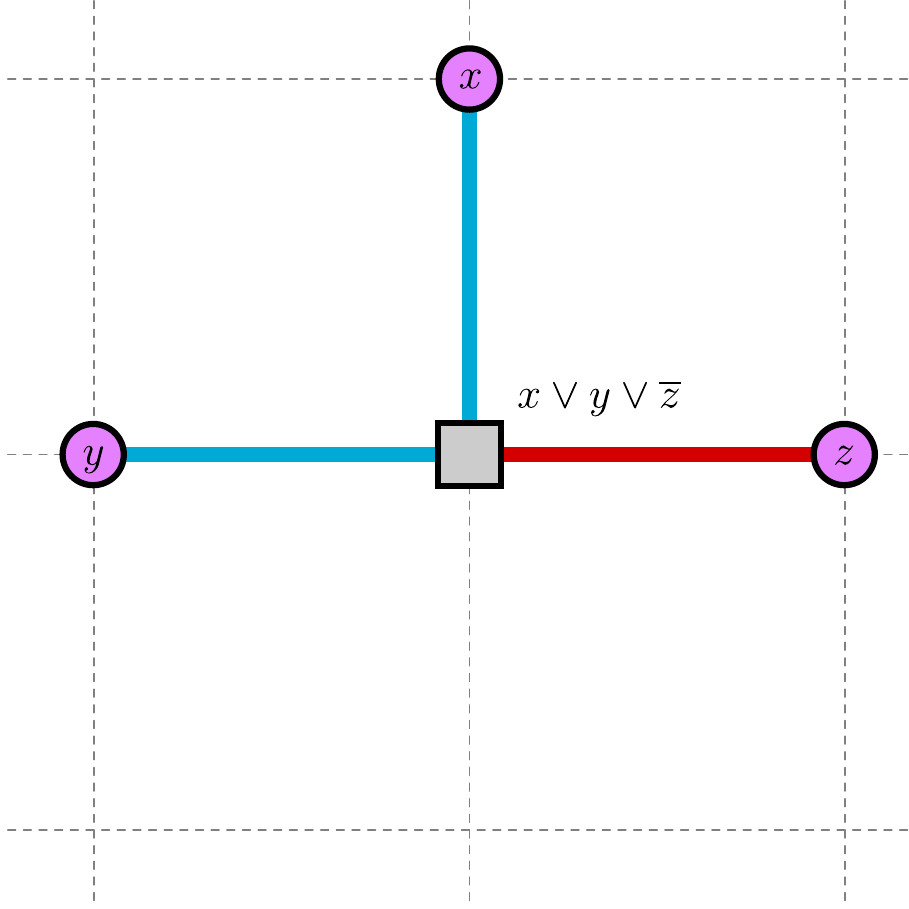}}}%
  \quad \to \quad
  \vcenter{\hbox{%
    \includegraphics
       [scale=0.4,page=2]%
       {figs/3sat_to_1in3sat}}}%
  $
  \caption{Reduction from 3SAT to 1-in-3SAT.
    Based on \cite[Figure 1]{Dyer-Freeze-1986}.}%
  \label{fig:3sat-to-1in3sat}
\end{figure}


\subsection{Periodic Planar 3DM}

In \defn{3DM (3-Dimensional Matching)}, we are given three disjoint sets $R,G,B$ (representing colors Red, Green, and Blue) and a set $T \subseteq R \times G \times B$ of trichromatic triples. A (perfect) \defn{3D matching} is a subset $M \subseteq T$ of trichromatic triples that cover every element $x \in R \cup G \cup B$ exactly once, i.e., exactly one $m \in M$ contains $x$.

A 3DM instance can be represented by a bipartite graph, where the elements of $R$, $G$, and $B$ are vertices colored red, green, and blue; and every trichromatic triple $t=(r,g,b)\in T$ is a blank vertex $t$ with an edge to vertices $r$, $g$, and $b$. The problem is then to determine whether a subset $M \subseteq T$ of the blank vertices can be selected such that each colored vertex has exactly one blank neighbor in the set.

The (local) \defn{Periodic 3DM} problem is the generalization of 3DM to periodic input bipartite graphs, and \defn{Planar Periodic 3DM} further restricts the periodic bipartite graph to be planar.

\begin{theorem}\label{thm:planar-3dm}
  Local Periodic Planar 3DM is co-RE-complete in 2D and PSPACE-complete in 1D.
\end{theorem}
\begin{proof}
  We follow Dyer and Freeze's NP-hardness reduction
  from (finite) Planar 1-in-3SAT to (finite) Planar 3DM \cite{Dyer-Freeze-1986}.
  Figure~\ref{fig:3dm-gadgets} gives the gadgets that locally replace each variable and clause, which are drawn on a constant-size grid.
\end{proof}

\begin{figure}
  \centering
  \subcaptionbox
     {Variable gadget.}%
     {\includegraphics[scale=0.25]{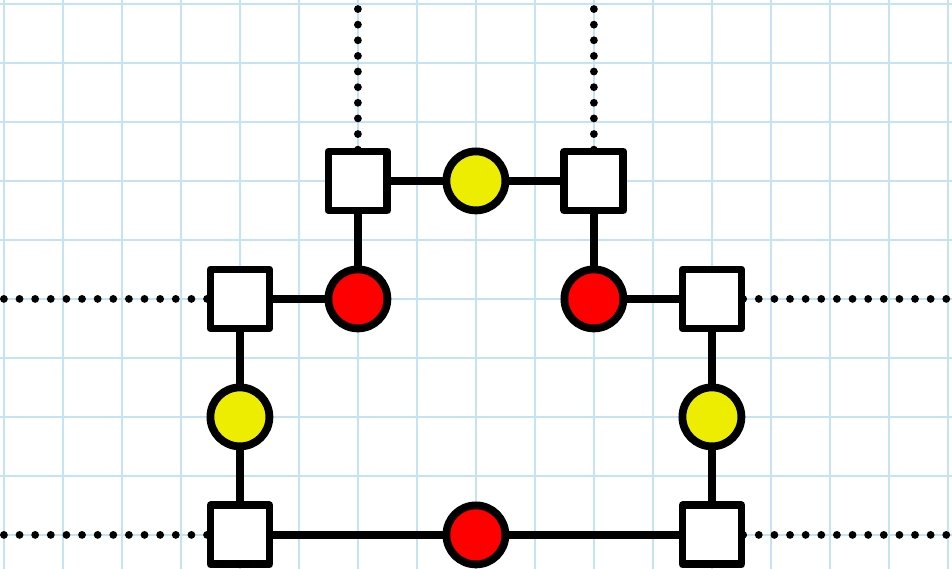}}%
  \hfil
  \subcaptionbox
     {Clause gadget.}%
     {\includegraphics[scale=0.25]{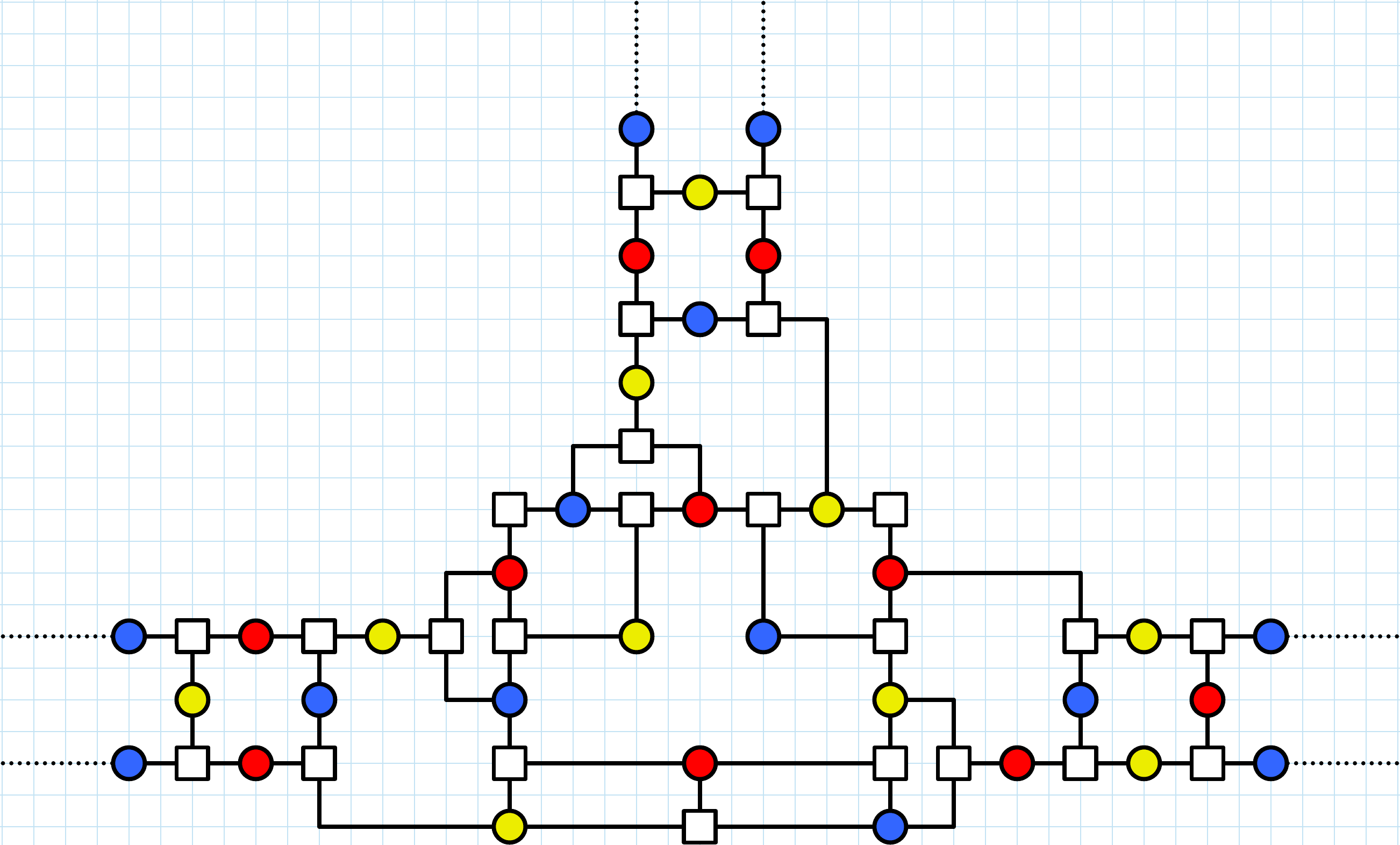}}%
  \caption{%
    Gadgets from Dyer and Freeze's reduction from Planar 1-in-3SAT to Planar 3DM.  Based on \cite[Figures 3, 4, and 5]{Dyer-Freeze-1986}.}%
  \label{fig:3dm-gadgets}
\end{figure}


\subsection{Periodic Planar Trichromatic Graph Orientation}

In \defn{Planar Trichromatic Graph Orientation}, we are given an undirected graph with vertex set $V = T \cup Z$, where edges are colored red, green, or blue.
Vertices in $T$ have exactly one incident edge of each color (trichromatic), and vertices in $Z$ are incident to edges of exactly one color (monochromatic).
The goal of the problem is to orient (direct) the edges so that the vertices in $T$ have indegree 0 or 3, and vertices in $Z$ have indegree exactly~1.
We can assume that the graph is 3-regular: vertices in $T$ are of degree 3 by definition; vertices in $Z$ of degree 2 can be contracted; and vertices in $Z$ of degree larger than 3 can be decomposed into a binary tree.
Thus we call vertices in $T$ \defn{0-or-3-in-3} and vertices in $Z$ \defn{1-in-3}.
This problem is a restricted version of the 1-in-3 Graph Orientation problem of Horiyama et al.~\cite{Horiyama-Ito-Nakatsuka-Suzuki-Uehara-2017}.
(Specifically, we add the planarity and trichromatic constraints.)

The \defn{Periodic Planar Trichromatic Graph Orientation} problem is one for which the underlying graph is infinite, periodic, and planar.

\begin{theorem}\label{thm:planar-graph-orientation}
  Periodic Planar Trichromatic Graph Orientation is co-RE-complete in 2D and PSPACE-complete in 1D.
\end{theorem}

\begin{proof}
  We reduce from Periodic Planar 3DM. Given a planar drawing of a planar periodic 3DM instance, we transform it into a planar periodic trichromatic graph instance. The graph and its drawing are identical, except that the vertices are uncolored. Vertices in $T$ in the 3DM instance become the vertices in $T$ in the trichromatic graph orientation instance, and vertices in $Z$ are the vertices in $R \cup G \cup B$ in the 3DM instance. Edges incident to $R$ in 3DM are colored red, edges incident to $G$ are colored green, and edges incident to $B$ are colored blue.
  Because all edges connect a colored vertex to a blank vertex in 3DM, the edge colors are well-defined in the trichromatic graph orientation instance.

  A solution to the 3DM instance can be transformed into a solution of the trichromatic graph orientation instance as follows: for each vertex $t \in M\subseteq T$, orient its three incident edges towards it; and for each vertex $t \in T\setminus M$, orient its incident edges away from it. This orientation satisfies the conditions of the trichromatic graph orientation problem if and only if $M$ is a perfect matching in the 3DM instance.

  Conversely, given a solution to the trichromatic graph orientation instance, we can extract a perfect matching for the 3DM instance by collecting all vertices $t \in T$ of indegree 3. The indegree-1 condition for vertices in $Z$ ensures that every colored vertex is covered exactly once.
\end{proof}

\section{Periodic Graph Algorithms \label{sec:periodic-algs}}

In this section, we give polynomial-time algorithms
for some $d$D periodic graph algorithms:
tractable cases of SAT
(Sections~\ref{sec:2sat} and~\ref{sec:horn})
and bipartite perfect matching
(Section~\ref{sec:matching}).

\subsection{2SAT and Reachability \label{sec:2sat}}

The \defn{Periodic 2SAT} problem is the special case of Periodic CNF SAT where each clause has at most two literals.  The 1-dimensional version of this problem was shown to be solvable in polynomial time by \cite{MaratheHuntStearnsRosenkrantz1995}, while the 2-dimensional version was shown to be decidable by \cite{Freedman1998ksat}.

Here we observe that existing results \cite{Freedman1998ksat, HoftingWanke1993}
imply $d$-dimensional Periodic 2SAT can be solved in polynomial time:

\begin{theorem}\label{thm:periodic-2sat}
  Periodic 2SAT can be solved in polynomial time in $d$D.
\end{theorem}
\begin{proof}
  Freedman \cite{Freedman1998ksat} reduces this problem to the \defn{Periodic Reachability} problem: given a directed periodic graph $\PER{G} = (\PER{V}, \PER{E})$ and two vertices $u, v \in \PER{V}$,
  determine whether there exists a path from $u$ to $v$.
  Their reduction is as follows.
  Given a $d$-dimensional Periodic 2SAT instance, construct a directed periodic graph which has a vertex for each literal, and edges $\overline{x} \to y$ and $\overline{y} \to x$ for each width-2 clause $x \lor y$, where $\overline{x}$ denotes the negation of the literal $x$.
  Then a solution to the $d$-dimensional Periodic 2SAT instance exists if and only if, for every variable $x$, at most one of the following is true:
  \begin{itemize}
  \item there is a path from $x$ to $\overline{x}$, or a path from $x$ to $\bot$;
  \item there is a path from $\overline{x}$ to $x$, or a path from $\overline{x}$ to $\bot$;
  \end{itemize}
  where $\bot$ denotes any literal $y$ such that there exists a width-1 clause $\overline{y}$.
  Since the instance is periodic, this condition needs only to be checked once for each protovariable.

  Hofting and Wanke \cite{HoftingWanke1993} give a polynomial-time algorithm for Periodic Reachability in $d$-dimensional directed local periodic graphs.
  It follows that $d$-dimensional local Periodic 2SAT can be solved in polynomial time as well.
\end{proof}

\subsection{(Dual) Horn SAT \label{sec:horn}}

The \defn{Periodic Horn SAT} problem is the special case of Periodic CNF SAT where every clause has at most one positive literal.
The \defn{Periodic Dual Horn SAT} problem is the opposite case where every clause has at most one negative literal.
Marathe et al.~\cite{MaratheHuntRosenkrantzStearns1998} gave a polynomial-time algorithm for the 1D and 2D versions of these problems.
Here we observe that this algorithm works independent of the number of dimensions and can be made to run in linear time.

\begin{theorem}\label{thm:periodic-horn}
  Periodic Horn SAT and Periodic Dual Horn SAT
  can be solved in linear time in $d$D.
  Furthermore, all satisfiable instances have periodic solutions with period 1.
\end{theorem}

\begin{proof}
  Given an instance of Periodic Horn SAT, we can construct a finite Horn SAT instance by identifying all the variables corresponding to the same protovariable.
  A solution to this Horn SAT instance is equivalent to a 1-periodic solution to the original Periodic Horn SAT problem.
  Because Horn SAT can be solved in linear time \cite{DowlingGallier1984}, it remains only to show that a solution to the Periodic Horn SAT instance exists if and only if a 1-periodic solution exists.

  A Horn clause with a positive literal, such as $x \lor \neg y \lor \neg z$, can be viewed as an inference rule
  "Given $y$ and $z$, derive $x$".
  Similarly, a Horn clause with only negative literals, such as $\neg y \lor \neg z$, can be viewed as the rule "Given $y$ and $z$, derive $\bot$".
  Call a variable \defn{derivable} if it can be obtained by finitely many applications of these rules.
  A solution to a (possibly infinite) Horn SAT instance exists if and only if $\bot$ is not derivable,
  and one such solution is obtained by setting the derivable variables to true and all other variables to false.
  For Periodic Horn SAT, if a variable is derivable, then so are all variables corresponding to the same protovariable;
  in other words, the set of derivable variables is 1-periodic.
  Thus a solution to a Periodic Horn SAT instance exists if and only if a 1-periodic solution exists.

  The same applies to Periodic Dual Horn SAT by negation of all variables.
\end{proof}

\subsection{Periodic Perfect Matching \label{sec:matching}}
A \defn{matching} in an undirected graph $G=(V,E)$ is a set $M\subseteq E$ of edges that are disjoint. A matching is \defn{perfect} if every vertex of $V$ belongs to an edge in $M$.
For a $d$-dimensional periodic graph $\PER{G}=(\PER{V},\PER{E})$, a matching $M\subseteq\PER{E}$ is \defn{periodic} if there are $d$ independent vectors $\vec \Delta_1,\ldots, \vec \Delta_d \in\mathbb Z^d$ such that an edge $(u^{\vec x},v^{\vec y})\in M$ if and only if $(u^{\vec x+\vec\Delta_i},v^{\vec y+\vec\Delta_i})\in M$ for all $i \in 1,\ldots,d$, or equivalently, if the subgraph $\PER{G}_M = (\PER{V},M)$ is a periodic graph (possibly with a different period than $\PER{G}$). If the period of $\PER{G}$ and $\PER{G}_M$ is the same, i.e., $\Delta_i$ is the unit vector along the $i$th axis, then we say the matching is \defn{1-periodic}.

For a (non-perfect) matching $M$, a vertex of $\PER{V}$ is \defn{free} if it is not in any edge of $M$. An \defn{alternating walk} is a path in $\PER{G}$ alternating edges not in $M$ and edges in $M$. In the discussion below, we allow alternating walks to be non-simple, i.e., to visit edges and vertices several times. When it is simple, we call it an \defn{alternating path}. 
If an alternating path $P$ is simple and joins two free vertices, then the matching $M$ can be augmented by taking the symmetric difference $M\oplus P$ between $M$ and $P$. In that case, $P$ is called an \defn{augmenting path}. 
Note that two alternating walks $P$ and $Q$ can never properly intersect: If they meet at some internal vertex $v$, then because exactly one of the edges incident to $v$ is in $M$, that edge must be in both $P$ and $Q$ and so $P$ and $Q$ share some subpath containing $v$.
The \defn{diameter} of a path or walk is the smallest $D$ such that for all $u^{\vec x}$ and $v^{\vec y}$ in $P$, $\|\vec x - \vec y\|_\infty \leq D$.

\begin{lemma}\label{lem:shortaugmenting}
  Let $\PER{G}=(\PER{V},\PER{E})$ be a $d$-dimensional local periodic graph, and let $M$ be a 1-periodic matching. If $\PER{G}$ admits a perfect matching, then either $M$ is perfect, or there exists an augmenting path of diameter $2d|E|$ starting from every free vertex.
\end{lemma}
\begin{proof}
  Let $M^{*}$ be a perfect matching for $\PER{G}$, not necessarily periodic. Consider the symmetric difference $M \oplus M^{*}$.
  This is a collection of vertex-disjoint cycles and paths (possibly infinite), where each path starts and ends at free vertices with respect to $M$, every free vertex is covered by a path (because $M^{*}$ is perfect), and each path is an alternating path with respect to $M$. Let $\cal P$ be this collection of alternating paths. 
  Consider the hypercube $H$ containing lattice points in $[1,K]^d \cap \mathbb Z^d$, and define $\PER{V}_H$ to be the set of vertices $v^{\vec x}$ with $\vec x \in H$. 
  By symmetry, if $M$ is not perfect, there is one free vertex in $\PER{G}$, and thus there are $K^d$ copies of that free vertex in $\PER{V}_H$, and each is the endpoint of an alternating path in $\cal P$. Let $\cal P'\subseteq \cal P$ be the paths of $\cal P$ with an endpoint at a copy of that free vertex in $\PER{V}_H$.  
  Suppose that none of the paths in $\cal P'$ are fully contained in $\PER{V}_H$.
  Because the paths are vertex-disjoint, all $K^d$ paths exit $\PER{V}_H$.
  The boundary of $H$ is connected to $2dK^{d-1}$ lattice points outside of $H$. Each path in $\cal P'$ must exit $\PER{V}_H$ into one of them. By the Pigeonhole Principle, one of themm is entered by at least $\frac{K^d}{2dK^{d-1}} = K/2d$ edges. As $\PER{G}$ is local, we have 
  $K/2d \leq |E|$, or $K \leq 2d|E|$.
  Therefore, if $K = 2d|E|+1$, then one of the paths in $\cal P'$ must be contained in $\PER{V}_H$, that is, its diameter is at most $2d|E|$.
\end{proof}

We define a \defn{bipartite periodic graph} $\PER{G}=(\PER{V},\PER{E})$ to satisfy $V=R\cup B$ and every edge $(u^x,v^y) \in \PER{E}$ has one protovertex (say, $u$) in $R$ and the other protovertex (say, $v$) in $B$.
That is, the 2-coloring of the bipartite graph is preserved by the periodicity.
If we had a connected periodic graph that is a bipartite graph, we can modify it into a bipartite periodic graph by doubling the period in each dimension. 


\begin{lemma}
  If a periodic graph $\PER{G}$ is bipartite and connected, then its bipartition is 2-periodic.
\end{lemma}
\begin{proof}
  Consider the 2-coloring induced by the bipartition of $\PER{G}$, which is unique up to renaming of color classes, as the graph $\PER{G}$ is connected.
  For any vertex $v^{\vec x}$ and $\vec{\Delta}\in\mathbb Z^d$, consider a path $P$ between $v^{\vec x}$ and $v^{\vec x + \vec\Delta}$, and construct a path $P'$ from $v^{\vec x}$ to $v^{\vec x + 2\vec\Delta}$ by concatenating $P$ and $P+\vec\Delta$, its translation by $\vec\Delta$. By construction, $P'$ is of even length, and thus $v^{x+2\Delta}$ and $v^{\vec x}$ are in the same color class.
\end{proof}

Note that this lemma does not hold if the graph is not connected. For example, take a 1D periodic graph with one protovertex $v$ and one protoedge $(v^{0}, v^{p})$ with $p$ prime, the periodicity of any 2-coloring of this graph is at least $2p$. Locality does not help either as each edge can be split into $p$ local edges. For instance, let $V=\{v_0,\ldots,v_{p-1}\}$ and $E=\{(v_i^{0},v_{i+1}^{1})\mid i=0,\ldots,p-1\}\cup\{(v_{p-1}^0,v_0^1)\}$, with $p$ prime. The local 1D periodic graph $\PER{G}=(\PER{V},\PER{E})$ is bipartite, but the periodicity of any 2-coloring of this graph is again at least $2p$.

\begin{lemma}\label{lem:bipartiteaugmenting}
  Let $\PER{G}=(\PER{V},\PER{E})$ be a $d$-dimensional bipartite periodic graph, and let $M$ be a 1-periodic matching.
  If $\PER{G}$ admits a perfect matching, then either $M$ is perfect, or there exists an augmenting path of length less than $|V|$, where each protovertex appears at most once.
\end{lemma}

\begin{proof}
  By Lemma \ref{lem:shortaugmenting}, if $\PER{G}$ admits a perfect matching and $M$ is imperfect, then there is a finite augmenting path, that is, a finite simple alternating walk joining two free vertices.
  Take a \emph{shortest} alternating walk $P=(v=v_{i_1}^{\vec x_1},v_{i_2}^{\vec x_2},\ldots,v_{i_k}^{\vec x_k})$
  that joins two free vertices in $\PER{G}$
  but where we allow vertices and edges to repeat (so it may not be augmenting).
  Because finite alternating walks have odd length and $\PER{G}$ is bipartite periodic,
  one endpoint must be in $\PER{R}$ and other endpoint must be in $\PER{B}$.
  Assume without loss of generality that $v_{i_1}\in B$ and $v_{i_k}\in R$.
  Suppose for contradiction that some protovertex appears more than once in $P$.
  Then we claim we can construct a shorter alternating walk, contradicting that $P$ was shortest.
  Let $j$ be the smallest integer for which protovertex $v_{i_j}$ appears more than once in $P$, and let $j' > j$ be the index of the second appearance $v_{i_{j'}}^{\vec x_{j'}}$ of the same protovertex $v_{i_j} = v_{i_{j'}}$. Then the walk $P$ and its translation $P+{\vec \Delta}$ by $\vec \Delta := \vec x_{j}-\vec x_{j'}$ intersect at $v_{i_j}^{\vec x_j} = v_{i_{j'}}^{\vec x_{j'}+\vec\Delta}$.
  Because $P$ starts and ends at free vertices, and $M$ is periodic so vertex freedom is preserved under translations by $\vec \Delta$, we must have $1 < j < k$.
  Thus vertex $v_{i_j}^{\vec x_j}$ must be incident to an edge in~$M$ that is both in $P$ and in $P+\Delta$; it cannot be $(v_{i_{j-1}}^{\vec x_{j-1}}, v_{i_j}^{\vec x_j})$
  or else that edge would be in both $P$ and $P+\vec \Delta$, so $v_{i_{j-1}} = v_{i_{j'-1}}$, contradicting that $j$ is smallest.
  Hence, the edge $(v_{i_j}^{\vec x_j}, v_{i_{j+1}}^{\vec x_{j+1}})$ is in $M$, and is common to both $P$ and $P+\vec \Delta$.
  Because $\PER{G}$ is local bipartite, both $P$ and $P+\vec \Delta$ start at a vertex in $\PER{B}$, so this edge is oriented the same in both walks, and thus
  $(v_{i_j}^{\vec x_j}, v_{i_{j+1}}^{\vec x_{j+1}}) =
  (v_{i_{j'}}^{\vec x_{j'}+\vec \Delta}, v_{i_{j'+1}}^{\vec x_{j'+1}+\vec \Delta})$.
  We can now construct a shorter alternating walk $P'=(v_{i_1}^{\vec x_1},v_{i_2}^{\vec x_2},\ldots,v_{i_j}^{\vec x_j},v_{i_{j'+1}}^{\vec x_{j'+1}+\vec \Delta},\ldots,v_{i_k}^{\vec x_k+\vec \Delta})$, skipping $j'-j > 0$ vertices strictly between $v_{i_j}^{\vec x_j}$ and $v_{i_{j'+1}}^{\vec x_{j'+1}}$.
  As argued above, this walk connects two free vertices (the same protovertices).
  This walk may repeat vertices and edges, which is why we needed to allow repetitions when defining $P$. But in the end, we show that the walk cannot repeat a protovertex, so it cannot actually repeat a vertex, and thus it is a valid augmenting path.
\end{proof}

\begin{theorem}\label{thm:periodicmatching}
  If a $d$-dimensional bipartite periodic graph $\PER{G}$ admits a perfect matching, then it admits a 1-periodic perfect matching.
\end{theorem}

\begin{proof}
  Suppose $\PER{G}=(\PER{V},\PER{E})$ admits a perfect matching $M^{*}$.
  We build a periodic perfect matching $\PER{M}$ incrementally by finding augmenting paths using Lemma~{\ref{lem:bipartiteaugmenting}}.
  As we maintain that the matching $\PER{M}$ is 1-periodic, it suffices to keep track of the set $M$ of protoedges in the matching, and initially $M=\emptyset$. A protovertex is \defn{free} if none of its vertices is incident to an edge in $\PER{M}$.
  We can reduce the number of free protovertices to zero as follows:
  \begin{enumerate}
  \item Take an augmenting path $P$ without any repeating protovertices, which is known to exist by Lemma~{\ref{lem:bipartiteaugmenting}}.
  \item Take all translations of $P$ by vectors in $\mathbb Z^d$.
  \item As each protovertex appears only once in $P$, all translations of $P$ are disjoint. Then augment $\PER{M}$ by the union of all translations of $P$, thereby reducing the number of free vertices. This can be done in linear time by just updating $M$.
  \item Repeat until no free protovertices remain. \qedhere
  \end{enumerate}
\end{proof}

Here again, the validity of this theorem relies on the fact that the periodic graph is bipartite periodic. Just being bipartite would not suffice which can be seen by taking the same example graph as earlier with one protovertex $v$ and one protoedge $(v^{0}, v^{p})$ with $p$ prime. This graph admits a perfect matching but its periodicity is at least $2p$. The same holds for the local 1D periodic graph shown earlier.

\begin{theorem}\label{thm:matchingalgorithm}
  Given a $d$-dimensional bipartite periodic graph $\PER{G}=(\PER{V},\PER{E})$, a perfect matching for $\PER{G}$, or whether it exists, can be computed in 
  $O(|E| \sqrt{|V|})$. The perfect matching being returned is 1-periodic.
\end{theorem}
  
\begin{proof}
  Given $V$ and $E$, we construct a graph $G'=(V,E')$ mapping the periodic graph onto a $d$-dimensional torus. For this, project each protoedge $(u^{\vec x},v^{\vec y})\in E$ onto an edge $(u,v)\in E'$. Duplicates can be removed. 
  If $\PER{G}$ admits a 1-periodic perfect matching $M$, then its projection $M'$ is a perfect matching in $G'$. Conversely, any perfect matching $M'$ in $G'$ can be lifted to a 1-periodic perfect matching $M$ in $\PER{G}$, by  picking for each edge in $M'$, any edge that projects to it (and all its translations).
  Now can use any finite perfect matching algorithms on $G'$, such as the $O(|E| \sqrt{|V|})$ Hopcroft–Karp algorithm \cite{hopcroft1973n}.
\end{proof}

\section{Tiling Hardness}
\subsection{Tromino + Periodic \label{sec:subspace}}

We now turn to the problem of tiling the plane with polyominoes. We first consider the discrete version of tiling the integer lattice $\mathbb Z^d$. In this setting, a (finite) tile $P\subset \mathbb Z^d$ \defn{tiles} a subset $E\subseteq \mathbb Z^d$ by translations $A\subseteq \mathbb Z^d$ if 
\begin{itemize}
\item the set $E$ is covered without overlap, that is, for every $e\in E$ there is exactly one translation $a\in A$ and one $p\in P$ such that $e=a+p$; and
\item for every translation $a\in A$, the translated tile is in $E$, that is, $a+P\subseteq E$.
\end{itemize}
A set $E\subseteq\mathbb Z^d$ is \defn{periodic} if it is invariant under $d$ independent translations, that is, there are $d$ vectors $\vec v_1,\ldots,\vec v_d\in\mathbb Z^d$ such that for every $E + \vec v_i = E$  for all $i=1,\ldots,d$.
Although periodic subsets are infinite, they can be described in finite space by a finite set of points and the $d$ translation vectors. 

A \defn{tromino} is a (connected) polyomino of size 3. There are two trominoes up to rotations: the \L (L) tromino and the \I (I) tromino.

\begin{theorem}\label{thm:periodicsubset}
  Tiling a given periodic subset $E\subseteq \mathbb Z^2$ with copies of a single tromino (\L or \I) is co-RE-complete and thus undecidable. In 1.5D, the problem is PSPACE-complete.
\end{theorem}

\begin{proof}
  The proof is by reduction from local Periodic Planar Trichromatic Graph Orientation from Theorem~\ref{thm:planar-graph-orientation}.
  Given a local planar periodic drawing of the graph instance, construct a local planar orthogonal periodic drawing of polynomial grid size, using Lemma~\ref{lem:orthodraw}.
  Apply Lemma~\ref{lem:degree3} to further refine so that all degree-3 vertices connect locally left, up, and right;
  and all trichromatic vertices appear in the orientation blue-red-green or green-red-blue (rotating via Figure~\ref{fig:rotation}).

  Let $M$ be the grid size of the resulting periodic drawing.
  Overlay the $M\times M$ grid graph with our drawing and consider the dual grid graph whose faces are $1/M \times 1/M$ squares. Each square intersects the drawing in a few possible ways:
  \begin{enumerate}
  \item \defn{Wire}: a horizontal or vertical edge portion (red, green or blue).
  \item \defn{Wire bend}: an edge bend (red, green or blue) in one of 4 orientations.
  \item \defn{Monochromatic vertex}: a monochromatic 1-in-3 degree-3 vertex (red, green or blue), in one orientation (left, up, right).
  \item \defn{Trichromatic vertex}: a trichromatic 0-or-3-in-3 vertex, of degree 3, in one of two orientations (blue-red-green or green-red-blue).
  \item The square face does not intersect the drawing.
  \end{enumerate}

  We transform the graph drawing into a periodic set $E\subseteq \mathbb Z^2$ with translation vectors $(0,6M)$ and $(6M,0)$. We only need to specify the elements in $E$ within the $[1,6M]^2$ square. Subdivide that square into $M\times M$ subsquares, each $6 \times 6$ corresponding to one of the square faces of the dual graph above, and within each of these subsquares, include elements in $E$ depending on the intersection type, as specified in Figure~\ref{fig:Lgadgets} for \L trominoes and Figure~\ref{fig:Igadgets} for \I trominoes, where $E$ consists of all \defn{pixels} that are not dark gray.
  Figure~\ref{fig:reductionexample} gives complete examples of the reductions.

  Figures~\ref{fig:Lgadgets} and~\ref{fig:Igadgets} show the only possible tilings for each gadget:
  \begin{enumerate}
  \item Wire and wire bend: There are only two ways to tile the gadget. The edge orientation is encoded by the vector from the (\I or \L) tromino center to the pixel matching the color of the edge.
  \item Monochromatic vertex: There are three solutions, determined by the edge covering the center pixel.
  \item Trichromatic vertex: There are two solutions, and the presence of a tromino in the center of the gadget determines if all edges are oriented inwards or outwards.
  \end{enumerate}

  Thus, the periodic subset $E$ can be tiled by trominoes if and only if the Periodic Planar Trichromatic Graph Orientation has a solution.
\end{proof}

\begin{figure}
  \centering
  \subcaptionbox
     {4-coloring of the square grid}%
     {\includegraphics[scale=0.8]{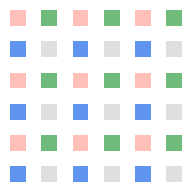}}%
  \hfil
  \subcaptionbox
     {Red wire gadget and its two solutions}%
     {\includegraphics[scale=0.8]{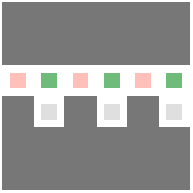}
      \ 
      \includegraphics[scale=0.8]{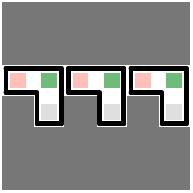}
      \ 
      \includegraphics[scale=0.8]{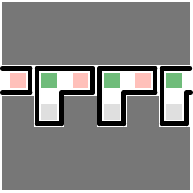}}%
  \par
  \subcaptionbox
     {Red wires and turns}%
     {\includegraphics[scale=0.8]{figs/L/wire/gadget}
      \
      \includegraphics[scale=0.8]{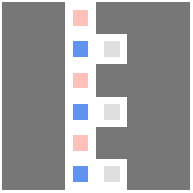}
      \
      \includegraphics[scale=0.8]{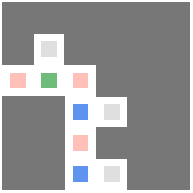}
      \
      \includegraphics[scale=0.8]{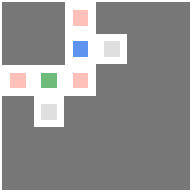}
      \
      \includegraphics[scale=0.8]{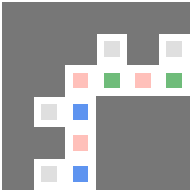}
      \
      \includegraphics[scale=0.8]{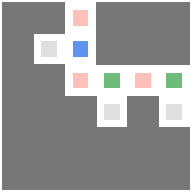}}%
  \par
  \subcaptionbox
     {Green wires and turns}%
     {\includegraphics[scale=0.8]{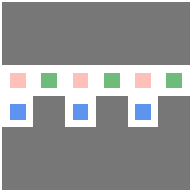}
      \
      \includegraphics[scale=0.8]{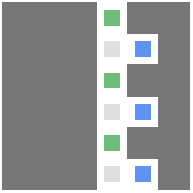}
      \
      \includegraphics[scale=0.8]{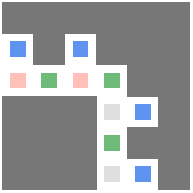}
      \
      \includegraphics[scale=0.8]{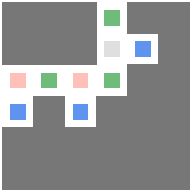}
      \
      \includegraphics[scale=0.8]{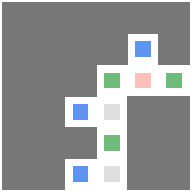}
      \
      \includegraphics[scale=0.8]{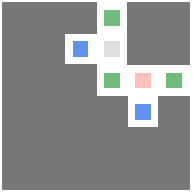}}%
  \par
  \subcaptionbox
     {Blue wires and turns}%
     {\includegraphics[scale=0.8]{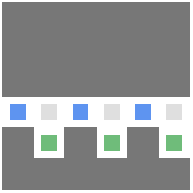}
      \
      \includegraphics[scale=0.8]{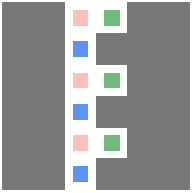}
      \
      \includegraphics[scale=0.8]{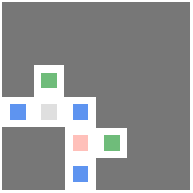}
      \
      \includegraphics[scale=0.8]{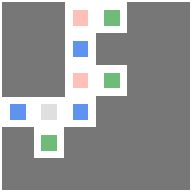}
      \
      \includegraphics[scale=0.8]{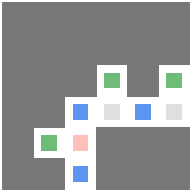}
      \
      \includegraphics[scale=0.8]{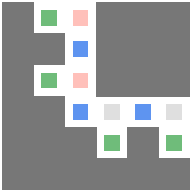}}%
  \par
  \subcaptionbox
     {Red 1-in-3 gadget and its three solutions, Green and Blue 1-in-3}%
     {\includegraphics[scale=0.8]{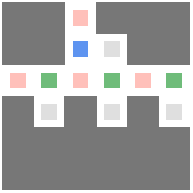}
      \ 
      \includegraphics[scale=0.8]{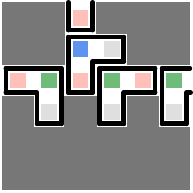}
      \ 
      \includegraphics[scale=0.8]{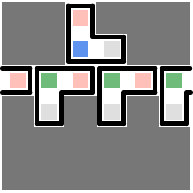}
      \ 
      \includegraphics[scale=0.8]{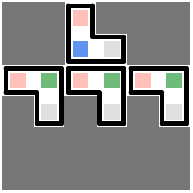}
      \ 
      \includegraphics[scale=0.8]{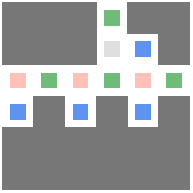}
      \ 
      \includegraphics[scale=0.8]{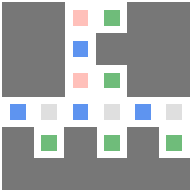}}%
  \par
  \subcaptionbox
     {0-or-3-in-3 Blue-Red-Green gadget and its two solutions, and Green-Red-Blue version}%
     {\includegraphics[scale=0.8]{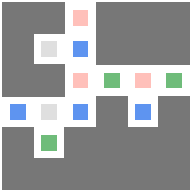}
      \ 
      \includegraphics[scale=0.8]{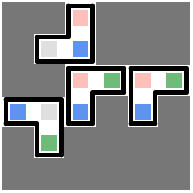}
      \ 
      \includegraphics[scale=0.8]{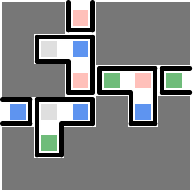}
      \ 
      \includegraphics[scale=0.8]{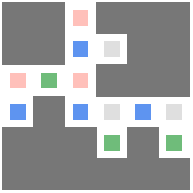}}%
  \caption{Reduction from Planar Trichromatic Graph Orientation to tiling with \L trominoes.}%
  \label{fig:Lgadgets}
\end{figure}

\begin{figure}
  \centering
  \subcaptionbox
     {4-coloring of the square grid}%
     {\includegraphics[scale=0.8]{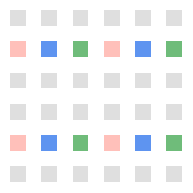}}%
  \hfil
  \subcaptionbox
     {Red wire gadget and its two solutions}%
     {\includegraphics[scale=0.8]{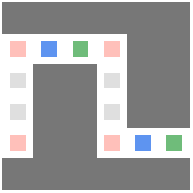}
      \ 
      \includegraphics[scale=0.8]{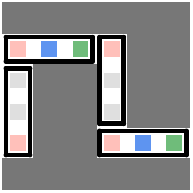}
      \ 
      \includegraphics[scale=0.8]{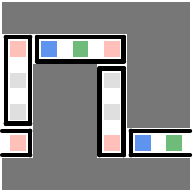}}%
  \par
  \subcaptionbox
     {Red wires and turns}%
     {\includegraphics[scale=0.8]{figs/I/wire/wire-WE-R}
      \
      \includegraphics[scale=0.8]{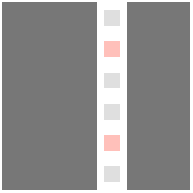}
      \
      \includegraphics[scale=0.8]{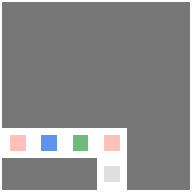}
      \
      \includegraphics[scale=0.8]{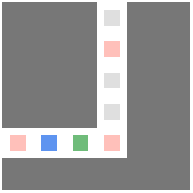}
      \
      \includegraphics[scale=0.8]{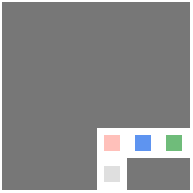}
      \
      \includegraphics[scale=0.8]{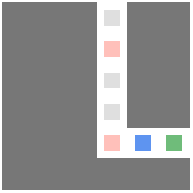}}%
  \par
  \subcaptionbox
     {Green wires and turns}%
     {\includegraphics[scale=0.8]{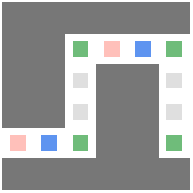}
      \
      \includegraphics[scale=0.8]{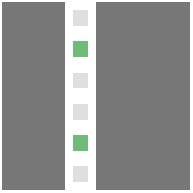}
      \
      \includegraphics[scale=0.8]{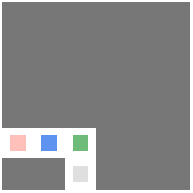}
      \
      \includegraphics[scale=0.8]{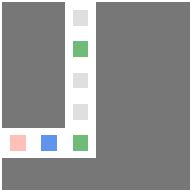}
      \
      \includegraphics[scale=0.8]{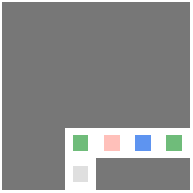}
      \
      \includegraphics[scale=0.8]{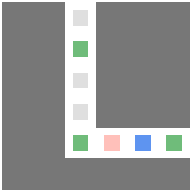}}%
  \par
  \subcaptionbox
     {Blue wires and turns}%
     {\includegraphics[scale=0.8]{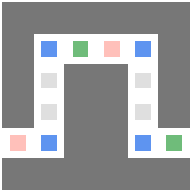}
      \
      \includegraphics[scale=0.8]{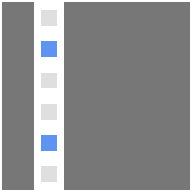}
      \
      \includegraphics[scale=0.8]{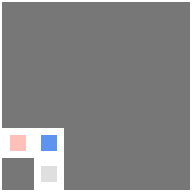}
      \
      \includegraphics[scale=0.8]{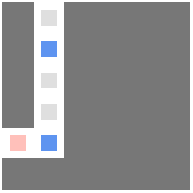}
      \
      \includegraphics[scale=0.8]{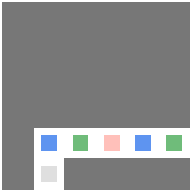}
      \
      \includegraphics[scale=0.8]{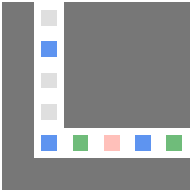}}%
  \par
  \subcaptionbox
     {Red 1-in-3 gadget and its three solutions, Green and Blue 1-in-3}%
     {\includegraphics[scale=0.8]{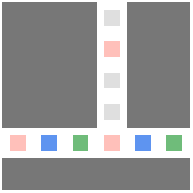}
      \ 
      \includegraphics[scale=0.8]{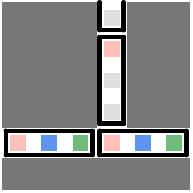}
      \ 
      \includegraphics[scale=0.8]{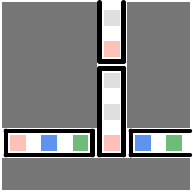}
      \ 
      \includegraphics[scale=0.8]{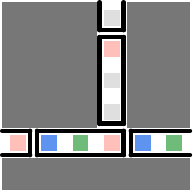}
      \ 
      \includegraphics[scale=0.8]{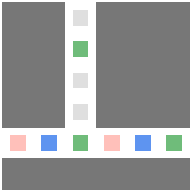}
      \ 
      \includegraphics[scale=0.8]{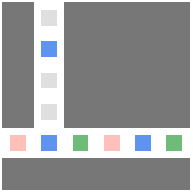}}%
  \par
  \subcaptionbox
     {0-or-3-in-3 Blue-Red-Green gadget and its two solutions, and Green-Red-Blue version}%
     {\includegraphics[scale=0.8]{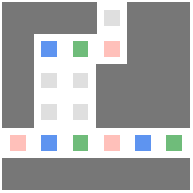}
      \ 
      \includegraphics[scale=0.8]{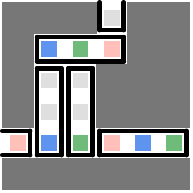}
      \ 
      \includegraphics[scale=0.8]{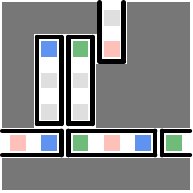}
      \ 
      \includegraphics[scale=0.8]{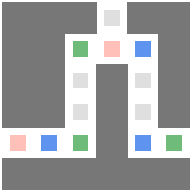}
      \ 
      \includegraphics[scale=0.8]{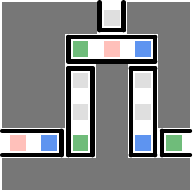}
      \ 
      \includegraphics[scale=0.8]{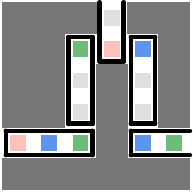}}%
  \caption{Reduction from Planar Trichromatic Graph Orientation to tiling with I trominoes.}%
  \label{fig:Igadgets}
\end{figure}

\begin{figure}
  \centering
  \subcaptionbox
     {Graph overlayed with the $M\times M$ grid}%
     [.45\linewidth]%
     {\includegraphics[width=\linewidth]{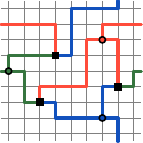}}%
  \hfill
  \subcaptionbox
     {Graph overlayed with the dual grid}%
     [.45\linewidth]%
     {\includegraphics[width=\linewidth]{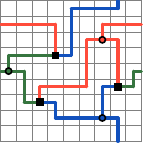}}%
  \par\vspace{\baselineskip}
  \subcaptionbox
     {Substitution of \L-tromino gadgets}%
     [.45\linewidth]%
     {\includegraphics[width=\linewidth]{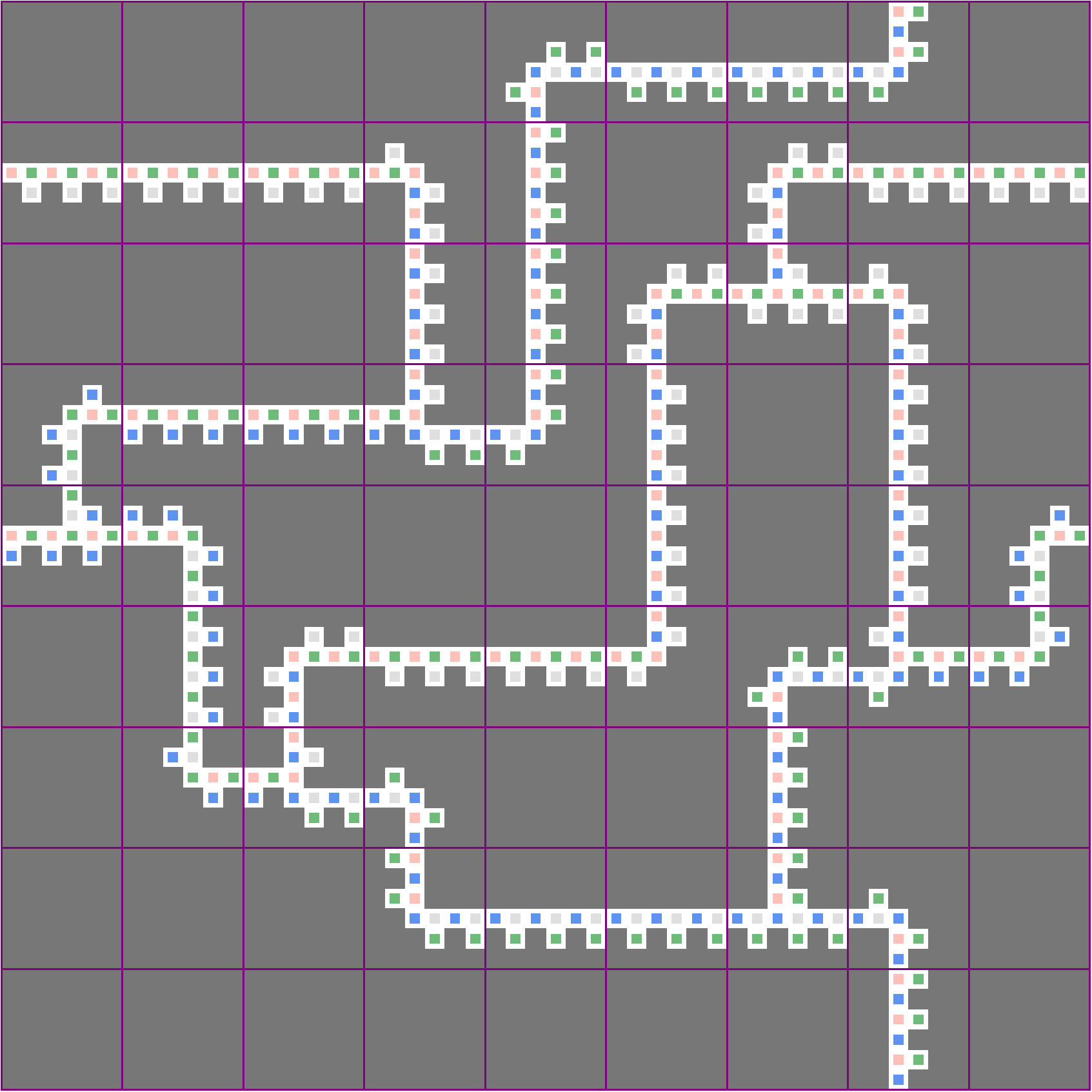}}%
  \hfill
  \subcaptionbox
     {Substitution of \I-tromino gadgets}%
     [.45\linewidth]%
     {\includegraphics[width=\linewidth]{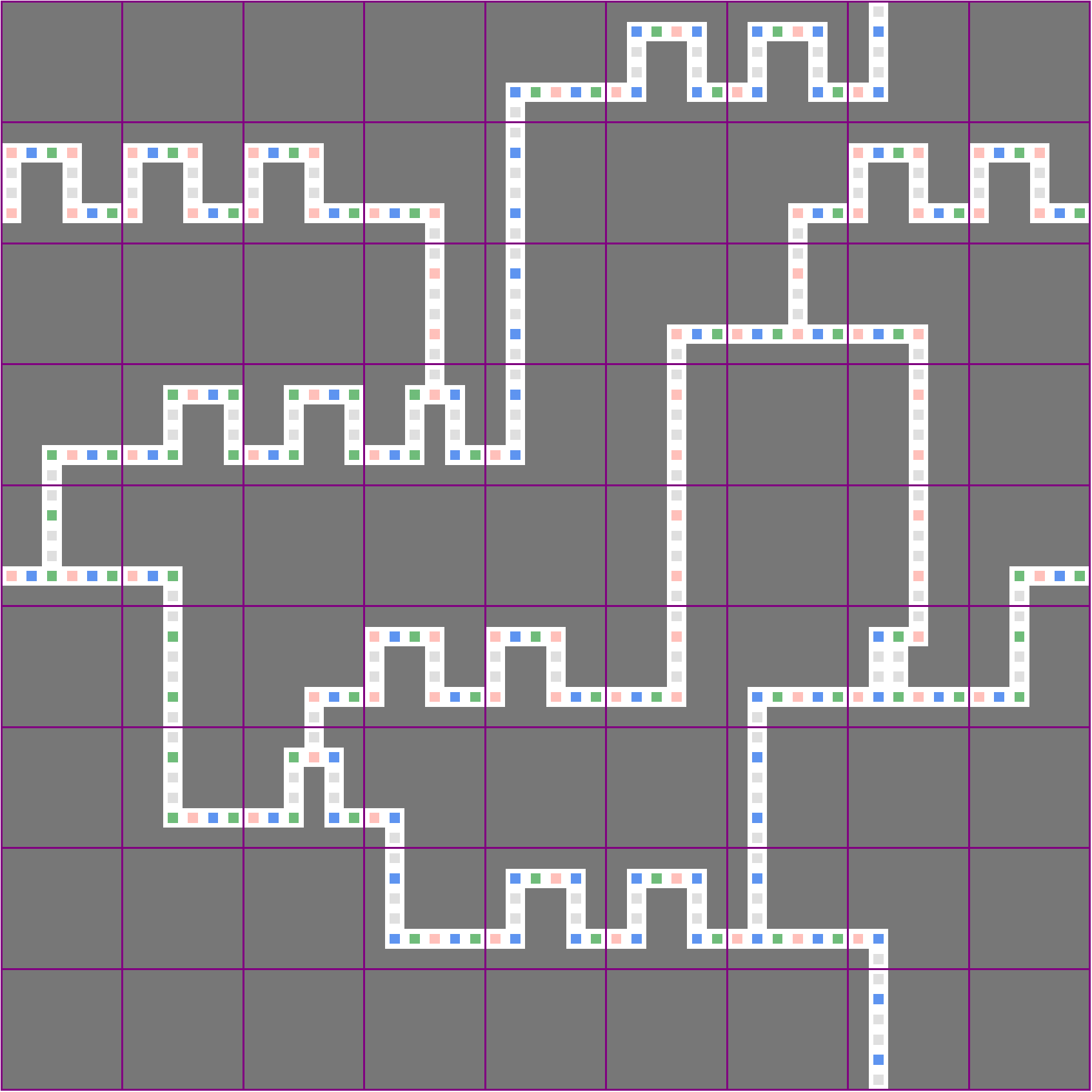}}%
  \caption{Reduction example from a periodic orthogonal graph drawing to a periodic subset to be tiled.}%
  \label{fig:reductionexample}
\end{figure}

\begin{corollary}\label{cor:periodicsubsettranslation}
  Tiling a given periodic subset $E\subseteq \mathbb Z^2$ with translations of two polyominoes (the I trominoes \I and \Ir) is co-RE-complete and thus undecidable. In 1.5D, the problem is PSPACE-complete.
\end{corollary}

\subsection{$O(1)$-omino + Disconnected Polyomino \label{sec:two-polyomino}}

\begin{theorem}\label{thm:twopolyomino}
  Tiling with two polyominoes, one of which is of constant size, the other  non-connected, is co-RE-complete and thus undecidable. In 1.5D, the problem is PSPACE-complete.
\end{theorem}

\begin{proof}
  We proceed as in the previous section in reducing from Periodic Planar Trichromatic Graph Orientation and transforming the instance into a periodic set $E$ with translation vectors $(6M,0)$ and $(0,6M)$. Let $E_0 := E \cap [1,6M]^2$. This is the pattern which is repeated by translation to produce the periodic set to be tiled.
  Now consider the complement $\bar E_0 :=[1,6M^2] \setminus E_0$. If the graph orientation problem has a solution, then $E_0$ and the \I tromino tile the plane as shown in the previous section, however the \I tromino can tile the plane on its own so we will need to modify both tiles so that none of them tile the plane on its own, and the pattern $E_0$ must combine with itself in a unique way to produce the complement of $E$. 
  
  To prevent the first (connected) polyomino from tiling the plane on its own, we refine the integer lattice, replacing every pixel by a $3\times 3$ square of 9 pixels. The \I tromino is replaced with the shape shown in Figure~\ref{fig:bumpytromino} which is produced by replacing each of the three tromino pixels by 5 sub-pixels forming a $+$ pattern, producing a 15-omino $P$.
  Likewise, we refine $E_0$ into $E_0' \subseteq [1,18M]^2$ replacing each pixel by the same pattern, and $E'$ by translating $E'_0$ with vectors $(18M,0)$ and $(0,18M)$. See Figure~\ref{fig:0-or-3-plus} for the substitution applied to the trichromatic gadget. 
  It can easily be verified that the 15-omino does not tile on its own. Yet any tiling of the refinement is strictly equivalent to a tiling of the original set $E$ with the \I tromino.
  \begin{figure}
    \centering
    \includegraphics[scale=0.8]{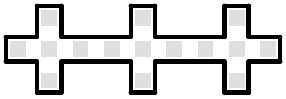}
    \caption{The 15-omino $P$ obtained by refining the \I tromino.}%
    \label{fig:bumpytromino}
  \end{figure}

  \begin{figure}
    \centering
    \includegraphics[width=0.3\linewidth]{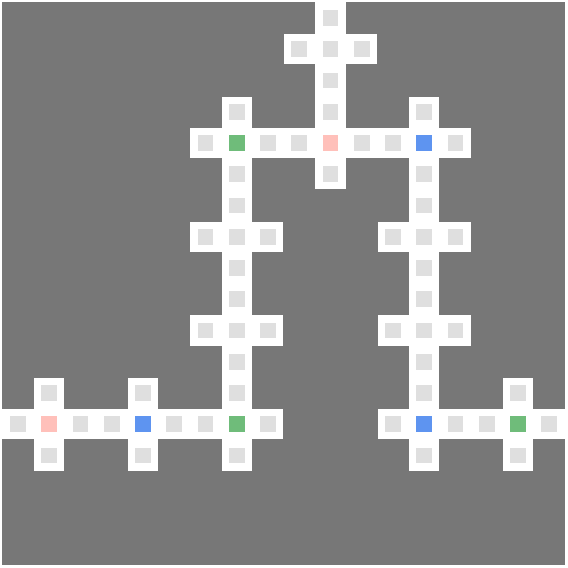}
    \hfill
    \includegraphics[width=0.3\linewidth]{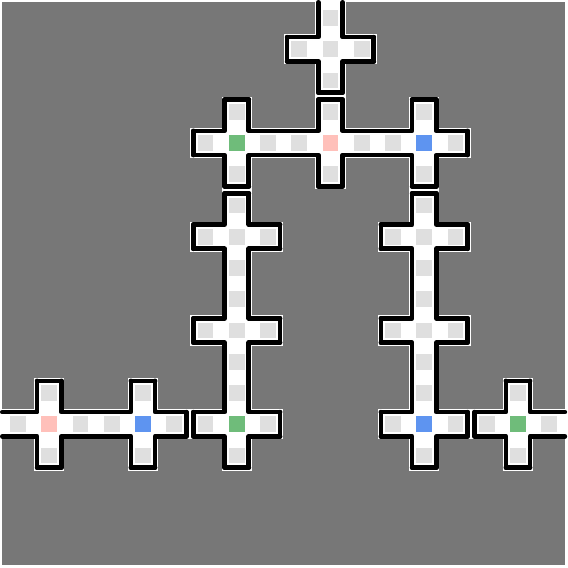}
    \hfill
    \includegraphics[width=0.3\linewidth]{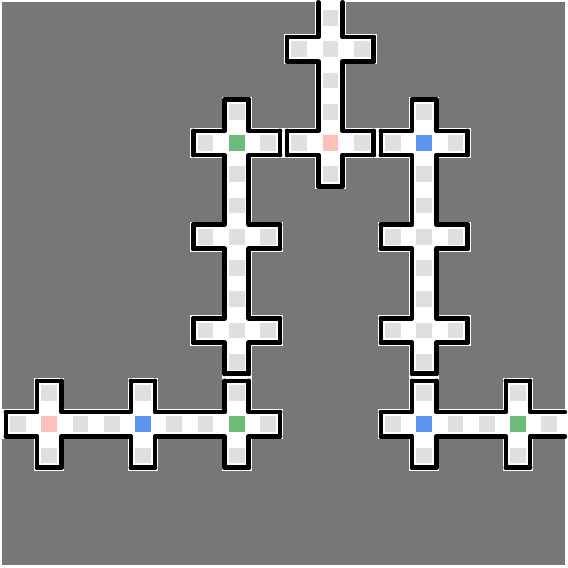}
    \caption{Refinement of the trichromatic gadget and its two solutions.}%
    \label{fig:0-or-3-plus}
  \end{figure}

  Notice that by construction of the orthogonal drawing, the $6\times 6$ squares at the 4 corners of $[1,6M]^2$ do not intersect the drawing and thus do not intersect $E$, and thus these squares are totally filled in $\bar E_0$. In the refinement $\bar E'_0$ the corresponding $18\times 18$ squares are totally filled as well. We construct the second, non-connected polyomino $Q$ by taking $\bar E'_0$ and replacing the $18\times 18$ squares at the corners by the patterns shown in Figure~\ref{fig:keycorners}. The patterns are designed with complementary parts that can only match each other so that any tiling with $Q$ must glue other $Q$'s in a grid pattern.

  \begin{figure}
    \centering
    \includegraphics[scale=0.4]{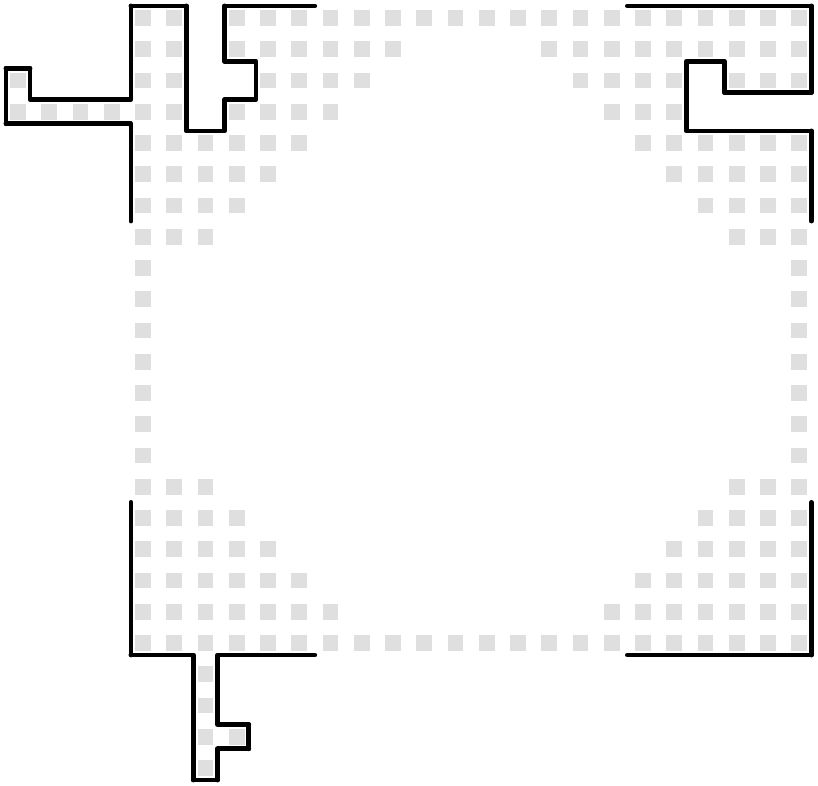}
    \caption{The four corner patterns used to construct polyomino $Q$.}%
    \label{fig:keycorners}
  \end{figure}

  Since $P$ does not tile on its own, there must be at least one $Q$ in the tiling, such $Q$ must be combined with other $Q$ to form a grid pattern. Together, they form exactly the complement of $E'$, which in turn can be tiled with $P$ if and only if the graph orientation problem has a solution.
\end{proof}

The previous proof uses only translations of $Q$. Also, $P$ has 2-fold rotational symmetry, so all its copies in the tiling are translations of two polyominoes (without rotation). Thus we obtain the following corollary:

\begin{corollary}\label{cor:threepolyomino-translation}
  Tiling with three polyominoes by translation, two of which are of constant size, the third non-connected, is co-RE-complete and thus undecidable. In 1.5D, the problem is PSPACE-complete.
\end{corollary}

\subsection{Two Connected Polycubes}

\begin{corollary}\label{cor:twopolycube}
  Tiling with two connected polycubes (in 2.5D or 3D), one of which is of constant size, is co-RE-complete and thus undecidable. 
\end{corollary}
\begin{proof}
  Given the 2D instance of the previous theorem, extrude the polyomino $Q$ by 1 unit to turn it into a polycube, and make it connected by gluing above it an $18M \times 18M$ square extruded by 1. Extrude the tromino $P$ by 1 unit to turn it into a polycube. Again $P$ cannot tile the space on its own, and $Q$ must be combined with other $Q$'s to form a grid pattern. The resulting polycubes tile space if and only if the original Planar Trichromatic Graph Orientation problem has a solution.
\end{proof}

\begin{corollary}\label{cor:threepolycube-translation}
  Tiling with three connected polycubes by translation, two of which are of constant size, is co-RE-complete and thus undecidable. 
\end{corollary}

\subsection{Tromino Completion \label{sec:completion}}

\begin{theorem}\label{thm:periodic-completion}
  Given an infinite periodic partial tiling of the plane with one type of tromino (\L or \I), deciding whether it can be completed to a full tiling is co-RE-complete and thus undecidable. In 1.5D, the problem is PSPACE-complete.
\end{theorem}
\begin{proof}
  This time we reduce from Periodic Planar 3SAT-3 from Theorem~\ref{thm:planar-3sat}.
  As in the previous section, we start with the planar graph representation of the SAT instance, and transform it into a planar orthogonal periodic drawing of polynomial grid size, using Lemma~\ref{lem:orthodraw}, and further refine it so all degree 3 vertices connect locally left, up, and right, by Lemma~\ref{lem:degree3}. 
  Rotate the graph by 45 degrees couterclockwise, and overlay the rotated square grid with a \defn{brick} pattern, as shown in Figure~\ref{fig:completionexample}{(b)}. The pattern is built by bisecting horizontal rows of the rotated grid with horizontal lines, and within each row, bisecting successive vertices of the grid with vertical segments. 
  To complete the reduction, we need to show how to implement 
  \begin{enumerate}
  \item straight wires (2 directions),
  \item wire bends (all 4 rotations),
  \item 3SAT gadgets (1 orientation only),
  \item variable gadgets, and
  \item full bricks (to fill grid vertices not touched by edges of the graph),
  \end{enumerate}
  each by partially prefilling a brick with trominoes.
  Each brick has $\leq 4$ \defn{connectors}, two on its top side and two on its bottom side, each representing a boolean value. 

  In order to preserve parity in the gadgets, we represent a connector by a pair of pixels $p$ and $q$ (left-to-right) on the brick boundary, where one pixel is covered by a tromino from above, and the other from below. The value is true if $p$ is covered from below (meaning the all other pixels of the tile containing pixel $p$ are contained in the brick below $p$), and false if $q$ is covered from below. 

  We decompose each brick into 5 \defn{subbricks}: A full width \defn{major subbrick} $K$ with $\leq 4$ connectors again, and four half-width \defn{minor subbricks} $A$, $B$, $C$, and $D$, each with one connector on the top and one on the bottom. To form a brick, we place subbricks $A$ and $B$ on top of $K$, and $C$, $D$ at the bottom of $K$. 

  The $A$, $B$, $C$, and $D$ subbricks implement one of four gadgets; see Figure~\ref{fig:Lminor} for \L trominoes and Figure~\ref{fig:Iminor} for \I trominoes:
  \begin{enumerate}
  \item An \defn{equal} gadget, ensuring the value of its top connector equals the value of its bottom connector.
  \item A \defn{not} gadget ensuring the value of its top connector is the negation of the value of its bottom connector.
  \item A \defn{plug} gadget, with only one connector (top or bottom) of any value.
  \item A \defn{filler} gadget with no connector.
  \end{enumerate}
  \begin{figure}
    \centering
    \subcaptionbox
       {Equal}%
       [.18\linewidth]%
       {\includegraphics[width=\linewidth]{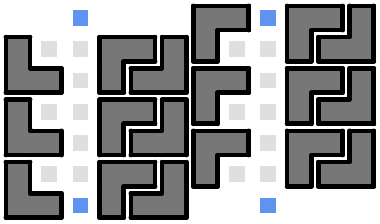}
        \par\medskip
        \includegraphics[width=\linewidth]{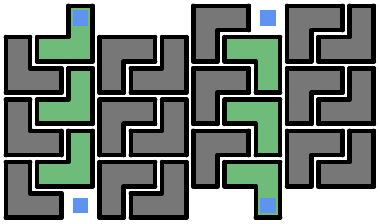}
        \par\medskip
        \includegraphics[width=\linewidth]{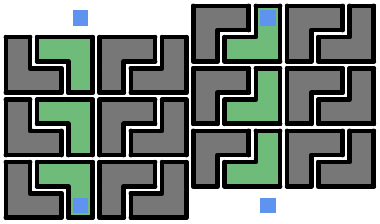}}%
    \hfil
    \subcaptionbox
       {Not}%
       [.18\linewidth]%
       {\includegraphics[width=\linewidth]{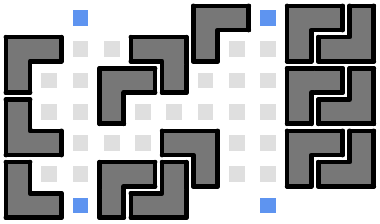}
        \par\medskip
        \includegraphics[width=\linewidth]{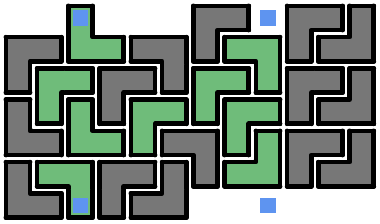}
        \par\medskip
        \includegraphics[width=\linewidth]{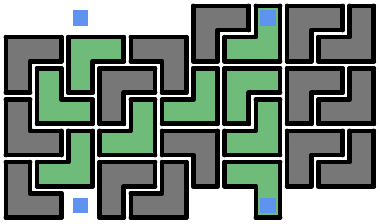}}%
    \hfil
    \subcaptionbox
       {Top plug}%
       [.18\linewidth]%
       {\includegraphics[width=\linewidth]{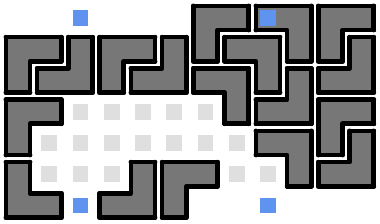}
        \par\medskip
        \includegraphics[width=\linewidth]{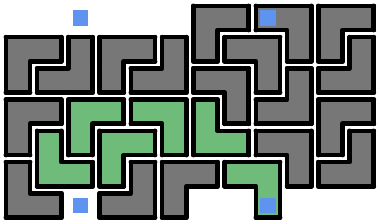}
        \par\medskip
        \includegraphics[width=\linewidth]{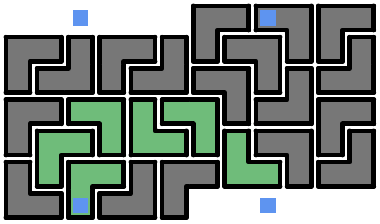}}%
    \hfil
    \subcaptionbox
       {Bottom plug}%
       [.18\linewidth]%
       {\includegraphics[width=\linewidth]{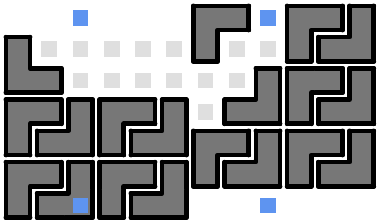}
        \par\medskip
        \includegraphics[width=\linewidth]{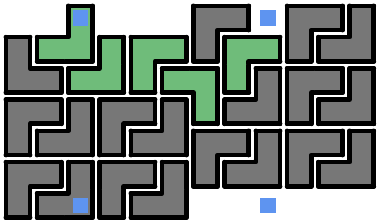}
        \par\medskip
        \includegraphics[width=\linewidth]{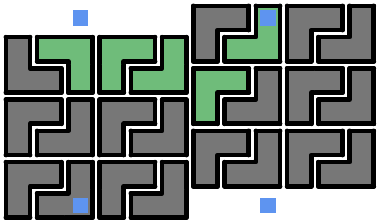}}%
    \hfil
    \subcaptionbox
       {Filler}%
       [.18\linewidth]%
       {\includegraphics[width=\linewidth]{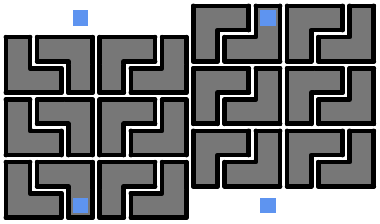}}%
    \caption{Minor subbricks for \L trominoes and their solutions.}%
    \label{fig:Lminor}
  \end{figure}

  \begin{figure}
    \centering
    \subcaptionbox
       {Equal}%
       [.18\linewidth]%
       {\includegraphics[width=\linewidth]{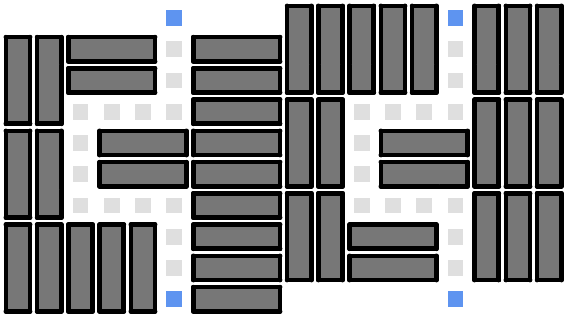}
        \par\medskip
        \includegraphics[width=\linewidth]{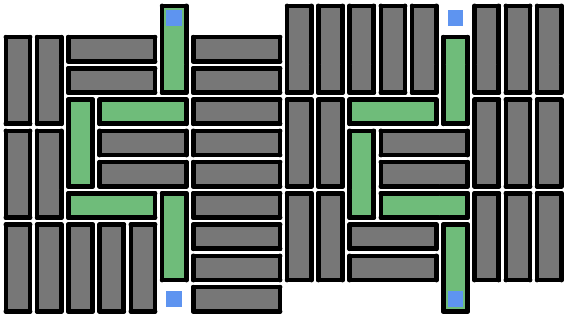}
        \par\medskip
        \includegraphics[width=\linewidth]{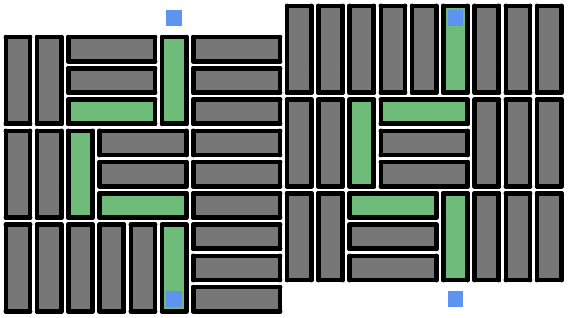}}%
    \hfill
    \subcaptionbox
       {Not}%
       [.18\linewidth]%
       {\includegraphics[width=\linewidth]{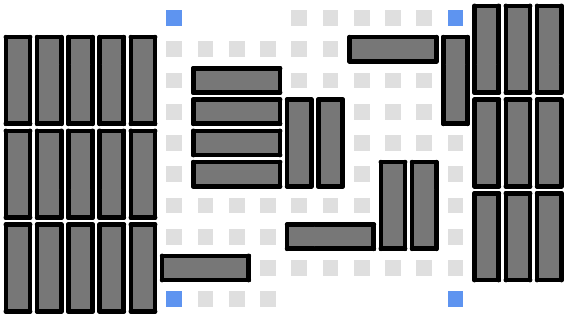}
        \par\medskip
        \includegraphics[width=\linewidth]{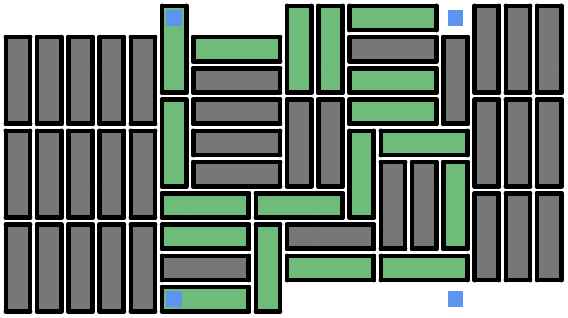}
        \par\medskip
        \includegraphics[width=\linewidth]{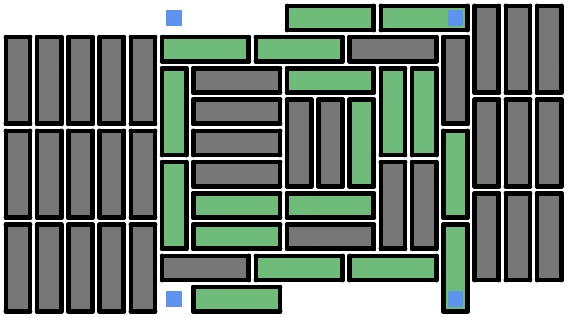}}%
    \hfill
    \subcaptionbox
       {Top plug}%
       [.18\linewidth]%
       {\includegraphics[width=\linewidth]{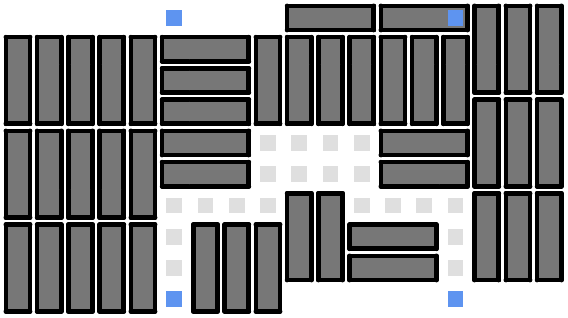}
        \par\medskip
        \includegraphics[width=\linewidth]{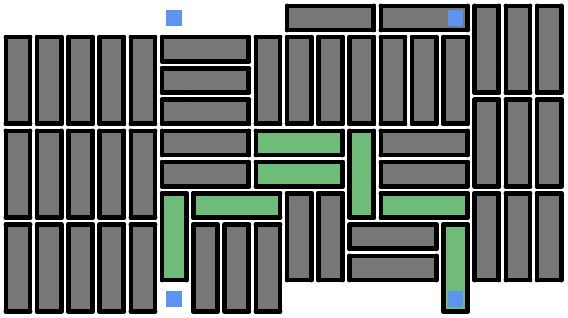}
        \par\medskip
        \includegraphics[width=\linewidth]{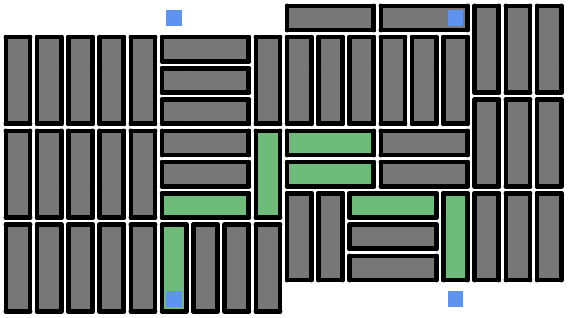}}%
    \hfill
    \subcaptionbox
       {Bottom plug}%
       [.18\linewidth]%
       {\includegraphics[width=\linewidth]{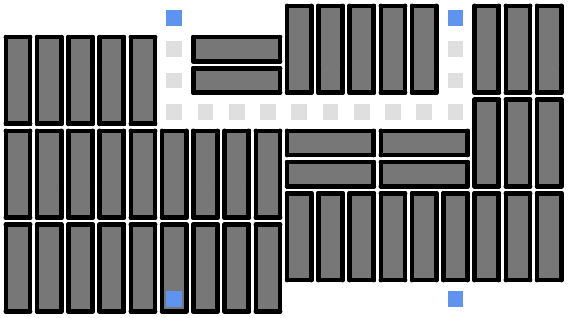}
        \par\medskip
        \includegraphics[width=\linewidth]{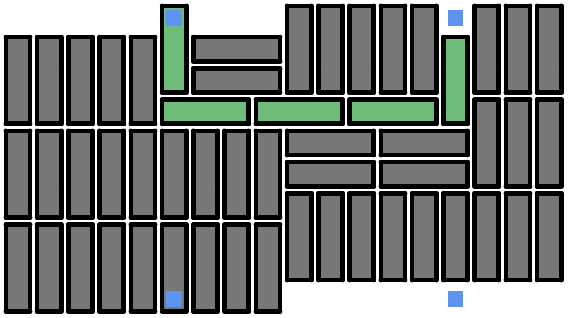}
        \par\medskip
        \includegraphics[width=\linewidth]{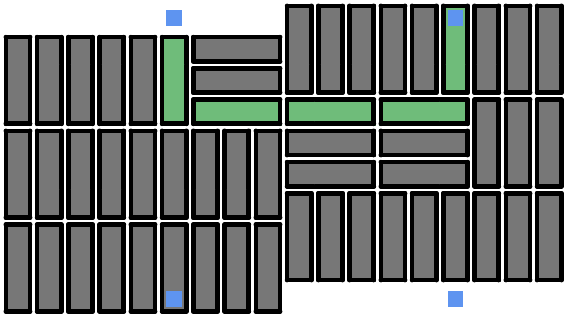}}%
    \hfill
    \subcaptionbox
       {Filler}%
       [.18\linewidth]%
       {\includegraphics[width=\linewidth]{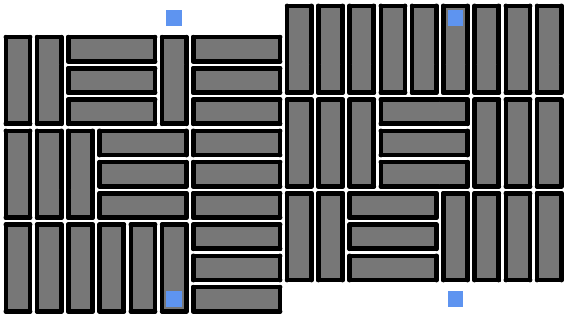}}%
    \caption{Minor subbricks for \I trominoes and their solutions.}%
    \label{fig:Iminor}
  \end{figure}

  The $K$ subbrick implements one of three gadgets:
  \begin{enumerate}
  \item A \defn{filler} (with zero connectors), by gluing 4 of the smaller filler bricks.
  \item A \defn{4-way duplicator} which ensures that the boolean value of all four connectors is equal. It is built by combining two half wires, and two upside-down plugs; see Figure~\ref{fig:duplicator}. In order to ensure the complementarity of the connectors, the 4 subbricks $A$, $B$, $C$, and $D$ surrounding the 4 way duplicator must be either plugs or nots.
  \item A \defn{monotone 3SAT clause}, with two connectors on the top and one on the bottom left, and which ensures that at least one of its 3 connectors has a true value. For \L-trominoes, the gadget is shown in Figure \ref{fig:L3SAT}. The gadget is thinner than the normal $K$ brick so it is thickened by adding equal gadgets. For \I-trominoes, the gadget is shown in Figure \ref{fig:I3SAT}. We first design a TFT-SAT gadget which is only tileable if the 3 connectors are True, False, and True, respectively. We then glue an equal and not gadget to the two top connectors to obtain the monotone 3SAT gadget. 
  \end{enumerate}
  \begin{figure}
    \centering
    \includegraphics[width=0.35\linewidth]{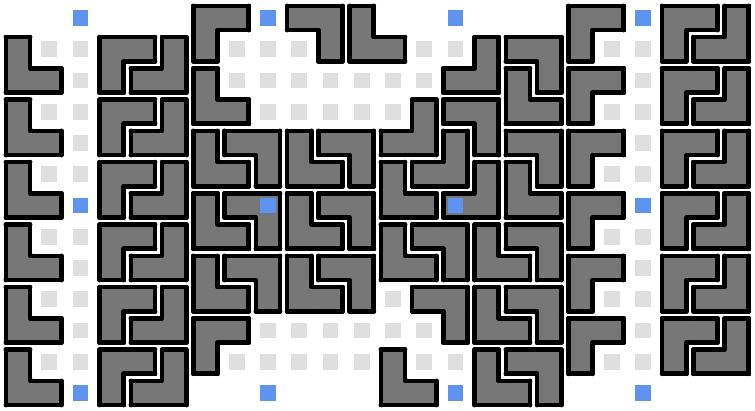}
    \hfil
    \includegraphics[width=0.35\linewidth]{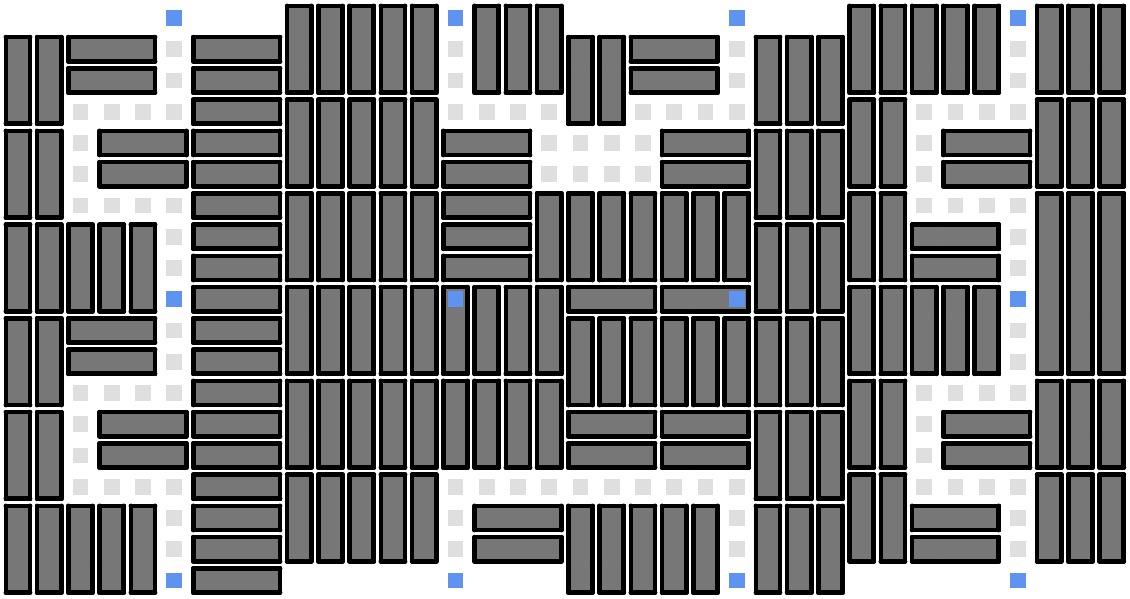}
    \caption{4-way duplicator for \L and \I trominoes.}%
    \label{fig:duplicator}
  \end{figure}

  \begin{figure}
    \centering
    \subcaptionbox
       {3SAT --- mini}%
       [.24\linewidth]%
       {\includegraphics[width=\linewidth]{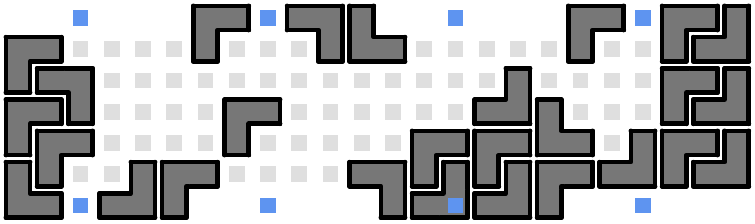}}%
    \hfill
    \subcaptionbox
       {True-True-True}%
       [.24\linewidth]%
       {\includegraphics[width=\linewidth]{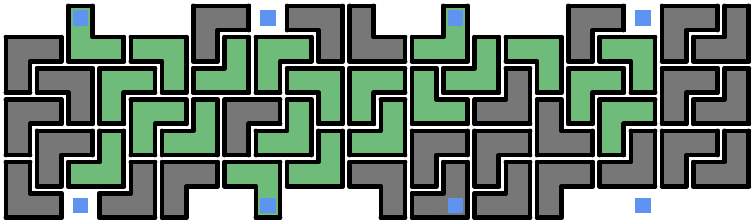}}%
    \hfill
    \subcaptionbox
       {True-True-False}%
       [.24\linewidth]%
       {\includegraphics[width=\linewidth]{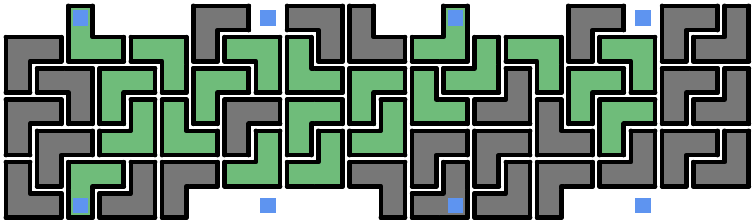}}%
    \hfill
    \subcaptionbox
       {True-False-True}%
       [.24\linewidth]%
       {\includegraphics[width=\linewidth]{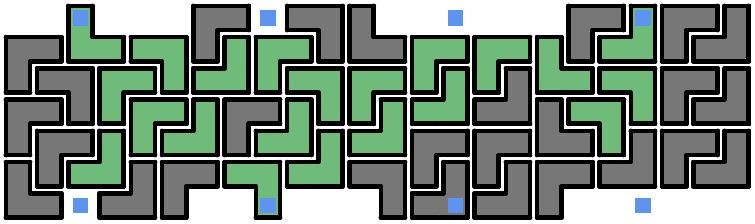}}%
    \par\medskip
    \subcaptionbox
       {True-False-False}%
       [.24\linewidth]%
       {\includegraphics[width=\linewidth]{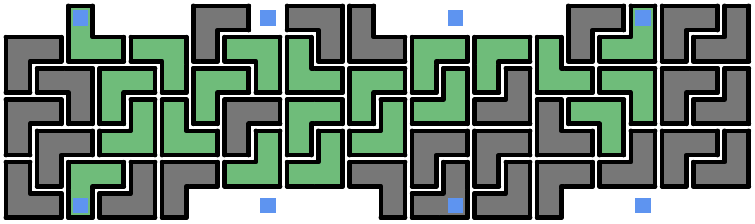}}%
    \hfill
    \subcaptionbox
       {False-True-True}%
       [.24\linewidth]%
       {\includegraphics[width=\linewidth]{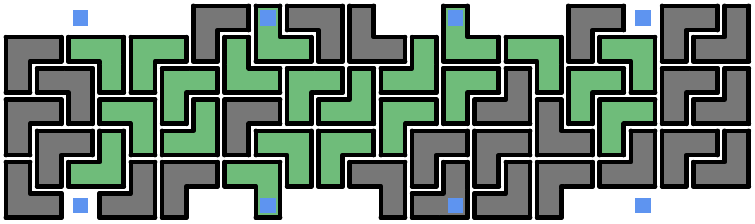}}%
    \hfill
    \subcaptionbox
       {False-True-False}%
       [.24\linewidth]%
       {\includegraphics[width=\linewidth]{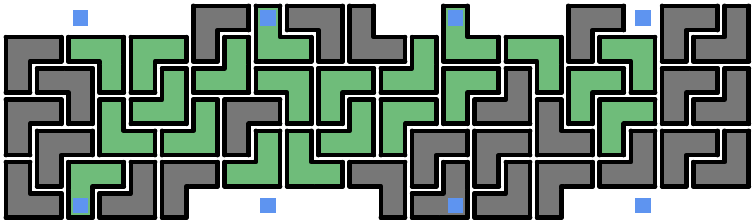}}%
    \hfill
    \subcaptionbox
       {False-False-True}%
       [.24\linewidth]%
       {\includegraphics[width=\linewidth]{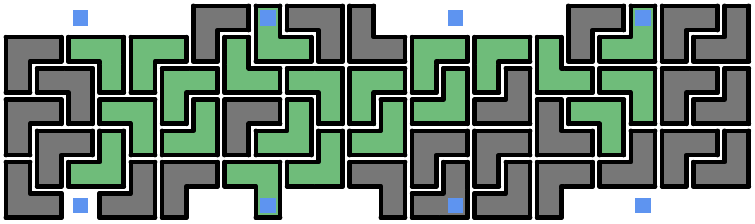}}%
    \par\medskip
    \subcaptionbox
       {3SAT --- full size}%
       [.24\linewidth]%
       {\includegraphics[width=\linewidth]{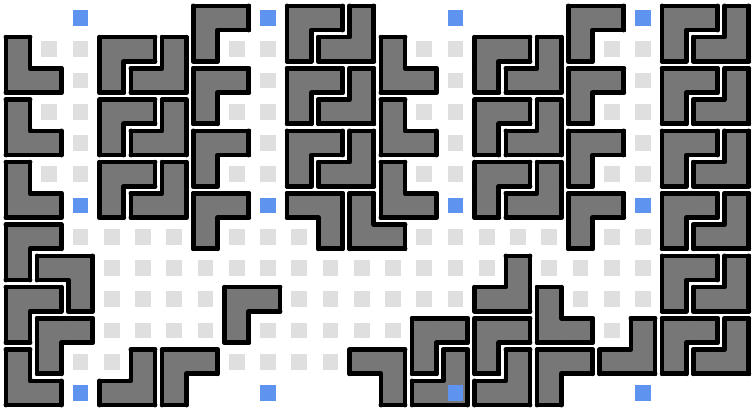}}%
    \caption{3SAT for \L trominoes.}%
    \label{fig:L3SAT}
  \end{figure}

  \begin{figure}
    \centering
    \subcaptionbox
       {True-False-True-SAT}%
       [.40\linewidth]%
       {\includegraphics[width=\linewidth]{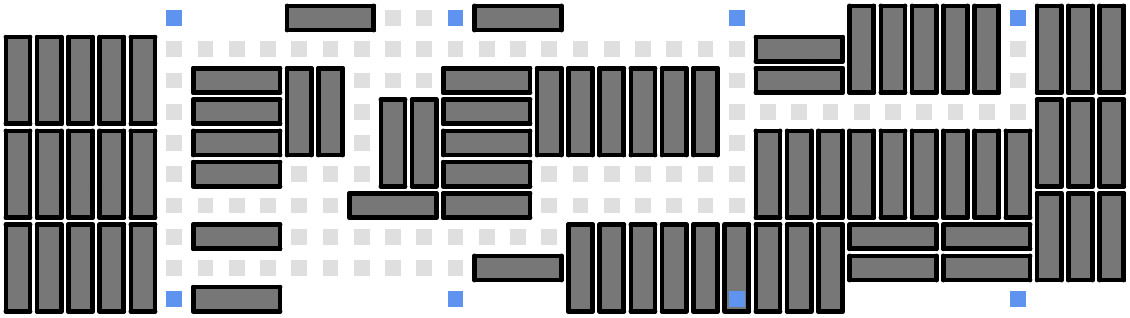}}%
    \hfil
    \subcaptionbox
       {False-False-False}%
       [.40\linewidth]%
       {\includegraphics[width=\linewidth]{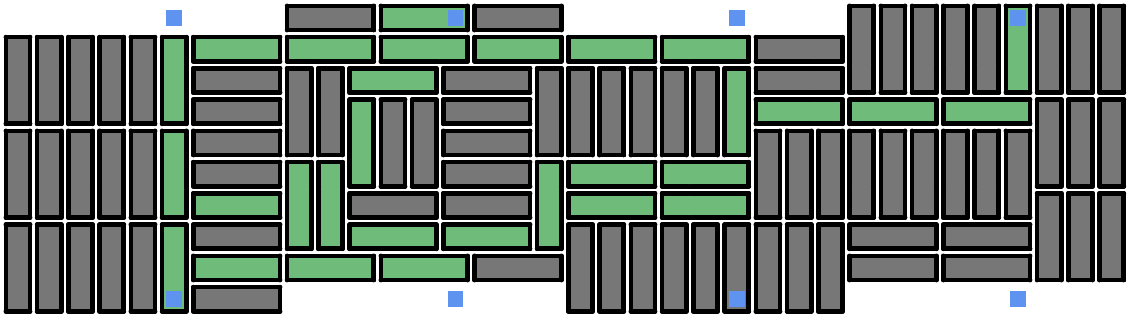}}%
    \par\medskip
    \subcaptionbox
       {True-X-False}%
       [.40\linewidth]%
       {\includegraphics[width=\linewidth]{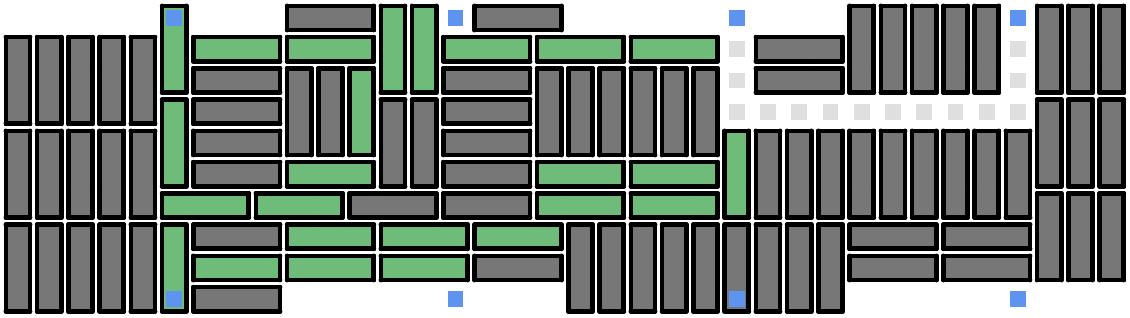}}%
    \hfil
    \subcaptionbox
       {True-X-True}%
       [.40\linewidth]%
       {\includegraphics[width=\linewidth]{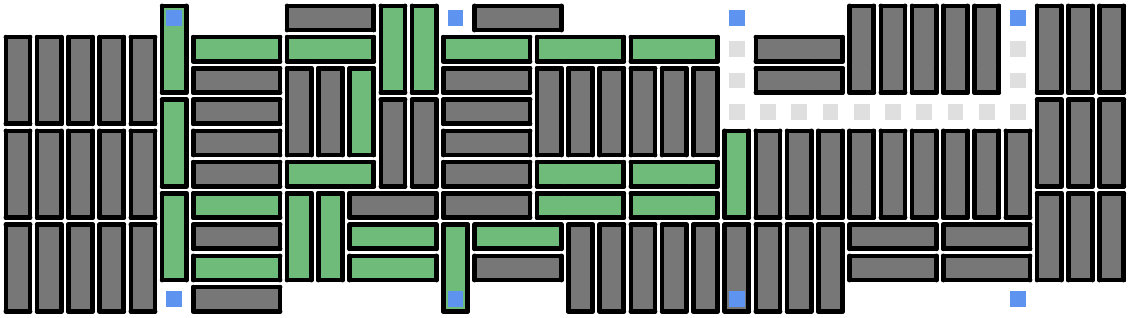}}%
    \par\medskip
    \subcaptionbox
       {3SAT}%
       [.40\linewidth]%
       {\includegraphics[width=\linewidth]{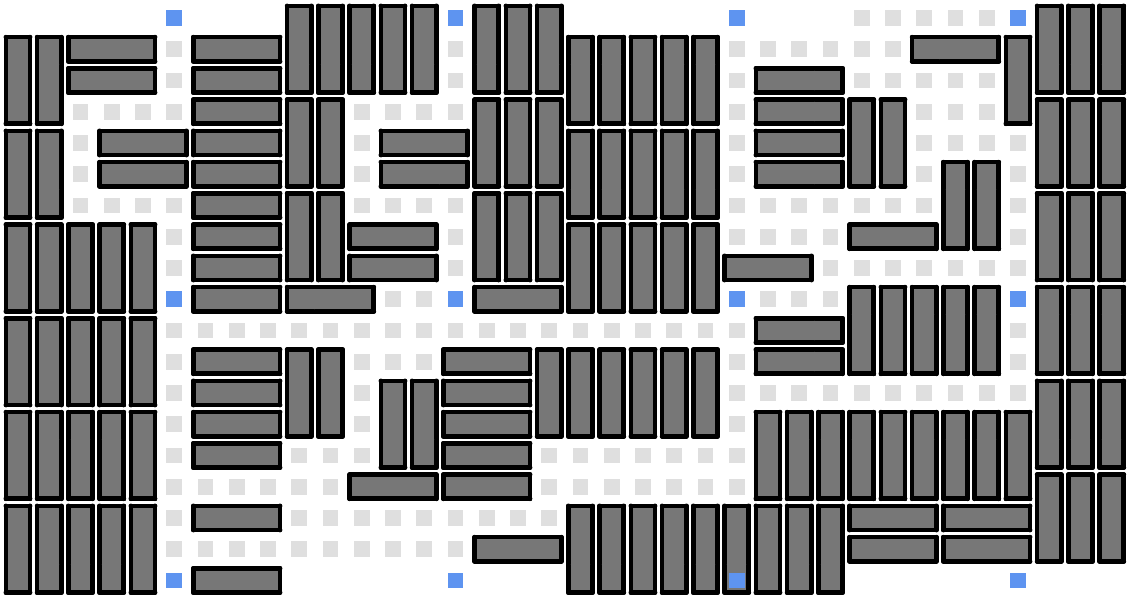}}%
    \caption{3SAT for \I trominoes.}%
    \label{fig:I3SAT}
  \end{figure}

  The subbricks are combined to obtain all required bricks:
  \begin{enumerate}
  \item Combining the 4-way duplicator with two plugs and two equals (or equivalently two nots), we obtain all straight wires and wire bends through the gridpoint of the brick.
  \item Combining the 4-way duplicator with 3 equals (or 3 nots) and one plug, we obtain a variable gadget.
  \item Combining the monotone 3SAT with equals or nots for on its three connectors, and the filler on the bottom right, we obtain all signed 3SUM clauses.
  \end{enumerate}

  Figure~\ref{fig:completionexample} shows complete examples of the reductions.
  The subbrick designs were found and verified through computer search.
  \begin{figure}
    \centering
    \subcaptionbox
       {3SAT graph overlayed with the $M\times M$ grid. Squares are clauses, dots are variables, red edges are negated literals, and blue edges are positive literals. }%
       [.48\linewidth]%
       {\includegraphics[width=0.7\linewidth]{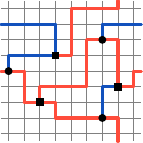}}%
    \hfill
    \subcaptionbox
       {Graph rotated and overlayed with the brick pattern.}%
       [.48\linewidth]%
       {\includegraphics[width=0.7\linewidth]{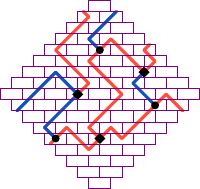}}%
    \par\vspace{\baselineskip}
    \subcaptionbox
       {Substitution of \L-tromino gadgets.}%
       [.48\linewidth]%
       {\includegraphics[width=0.9\linewidth]{figs/Graph/sat-compact-Lcomp}}%
    \hfill
    \subcaptionbox
       {Substitution of \I-tromino gadgets.}%
       [.48\linewidth]%
       {\includegraphics[width=0.9\linewidth]{figs/Graph/sat-compact-Icomp}}%
    \caption{Reduction example from a periodic orthogonal 3SAT-3 graph drawing to a partial periodic tromino tiling.}%
    \label{fig:completionexample}
  \end{figure}
\end{proof}

\hide{%
  \begin{figure}
    \centering
    \begin{tikzpicture}
      \def\cliprange{5}
      \def\drawrange{7}
      \def\brickw{1.4142}  
      \def\brickh{0.7071}  
      \clip (-\cliprange,-\cliprange) rectangle (\cliprange,\cliprange);
      \begin{scope}[rotate=45]
        \foreach \x in {-7,...,7} {
          \draw[black] (\x,-\drawrange) -- (\x,\drawrange);
        }
        \foreach \y in {-7,...,7} {
          \draw[black] (-\drawrange,\y) -- (\drawrange,\y);
        }
      \end{scope}
      \foreach \i in {-10,...,10} {
        \pgfmathsetmacro{\y}{(\i + 0.5) * \brickh}
        \draw[red] (-10,\y) -- (10,\y);
        \ifodd\i
          \def\xshift{0.5}
        \else
          \def\xshift{0.0}
        \fi
        \foreach \j in {-10,...,10} {
          \pgfmathsetmacro{\x}{(\j + \xshift) * \brickw}
          \pgfmathsetmacro{\ytop}{(\i + 1 + 0.5) * \brickh}  
          \pgfmathsetmacro{\ybot}{(\i + 0.5) * \brickh}      
          \draw[red] (\x,\ybot) -- (\x,\ytop);
        }
      }
    \end{tikzpicture}
    \caption{A rotated (black) square lattice with a (red) brick pattern.}%
    \label{fig:brickpattern}
  \end{figure}}%

\begin{corollary}\label{cor:aperiodic-completion}
  There exists a partial periodic covering of the plane by \L trominoes (respectively \I trominoes) that can be completed to a full tiling of the plane by \L trominoes (respectively \I trominoes), but all such completions are aperiodic.
\end{corollary}

Our result starts from an infinite seed of preplaced tiles.
This is necessary: tiling completion from a finite seed of trominoes
is decidable.

\begin{theorem}\label{thm:finite-completion}
  For either the \L or \I tromino,
  it is NP-complete to decide whether a finite set of preplaced trominoes can be completed to a plane tiling with those trominoes.
\end{theorem}

\begin{proof}
  NP-hardness follows from our reductions in Theorem~\ref{thm:periodic-completion}, but starting from 3SAT instead of Periodic 3SAT. Then the construction fits in a rectangle, and the exterior is always tilable, so the problem reduces to whether the interior can be tiled, which is NP-hard.

  To prove membership in NP, suppose a tiling exists.
  Let $B$ be the bounding box of the preplaced trominoes.
  The witness is the set of additional tiles that intersect $B$.
  Together with the preplaced trominoes, verify that this tiles all of $B$ and that all tiles intersect~$B$.
  We claim this is a witness, i.e., it implies the existence of a plane tiling completion.

  First, for \L trominoes:
  \begin{enumerate}
  \item We can grow an edge of $B$ by $1$ by proceeding from one end to the other.
  \item If the next gap is $>1$, put an \L so that you fill two pixels, and one pixel in the next layer. Leave a gap that's 2 smaller.
  \item If the next gap of exactly $1$, put an \L so that you fill that one pixel, and two pixels in the next layer, pointing in the direction where you haven't filled anything yet so no collisions.
  \item Repeat on all sides to tile the plane.
  \end{enumerate}

  Second, for \I trominoes:
  \begin{enumerate}
  \item Extend each column by stacking vertical {\I}s
  \item Extend each row by stacking horizontal {\I}s
  \item Left with quarter-planes; pack those with horizontal {\I}s say.
     \qedhere
  \end{enumerate}
\end{proof}


\subsection{Domino Tiling and Completion}

The relationship between domino tilings and perfect matchings was previously established in \cite{beauquier1995tiling}.

\begin{theorem}\label{thm:domino-periodic}
  If a periodic polycube subset of $\mathbb R^d$ can be tiled by dominoes, then it can be tiled by dominoes periodically with period $1$.
\end{theorem}
\begin{proof}
  Consider the dual graph of the hypercubic lattice, and the subgraph $\PER{G}$ induced by the cells in the periodic subset to tile. The dual of the lattice, and thus $\PER{G}$, is bipartite. Any valid domino placement corresponds to an edge of $\PER{G}$; a tiling of the periodic subset corresponds to a perfect matching; and any perfect matching corresponds to a tiling of the periodic subset. By Theorem~\ref{thm:periodicmatching}, if $\PER{G}$ has a perfect matching, then it has one that is periodic with period 1.
\end{proof}

\begin{corollary}\label{cor:domino-completion-periodic}
  Any periodic partial tiling of $\mathbb R^d$ by dominoes that can be completed can be completed periodically with period $1$.
\end{corollary}

Applying Theorem~\ref{thm:matchingalgorithm}, we obtain the following.

\begin{corollary}\label{cor:domino-periodic-polytime}
  Tiling a periodic polycube subset of $d$D with dominoes can be decided in polynomial time, in any dimension $d$.
\end{corollary}

\section*{Acknowledgments}

This work grew out of two research groups: the MIT Hardness Group
and the MIT CompGeom Group.
We thank the other members of these groups --- in particular,
Josh Brunner,
Craig Kaplan,
Hayashi Layers,
Anna Lubiw,
Joseph O'Rourke,
Mikhail Rudoy,
and
Frederick Stock
--- for helpful discussions.
Most figures are drawn with SVG Tiler
[%
\protect\url{https://github.com/edemaine/svgtiler/}%
].

\bibliography{tiling,trominoes}%
\bibliographystyle{alpha}%

\end{document}

%% file: connectors.tex
\begin{tikzpicture}[scale=0.6, line cap=round]
  \begin{scope}[shift={(0,0)}]
    \draw (0,0) rectangle (2,5);
    \fill (1,5) circle (2pt);
    \fill (2,4) circle (2pt);
    \fill (2,3) circle (2pt);
    \fill (1,0) circle (2pt);
    \fill (1,2.5) circle (2pt);
    \draw[blue, line width=0.8pt] (1,5) -- (1,4.5) -- (0.5,4.5) -- (0.5,4.0) -- (0.5,3.5) -- (0.5,3.0) -- (0.5,2.5) -- (1.0,2.5);
    \draw[green, line width=0.8pt] (2,4) -- (1.5,4) -- (1.0,4) -- (1.0,3.5) -- (1.0,3.0) -- (1.0,2.5);
    \draw[red, line width=0.8pt] (2,3) -- (1.5,3) -- (1.5,2.5) -- (1.0,2.5);
    \draw[purple, line width=0.8pt] (1,0) -- (1,0.5) -- (1,1.0) -- (1,1.5) -- (1,2.0) -- (1,2.5);
  \end{scope}
  \begin{scope}[shift={(3,0)}]
    \draw (0,0) rectangle (2,5);
    \fill (2,4) circle (2pt);
    \fill (2,3) circle (2pt);
    \fill (2,2) circle (2pt);
    \fill (1,0) circle (2pt);
    \fill (1,2.5) circle (2pt);
    \draw[blue, line width=0.8pt] (2,4) -- (1.5,4) -- (1.0,4) -- (1.0,3.5) -- (1.0,3.0) -- (1.0,2.5);
    \draw[green, line width=0.8pt] (2,3) -- (1.5,3) -- (1.5,2.5) -- (1.0,2.5);
    \draw[red, line width=0.8pt] (2,2) -- (1.5,2) -- (1.0,2) -- (1.0,2.5);
    \draw[purple, line width=0.8pt] (1,0) -- (1,0.5) -- (1,1.0) -- (1,1.5) -- (0.5,1.5) -- (0.5,2.0) -- (0.5,2.5) -- (1.0,2.5);
  \end{scope}
  \begin{scope}[shift={(6,0)}]
    \draw (0,0) rectangle (2,5);
    \fill (1,5) circle (2pt);
    \fill (2,4) circle (2pt);
    \fill (2,3) circle (2pt);
    \fill (2,2) circle (2pt);
    \fill (1,2.5) circle (2pt);
    \draw[blue, line width=0.8pt] (1,5) -- (1,4.5) -- (0.5,4.5) -- (0.5,4.0) -- (0.5,3.5) -- (0.5,3.0) -- (0.5,2.5) -- (1.0,2.5);
    \draw[green, line width=0.8pt] (2,4) -- (1.5,4) -- (1.0,4) -- (1.0,3.5) -- (1.0,3.0) -- (1.0,2.5);
    \draw[red, line width=0.8pt] (2,3) -- (1.5,3) -- (1.5,2.5) -- (1.0,2.5);
    \draw[purple, line width=0.8pt] (2,2) -- (1.5,2) -- (1.0,2) -- (1.0,2.5);
  \end{scope}
  \begin{scope}[shift={(9,0)}]
    \draw (0,0) rectangle (2,5);
    \fill (2,4) circle (2pt);
    \fill (2,3) circle (2pt);
    \fill (2,2) circle (2pt);
    \fill (2,1) circle (2pt);
    \fill (1,2.5) circle (2pt);
    \draw[blue, line width=0.8pt] (2,4) -- (1.5,4) -- (1.0,4) -- (1.0,3.5) -- (1.0,3.0) -- (1.0,2.5);
    \draw[green, line width=0.8pt] (2,3) -- (1.5,3) -- (1.5,2.5) -- (1.0,2.5);
    \draw[red, line width=0.8pt] (2,2) -- (1.5,2) -- (1.0,2) -- (1.0,2.5);
    \draw[purple, line width=0.8pt] (2,1) -- (1.5,1) -- (1.0,1) -- (0.5,1) -- (0.5,1.5) -- (0.5,2.0) -- (0.5,2.5) -- (1.0,2.5);
  \end{scope}
  \begin{scope}[shift={(12,0)}]
    \draw (0,0) rectangle (2,5);
    \fill (1,5) circle (2pt);
    \fill (0,1) circle (2pt);
    \fill (2,4) circle (2pt);
    \fill (1,0) circle (2pt);
    \fill (1,2.5) circle (2pt);
    \draw[blue, line width=0.8pt] (1,5) -- (1,4.5) -- (1,4.0) -- (1,3.5) -- (1,3.0) -- (1,2.5);
    \draw[green, line width=0.8pt] (0,1) -- (0.5,1) -- (0.5,1.5) -- (0.5,2.0) -- (0.5,2.5) -- (1.0,2.5);
    \draw[red, line width=0.8pt] (2,4) -- (1.5,4) -- (1.5,3.5) -- (1.5,3.0) -- (1.5,2.5) -- (1.0,2.5);
    \draw[purple, line width=0.8pt] (1,0) -- (1,0.5) -- (1,1.0) -- (1,1.5) -- (1,2.0) -- (1,2.5);
  \end{scope}
  \begin{scope}[shift={(15,0)}]
    \draw (0,0) rectangle (2,5);
    \fill (0,1) circle (2pt);
    \fill (2,4) circle (2pt);
    \fill (2,3) circle (2pt);
    \fill (1,0) circle (2pt);
    \fill (1,2.5) circle (2pt);
    \draw[blue, line width=0.8pt] (0,1) -- (0.5,1) -- (1.0,1) -- (1.0,1.5) -- (1.0,2.0) -- (1.0,2.5);
    \draw[green, line width=0.8pt] (2,4) -- (1.5,4) -- (1.0,4) -- (0.5,4) -- (0.5,3.5) -- (0.5,3.0) -- (0.5,2.5) -- (1.0,2.5);
    \draw[red, line width=0.8pt] (2,3) -- (1.5,3) -- (1.0,3) -- (1.0,2.5);
    \draw[purple, line width=0.8pt] (1,0) -- (1,0.5) -- (1.5,0.5) -- (1.5,1.0) -- (1.5,1.5) -- (1.5,2.0) -- (1.5,2.5) -- (1.0,2.5);
  \end{scope}
  \begin{scope}[shift={(18,-0)}]
    \draw (0,0) rectangle (2,5);
    \fill (1,5) circle (2pt);
    \fill (0,1) circle (2pt);
    \fill (2,4) circle (2pt);
    \fill (2,3) circle (2pt);
    \fill (1,2.5) circle (2pt);
    \draw[blue, line width=0.8pt] (1,5) -- (1,4.5) -- (0.5,4.5) -- (0.5,4.0) -- (0.5,3.5) -- (0.5,3.0) -- (0.5,2.5) -- (1.0,2.5);
    \draw[green, line width=0.8pt] (0,1) -- (0.5,1) -- (1.0,1) -- (1.0,1.5) -- (1.0,2.0) -- (1.0,2.5);
    \draw[red, line width=0.8pt] (2,4) -- (1.5,4) -- (1.0,4) -- (1.0,3.5) -- (1.0,3.0) -- (1.0,2.5);
    \draw[purple, line width=0.8pt] (2,3) -- (1.5,3) -- (1.5,2.5) -- (1.0,2.5);
  \end{scope}
  \begin{scope}[shift={(21,0)}]
    \draw (0,0) rectangle (2,5);
    \fill (0,1) circle (2pt);
    \fill (2,4) circle (2pt);
    \fill (2,3) circle (2pt);
    \fill (2,2) circle (2pt);
    \fill (1,2.5) circle (2pt);
    \draw[blue, line width=0.8pt] (0,1) -- (0.5,1) -- (1.0,1) -- (1.0,1.5) -- (1.0,2.0) -- (1.0,2.5);
    \draw[green, line width=0.8pt] (2,4) -- (1.5,4) -- (1.0,4) -- (0.5,4) -- (0.5,3.5) -- (0.5,3.0) -- (0.5,2.5) -- (1.0,2.5);
    \draw[red, line width=0.8pt] (2,3) -- (1.5,3) -- (1.0,3) -- (1.0,2.5);
    \draw[purple, line width=0.8pt] (2,2) -- (1.5,2) -- (1.5,2.5) -- (1.0,2.5);
  \end{scope}
  \begin{scope}[shift={(0,-6)}]
    \draw (0,0) rectangle (2,5);
    \fill (1,5) circle (2pt);
    \fill (0,1) circle (2pt);
    \fill (0,2) circle (2pt);
    \fill (1,0) circle (2pt);
    \fill (1,2.5) circle (2pt);
    \draw[blue, line width=0.8pt] (1,5) -- (1,4.5) -- (1,4.0) -- (1,3.5) -- (1,3.0) -- (1,2.5);
    \draw[green, line width=0.8pt] (0,1) -- (0.5,1) -- (1.0,1) -- (1.0,1.5) -- (1.0,2.0) -- (1.0,2.5);
    \draw[red, line width=0.8pt] (0,2) -- (0.5,2) -- (0.5,2.5) -- (1.0,2.5);
    \draw[purple, line width=0.8pt] (1,0) -- (1,0.5) -- (1.5,0.5) -- (1.5,1.0) -- (1.5,1.5) -- (1.5,2.0) -- (1.5,2.5) -- (1.0,2.5);
  \end{scope}
  \begin{scope}[shift={(3,-6)}]
    \draw (0,0) rectangle (2,5);
    \fill (0,1) circle (2pt);
    \fill (0,2) circle (2pt);
    \fill (2,4) circle (2pt);
    \fill (1,0) circle (2pt);
    \fill (1,2.5) circle (2pt);
    \draw[blue, line width=0.8pt] (0,1) -- (0.5,1) -- (1.0,1) -- (1.0,1.5) -- (1.0,2.0) -- (1.0,2.5);
    \draw[green, line width=0.8pt] (0,2) -- (0.5,2) -- (0.5,2.5) -- (1.0,2.5);
    \draw[red, line width=0.8pt] (2,4) -- (1.5,4) -- (1.0,4) -- (1.0,3.5) -- (1.0,3.0) -- (1.0,2.5);
    \draw[purple, line width=0.8pt] (1,0) -- (1,0.5) -- (1.5,0.5) -- (1.5,1.0) -- (1.5,1.5) -- (1.5,2.0) -- (1.5,2.5) -- (1.0,2.5);
  \end{scope}
  \begin{scope}[shift={(6,-6)}]
    \draw (0,0) rectangle (2,5);
    \fill (1,5) circle (2pt);
    \fill (0,1) circle (2pt);
    \fill (0,2) circle (2pt);
    \fill (2,4) circle (2pt);
    \fill (1,2.5) circle (2pt);
    \draw[blue, line width=0.8pt] (1,5) -- (1,4.5) -- (1,4.0) -- (1,3.5) -- (1,3.0) -- (1,2.5);
    \draw[green, line width=0.8pt] (0,1) -- (0.5,1) -- (1.0,1) -- (1.0,1.5) -- (1.0,2.0) -- (1.0,2.5);
    \draw[red, line width=0.8pt] (0,2) -- (0.5,2) -- (0.5,2.5) -- (1.0,2.5);
    \draw[purple, line width=0.8pt] (2,4) -- (1.5,4) -- (1.5,3.5) -- (1.5,3.0) -- (1.5,2.5) -- (1.0,2.5);
  \end{scope}
  \begin{scope}[shift={(9,-6)}]
    \draw (0,0) rectangle (2,5);
    \fill (0,1) circle (2pt);
    \fill (0,2) circle (2pt);
    \fill (2,4) circle (2pt);
    \fill (2,3) circle (2pt);
    \fill (1,2.5) circle (2pt);
    \draw[blue, line width=0.8pt] (0,1) -- (0.5,1) -- (1.0,1) -- (1.0,1.5) -- (1.0,2.0) -- (1.0,2.5);
    \draw[green, line width=0.8pt] (0,2) -- (0.5,2) -- (0.5,2.5) -- (1.0,2.5);
    \draw[red, line width=0.8pt] (2,4) -- (1.5,4) -- (1.0,4) -- (1.0,3.5) -- (1.0,3.0) -- (1.0,2.5);
    \draw[purple, line width=0.8pt] (2,3) -- (1.5,3) -- (1.5,2.5) -- (1.0,2.5);
  \end{scope}
  \begin{scope}[shift={(12,-6)}]
    \draw (0,0) rectangle (2,5);
    \fill (0,1) circle (2pt);
    \fill (0,2) circle (2pt);
    \fill (0,3) circle (2pt);
    \fill (1,0) circle (2pt);
    \fill (1,2.5) circle (2pt);
    \draw[blue, line width=0.8pt] (0,1) -- (0.5,1) -- (1.0,1) -- (1.0,1.5) -- (1.0,2.0) -- (1.0,2.5);
    \draw[green, line width=0.8pt] (0,2) -- (0.5,2) -- (0.5,2.5) -- (1.0,2.5);
    \draw[red, line width=0.8pt] (0,3) -- (0.5,3) -- (1.0,3) -- (1.0,2.5);
    \draw[purple, line width=0.8pt] (1,0) -- (1,0.5) -- (1.5,0.5) -- (1.5,1.0) -- (1.5,1.5) -- (1.5,2.0) -- (1.5,2.5) -- (1.0,2.5);
  \end{scope}
  \begin{scope}[shift={(15,-6)}]
    \draw (0,0) rectangle (2,5);
    \fill (1,5) circle (2pt);
    \fill (0,1) circle (2pt);
    \fill (0,2) circle (2pt);
    \fill (0,3) circle (2pt);
    \fill (1,2.5) circle (2pt);
    \draw[blue, line width=0.8pt] (1,5) -- (1,4.5) -- (1,4.0) -- (1,3.5) -- (1,3.0) -- (1,2.5);
    \draw[green, line width=0.8pt] (0,1) -- (0.5,1) -- (1.0,1) -- (1.5,1) -- (1.5,1.5) -- (1.5,2.0) -- (1.5,2.5) -- (1.0,2.5);
    \draw[red, line width=0.8pt] (0,2) -- (0.5,2) -- (1.0,2) -- (1.0,2.5);
    \draw[purple, line width=0.8pt] (0,3) -- (0.5,3) -- (0.5,2.5) -- (1.0,2.5);
  \end{scope}
  \begin{scope}[shift={(18,-6)}]
    \draw (0,0) rectangle (2,5);
    \fill (0,1) circle (2pt);
    \fill (0,2) circle (2pt);
    \fill (0,3) circle (2pt);
    \fill (2,4) circle (2pt);
    \fill (1,2.5) circle (2pt);
    \draw[blue, line width=0.8pt] (0,1) -- (0.5,1) -- (1.0,1) -- (1.0,1.5) -- (1.0,2.0) -- (1.0,2.5);
    \draw[green, line width=0.8pt] (0,2) -- (0.5,2) -- (0.5,2.5) -- (1.0,2.5);
    \draw[red, line width=0.8pt] (0,3) -- (0.5,3) -- (1.0,3) -- (1.0,2.5);
    \draw[purple, line width=0.8pt] (2,4) -- (1.5,4) -- (1.5,3.5) -- (1.5,3.0) -- (1.5,2.5) -- (1.0,2.5);
  \end{scope}
  \begin{scope}[shift={(21,-6)}]
    \draw (0,0) rectangle (2,5);
    \fill (0,1) circle (2pt);
    \fill (0,2) circle (2pt);
    \fill (0,3) circle (2pt);
    \fill (0,4) circle (2pt);
    \fill (1,2.5) circle (2pt);
    \draw[blue, line width=0.8pt] (0,1) -- (0.5,1) -- (1.0,1) -- (1.0,1.5) -- (1.0,2.0) -- (1.0,2.5);
    \draw[green, line width=0.8pt] (0,2) -- (0.5,2) -- (0.5,2.5) -- (1.0,2.5);
    \draw[red, line width=0.8pt] (0,3) -- (0.5,3) -- (1.0,3) -- (1.0,2.5);
    \draw[purple, line width=0.8pt] (0,4) -- (0.5,4) -- (1.0,4) -- (1.5,4) -- (1.5,3.5) -- (1.5,3.0) -- (1.5,2.5) -- (1.0,2.5);
  \end{scope}
\end{tikzpicture}